\documentclass[a4paper, reqno]{amsart}

\usepackage[english]{babel}
\usepackage[latin1]{inputenc}
\usepackage[T1]{fontenc}

\usepackage{graphicx}
\usepackage{subcaption}
\usepackage{amsmath, amsthm, mathtools, amssymb, yhmath, mathabx, cancel}
\usepackage{amsfonts, dsfont, mathrsfs}
\usepackage{float}
\usepackage{enumitem}
\usepackage{nicefrac}
\usepackage{upgreek}

\usepackage{pgfplots}
\pgfplotsset{compat = newest}

\usepackage{hyperref} 

\DeclareRobustCommand{\SkipTocEntry}[5]{}

\DeclareMathOperator{\vect}{span}
\DeclareMathOperator{\Ker}{Ker}

\begin{document}

\newcommand{\R}{\mathbb{R}}
\newcommand{\Z}{\mathbb{Z}}
\newcommand{\C}{\mathbb{C}}
\newcommand{\N}{\mathbb{N}}
\newcommand\K{\mathbb{K}}
\newcommand{\dd}{\,\textrm{d}}
\newcommand{\wto}{\rightharpoonup}
\newcommand{\pscal}[1]{\ensuremath{\left\langle #1 \right\rangle}}
\newcommand{\pscalSM}[1]{\ensuremath{\langle #1 \rangle}}
\newcommand{\norm}[1]{\left\vert\kern-0.25ex\left\vert #1 \right\vert\kern-0.25ex\right\vert}
\newcommand{\normSM}[1]{\vert\kern-0.25ex\vert #1 \vert\kern-0.25ex\vert}
\newcommand{\normt}[1]{\left\vert\kern-0.25ex\left\vert\kern-0.25ex\left\vert #1 \right\vert\kern-0.25ex\right\vert\kern-0.25ex\right\vert}
\renewcommand{\phi}{\varphi}
\renewcommand{\epsilon}{\varepsilon}
\renewcommand{\tau}{\uptau} 

\numberwithin{equation}{section}

\newtheorem{thm}{Theorem}
\newtheorem{prop}[thm]{Proposition}
\newtheorem*{prop*}{Proposition}
\newtheorem{lemme}[thm]{Lemma}
\newtheorem{cor}[thm]{Corollary}
\newtheorem{mydef}[thm]{Definition}
\newtheorem{rmq}[thm]{Remark}
\newtheorem*{rmq*}{Remark}
\newtheorem{cjt}[thm]{Conjecture}

\newcommand{\clm}[1]{\hyperref[#1]{Lemma~\ref*{#1}}}
\newcommand{\cth}[1]{\hyperref[#1]{Theorem~\ref*{#1}}}
\newcommand{\cpr}[1]{\hyperref[#1]{Proposition~\ref*{#1}}}
\newcommand{\ccr}[1]{\hyperref[#1]{Corollary~\ref*{#1}}}
\newcommand{\cfg}[1]{\hyperref[#1]{Figure~\ref*{#1}}}
\newcommand{\cAn}[1]{\hyperref[#1]{Annexe~\ref*{#1}}}
\newcommand{\crm}[1]{\hyperref[#1]{Remark~\ref*{#1}}}
\newcommand{\ccjt}[1]{\hyperref[#1]{Conjecture~\ref*{#1}}}


\title[Symmetry breaking in the periodic TFDW model]{Symmetry breaking in the periodic Thomas--Fermi--Dirac--von~Weizsäcker model}

\author[J. Ricaud]{Julien RICAUD}
\address{CNRS, Université de Cergy-Pontoise, Département de Mathématiques, 95000 Cergy\nobreakdash-Pontoise, France}
\address{CNRS, Ceremade, Université Paris-Dauphine, PSL Research University, 75016 Paris, France}
\email{julien.ricaud@u-cergy.fr}

\thanks{The author is grateful to Mathieu Lewin for helpful discussions and advices, and to Enno Lenzmann for bringing our attention to the facts mentioned in the first remark after \ccjt{R3_eff_model_conjecture_uniqueness_and_monotony_M}. The author acknowledges financial support from the European Research Council under the European Community's Seventh Framework Program (FP7/2007-2013 Grant Agreement MNIQS 258023).}

\date{\today}

\begin{abstract}
We consider the Thomas--Fermi--Dirac--von~Weizsäcker model for a system composed of infinitely many nuclei placed on a periodic lattice and electrons with a periodic density. We prove that if the Dirac constant is small enough, the electrons have the same periodicity as the nuclei. On the other hand if the Dirac constant is large enough, the 2-periodic electronic minimizer is not 1-periodic, hence symmetry breaking occurs. We analyze in detail the behavior of the electrons when the Dirac constant tends to infinity and show that the electrons all concentrate around exactly one of the 8 nuclei of the unit cell of size 2, which is the explanation of the breaking of symmetry. Zooming at this point, the electronic density solves an effective nonlinear Schr\"odinger equation in the whole space with nonlinearity $u^{7/3}-u^{4/3}$. Our results rely on the analysis of this nonlinear equation, in particular on the uniqueness and non-degeneracy of positive solutions.
\end{abstract}

\maketitle

\tableofcontents

\section{Introduction}
Symmetry breaking is a fundamental question in Physics which is largely discussed in the literature. In this paper, we consider the particular case of electrons in a periodic arrangement of nuclei. We assume that we have classical nuclei located on a 3D periodic lattice and we ask whether the quantum electrons will have the symmetry of this lattice. We study this question for the Thomas--Fermi--Dirac--von~Weizsäcker (TFDW) model which is the most famous non-convex model occurring in Orbital-free Density Functional Theory. In short, the energy of this model takes the form
\begin{equation}\label{introdution_energy_formulae}
\int_\K{|\nabla \sqrt{\rho}|^2}+\frac35c_{TF}\int_\K{\rho^{\frac53}}-\frac34c\int_\K{\rho^{\frac43}}+\frac12\int_\K{(G\star\rho)\rho}-\int_\K{G \rho},
\end{equation}
where $\K$ is the unit cell, $\rho$ is the density of the electrons and $G$ is the periodic Coulomb potential. The non-convexity is (only) due to the term $-\frac34 c \int{\rho^{\frac43}}$. We refer to~\cite{GraSol-94, Friesecke-97, BokMau-99, BokGreMau-03, Seiringer-06} for a derivation of models of this type in various settings.

We study the question of symmetry breaking with respect to the parameter $c>0$. In this paper, we prove for $c>0$ that:
\begin{itemize}[label=$\bullet$]
	\item if $c$ is small enough, then the density $\rho$ of the electrons is unique and has the same periodicity as the nuclei, that is, there is no symmetry breaking;
	\item if $c$ is large enough, then there exist $2$-periodic arrangements of the electrons which have an energy that is lower than any $1$-periodic arrangement, that is, there is symmetry breaking.
\end{itemize}

Our method for proving the above two results is perturbative and does not provide any quantitative bound on the value of $c$ in the  two regimes. For small $c$ we perturb around $c=0$ and use the uniqueness and non degeneracy of the TFW minimizer, which comes from the strict convexity of the associated  functional. This is very similar in spirit to a result by Le Bris \cite{LeBris-93} in the whole space.

The main novelty of the paper, is the regime of large $c$. The $\rho^{\frac43}$ term in~\eqref{introdution_energy_formulae} favours concentration and we will prove that the electronic density concentrates at some points in the unit cell $\K$ in the limit $c\to\infty$ (it converges weakly to a sum of Dirac deltas). Zooming around one point of concentration at the scale $1/c$ we get a simple effective model posed on the whole space $\R^3$ where all the Coulomb terms have disappeared. The effective minimization problem is of NLS-type with two subcritical power nonlinearities:
\begin{equation}\label{introdution_R3_eff_model_minimization}
J_{\R^3}(\lambda)=\inf\limits_{\substack{v\in H^1(\R^3) \\ \norm{v}^2_{L^2(\R^3)}=\lambda}} \left\{\int_{\R^3}{|\nabla v|^2}+\frac35c_{TF}\int_{\R^3}{|v|^{\frac{10}3}}-\frac34\int_{\R^3}{|v|^{\frac83}} \right\}.
\end{equation}
The main argument is that it is favourable to put all the mass of the unit cell at one concentration point, due to the strict binding inequality 
$$J_{\R^3}(\lambda)< J_{\R^3}(\lambda')+J_{\R^3}(\lambda-\lambda')$$
that we prove in Section~\ref{section_R3_eff_model_existence}. Hence for the $2$-periodic problem, when $c$ is very large the $8$ electrons of the double unit cell prefer to concentrate at only one point of mass $8$, instead of $8$ points of mass $1$. This is the origin of the symmetry breaking for large $c$. Of course the exact same argument works for a union of $n^3$ unit cells.

Let us remark that the uniqueness of minimizers for the effective model $J_{\R^3}(\lambda)$ in~\eqref{introdution_R3_eff_model_minimization} is an open problem that we discuss in Section~\ref{Section_effective_model_R3}. We can however prove that any nonnegative solution of the corresponding nonlinear equation
$$-\Delta Q_\mu+c_{TF}{Q_\mu}^{\frac73}-{Q_\mu}^{\frac53}=-\mu Q_\mu$$
is unique and nondegenerate (up to translations). We conjecture (but are unable to prove) that the mass $\int {Q_\mu}^2$ is an increasing function of $\mu$. This would imply uniqueness of minimizers and is strongly supported by numerical simulations. Under this conjecture we can prove that there are exactly $8$ minimizers for $c$ large enough, which are obtained one from each other by applying $1$-translations.

The TFDW model studied in this paper is a very simple spinless empirical theory which approximates the true many-particle Schrödinger problem. The term $-\frac34 c \int\limits{\rho^{\frac43}}$ is an approximation to the~\emph{exchange-correlation energy} and $c$ is only determined on empirical grounds. The exchange part was computed by Dirac \cite{Dirac-30b} in 1930 using an infinite non-interacting Fermi gas leading to the value $c_D:=\sqrt[3]{6q^{-1}\pi^{-1}}$, where $q$ is the number of spin states. For the spinless model (i.e. $q=1$) that we are studying, this gives $\frac34c_D\approx0.93$, which is the constant generally appearing in the literature. It is natural to use a constant $c>c_D$ in order to account for correlation effects. On the other hand, the famous Lieb-Oxford inequality~\cite{Lieb-79, LieOxf-80, ChaHan-99, LieSei-10} suggests to take $\frac34 c_D\leq 1.64$. It has been argued in~\cite{Perdew-91,PerWan-92,LevPer-93} that for the classical interacting uniform electron gas one should use the value $\frac34 c\approx 1.44$ which is the energy of Jellium in the body-centered cubic (BCC) Wigner crystal and is implemented in the most famous Kohn-Sham functionals~\cite{PerBurErn-96}. However, this has recently been questioned in~\cite{LewLie-14} by Lewin and Lieb. In any case, all physically reasonable choices lead to $\frac34c$ of the order of $1$.

We have run some numerical simulations presented later in Section~\ref{num_sim}, using nuclei of (pseudo) charge $Z=1$ on a BCC lattice of side-length $4$\AA. We found that symmetry breaking occurs at about $\frac34 c\approx 2.48$. More precisely, the $2$-periodic ground state was found to be $1$-periodic if $\frac34c\lesssim2.474$ but really $2$-periodic for $\frac34c\gtrsim2.482$. The numerical value $\frac34 c\approx 2.48$ obtained as critical constant in our numerical simulations is above the usual values chosen in the literature. However, it is of the same order of magnitude and this critical constant could be closer to $1$ for other periodic arrangements of nuclei.

There exist various works on the TFDW model for $N$ electrons on the whole space $\R^3$. For example, Le Bris proved in~\cite{LeBris-93} that there exists $\epsilon>0$ such that minimizers exist for $N<Z+\epsilon$, improving the result for $N\leq Z$ by Lions~\cite{Lions-87}. It is also proved in~\cite{LeBris-93} that minimizers are unique for $c$ small enough if $N\leq Z$. Non existence if $N$ is large enough and $Z$ small enough has been proved by Nam and Van~Den~Bosch in~\cite{NamVDB-17}.

On the other hand, symmetry breaking has been studied in many situations. For discrete models on lattices, the instability of solutions having the same periodicity as the lattice is proved in~\cite{Frohlich-54b, Peierls} while~\cite{KenLie-86, Lieb-86, KenLie-87, LieNac-95, LieNac-95b, LieNac-95c, FraLie-11, GarSer-12} prove for different models (and different dimensions) that the solutions have a different periodicity than the lattice. On finite domains and at zero temperature, symmetry breaking is proved in~\cite{ProNor-01} for a one dimensional gas on a circle of finite length and in~\cite{Prodan-05} on toruses and spheres in dimension $d\leq3$. On the whole space $\R^3$, symmetry breaking is proved in \cite{BelGhi-16}, namely, the minimizers are not radial for $N$ large enough.

The paper is organized as follows. We present our main results for the periodic TFDW model and for the effective model, together with our numerical simulations, in Section~\ref{section_main_results}. In Section~\ref{section_R3_eff_model}, we study the effective model $J_{\R^3}(\lambda)$ on the whole space. Then, in Section~\ref{section_small_c}, we prove our results for the regime of small $c$. Finally, we prove the symmetry breaking in the regime of large $c$ in Section~\ref{section_large_c}.

\section{Main results}\label{section_main_results}
For simplicity, we restrict ourselves to the case of a cubic lattice with one atom of charge $Z=1$ at the center of each unit cell. We denote by ${\mathscr L}_\K$ our lattice which is based on the natural basis and its unit cell is the cube $\K:=\left[-\frac12;\frac12\right)^3$, which contains one atom of charge $Z=1$ at the position $R=0$. The Thomas--Fermi--Dirac--von~Weizsäcker model we are studying is then the functional energy
\begin{equation}\label{K_NRJ_u}
\mathscr E_{\K,c}(w)=\int_\K{|\nabla w|^2}+\frac35c_{TF}\int_\K{|w|^{\frac{10}3}}-\frac34c\int_\K{|w|^{\frac83}}+\frac12D_\K(|w|^2,|w|^2)-\int_\K{G_\K |w|^2},
\end{equation}
on the unit cell $\K$. Here
$$D_\K(f,g)=\int_\K{\int_\K{f(x)G_\K(x-y)g(y)\dd y}\dd x},$$
where $G_\K$ is the $\K$-periodic Coulomb potential which satisfies
\begin{equation}\label{Laplacian_G_K_equation}
-\Delta G_\K = 4\pi\left(\sum\limits_{k\in{\mathscr L}_\K} \delta_k -1\right)
\end{equation}
and is uniquely defined up to a constant that we fix by imposing $\min\limits_{x\in \K} G_\K(x) = 0$.

We are interested in the behavior when $c$ varies of the minimization problem
\begin{equation}\label{K_complete_model_minimization}
E_{\K,\lambda}(c)=\inf\limits_{\substack{w\in H^1_{\textrm{per}}(\K) \\ \norm{w}_{L^2(\K)}^2=\lambda}}\mathscr E_{\K,c}(w),
\end{equation}
where the subscript \emph{per} stands for $\K$-periodic boundary conditions. We want to emphasize that even if the true $\K$-periodic TFDW model requires that $\lambda=Z$ (see~\cite{CatBriLio-98b}), we study it for any $\lambda$ in this paper.

Finally, for any $N\in\N\setminus\{0\}$, we denote by $N\cdot\K$ the union of $N^3$ cubes $\K$ forming the cube $N\cdot\K=\left[-\frac{N}2;\frac{N}2\right)^3$. The $N^3$ charges are then located at the positions
$$\{R_j\}_{1\leq j\leq N^3}\subset\left\{\left.\left(n_1-\frac{N+1}2,n_2-\frac{N+1}2,n_3-\frac{N+1}2\right) \right| n_i\in \N\cap [1;N]\right\}.$$

\subsection{Symmetry breaking}\label{section_symmetry_breaking}
The main results presented in this paper are the two following theorems.

\begin{thm}[Uniqueness for small $c$]\label{main_result_2}
Let $\K$ be the unit cube and $c_{TF},\lambda$ be two positive constants. There exists $\delta>0$ such that for any $0\leq c<\delta$, the following holds:
\begin{enumerate}[label=\roman*.,leftmargin=1.4em]
	\item The minimizer $w_c$ of the periodic TFDW problem $E_{\K,\lambda}(c)$ in~\eqref{K_complete_model_minimization} is \emph{unique}, up to a phase factor. It is non constant, positive, in $H^2_{\textrm{per}}({\K})$ and the unique ground-state eigenfunction of the $\K$-periodic self-adjoint operator
$$H_{c}:=-\Delta +c_{TF}|w_c|^{\frac43}-c|w_c|^{\frac23}-G_\K +(|w_c|^2\star G_\K).$$
	\item The $N\K$-periodic  extension of the $\K$-periodic minimizer $w_c$ is the \emph{unique minimizer} of all the $N\K$-periodic TFDW problems $E_{N\cdot\K,N^3\lambda}(c)$, for any integer $N\geq1$. Moreover
	$$E_{N\cdot\K,N^3\lambda}(c)=N^3 E_{\K,\lambda}(c).$$
\end{enumerate}
\end{thm}

\begin{thm}[Asymptotics for large $c$]\label{main_result}
Let $\K$ be the unit cube, $c_{TF},\lambda$ be two positive constants, and $N\geq1$ be an integer. For $c$ large enough, the periodic TFDW problem $E_{N\cdot\K,N^3\lambda}(c)$ on $N\cdot\K$ admits (at least) $N^3$ distinct nonnegative minimizers which are obtained one from each other by applying translations of the lattice ${\mathscr L}_\K$. Denoting $w_c$ any one of these minimizers, there exists a subsequence $c_n\to\infty$ such that
\begin{equation}\label{main_result_convergence_minimizers_with_Rn}
{c_n}^{-\frac32}w_{c_n}\Big(R+\frac\cdot{c_n}\Big)\underset{n\to\infty}{\longrightarrow}Q,
\end{equation}
strongly in $L^p_{\textrm{loc}}(\R^3)$ for $2\leq p<+\infty$, with $R$ the position of one of the $N^3$ charges in $N\cdot\K$. Here $Q$ is a minimizer of the variational problem for the effective model
\begin{equation}\label{R3_eff_model_J_def}
J_{\R^3}(N^3\lambda)=\inf\limits_{\substack{u\in H^1(\R^3) \\ \norm{u}^2_{L^2(\R^3)}=N^3\lambda}} \left\{\int_{\R^3}{|\nabla u|^2}+\frac35c_{TF}\int_{\R^3}{|u|^{\frac{10}3}}-\frac34\int_{\R^3}{|u|^{\frac83}} \right\},
\end{equation}
which must additionally minimize
\begin{equation}\label{R3_eff_model_def_functional_coeff_of_c}
S(N^3\lambda)=\inf_v\left\{\frac12\int_{\R^3}{\int_{\R^3}{\frac{|v(x)|^2|v(y)|^2}{|x-y|}\dd y}\dd x} - \int_{\R^3} \frac{|v(x)|^2}{|x|} \dd x\right\},
\end{equation}
where the minimization is performed among all possible minimizers of~\eqref{R3_eff_model_J_def}. Finally, when $c\to\infty$, $E_{N\cdot\K,N^3\lambda}(c)$ has the expansion
\begin{equation}\label{K_complete_model_expansion_E}
E_{N\cdot\K,N^3\lambda}(c)=c^2 J_{\R^3}(N^3\lambda)+c S(N^3\lambda)+o(c).
\end{equation}
\end{thm}

\cth{main_result_2} will be proved in Section~\ref{section_small_c} while Section~\ref{section_large_c} will be dedicated to the proof of~\cth{main_result}. The leading order in~\eqref{K_complete_model_expansion_E} 
$$E_{N\cdot\K,N^3\lambda}(c)=c^2 J_{\R^3}(N^3\lambda)+o(c^2)$$
together with the strict binding inequality $J_{\R^3}(N^3\lambda)<N^3J_{\R^3}(\lambda)$ for $N\geq2$, proved later in~\cpr{R3_eff_model_strict_binding} of Section~\ref{section_R3_eff_model}, imply immediately that symmetry breaking occurs.
\begin{cor}[Symmetry breaking for large $c$]\label{main_resultsqdqdqsdqdsqs}
Let $\K$ be the unit cube, $c_{TF},\lambda$ be two positive constants, and $N\geq2$ be an integer. For $c$ large enough, symmetry breaking occurs:
	$$E_{N\cdot\K,N^3\lambda}(c)<N^3 E_{\K,\lambda}(c).$$
\end{cor}
Although the leading order is sufficient to prove the occurrence of symmetry breaking, \cth{main_result} gives a precise description of the behavior of the electrons, which all concentrate at one of the $N^3$ nuclei of the cell $N\cdot\K$. A natural question that comes with~\cth{main_result} is to know if $c$ needs to be really large for the symmetry breaking to happen. We present some numerical answers to this question later in Section~\ref{num_sim}.

\begin{rmq*}[Generalizations]
For simplicity we have chosen to deal with a cubic lattice with one nucleus of charge $1$ per unit cell, but the exact same results hold in a more general situation. We can take a charge $Z$ larger than $1$, several charges (of different values) per unit cell and a more general lattice than $\Z^3$. More precisely, the $\K$-periodic Coulomb potential $G_\K$ appearing in~\eqref{K_NRJ_u}, in both $D_\K$ and $\int G|w|^2$, should then verify
$$-\Delta G_\K = 4\pi\left(\sum\limits_{k\in{\mathscr L}_\K} \delta_k -\frac1{|\K|}\right),$$
and the term $\int_\K{G_\K |w|^2}$ should be replaced by
$\int_\K{\sum_{i=1}^{N_q} z_i G_\K(\cdot-R_i)|w|^2}$
where $z_i$ and $R_i$ and the charges and locations of the $N_q$ nuclei in the unit cell $\K$.

Finally, in~\cth{main_result}, denoting by $z_+:=\max_{1\leq i\leq N_q}\{z_i\}>0$ the largest charge inside $\K$ and by $N_+\geq1$ the number of charges inside $\K$ that are equal to $z_+$, the location $R$ would now be one of the $N_+\K^3$ positions of charges $z_+$ --- which means that the minimizer concentrate on one of the nuclei with largest charge --- and $S$ would be replaced by
$$S(\lambda)=\inf_v\left\{\frac12\int_{\R^3}{\int_{\R^3}{\frac{|v(x)|^2|v(y)|^2}{|x-y|}\dd y}\dd x} - z_+\int_{\R^3} \frac{|v(x)|^2}{|x|} \dd x\right\}.$$
\end{rmq*}

\begin{rmq*}[Model on $\R^3$]
In this paper, we study the TFDW model for a periodic system, because such orbital-free theories are often used in practice for infinite systems. However, \cth{main_result} can be adapted to the TFDW model in the whole space $\R^3$, with finitely many nuclei of charges $z_1,\dots,z_n$ and $\lambda\leq\sum_i{z_i}$ electrons, using similar proofs. In the limit $c\to\infty$, the $\lambda$ electrons all concentrate at one of the nuclei with the largest charge $z_+:=\max\{z_i\}$ and solve the same effective problem. Therefore, uniqueness does not hold if there are several such nuclei of charge $z_+$.
\end{rmq*}

\subsection{Study of the effective model in \texorpdfstring{$\R^3$}{R3}}\label{Section_effective_model_R3}
We present in this section the effective model in the whole space $\R^3$. We want to already emphasize that the uniqueness of minimizers for this problem is an open difficult question as we will explain later in this section.

The functional to be considered is
\begin{equation}\label{R3_eff_model_functional}
u\mapsto{\mathscr J}_{\R^3}(u)=\int_{\R^3}{|\nabla u|^2}+\frac35c_{TF}\int_{\R^3}{|u|^{\frac{10}3}}-\frac34 \int_{\R^3}{|u|^{\frac83}}
\end{equation}
and the minimization problem~\eqref{R3_eff_model_J_def} is
\begin{equation}\label{R3_eff_model_minimization}
J_{\R^3}(\lambda)=\inf\limits_{\substack{u\in H^1(\R^3) \\ \norm{u}^2_{L^2(\R^3)}=\lambda}} {\mathscr J}_{\R^3}(u).
\end{equation}

The first important result for this effective model is about the existence of minimizers and the fact that they are radial decreasing. We state those results in the following theorem, the proof of which is the subject of Section~\ref{section_R3_eff_model_existence}.
\begin{thm}[Existence of minimizers for the effective model in $\R^3$]\label{R3_eff_model_existence_thm}
Let $c_{TF}>0$ and $\lambda>0$ be fixed constants.
\begin{enumerate}[label=\roman*.,leftmargin=1.4em]
	\item There exist minimizers for $J_{\R^3}(\lambda)$. Up to a phase factor and a space translation, any minimizer $Q$ is a positive radial strictly decreasing $H^2(\R^3)$-solution of
\begin{equation}\label{R3_eff_model_EulerLagrange}
	-\Delta Q+c_{TF}|Q|^{\frac43}Q-|Q|^{\frac23}Q=-\mu Q.
\end{equation}
Here $-\mu<0$ is simple and is the smallest eigenvalue of the self-adjoint operator $H_{Q}:=-\Delta +c_{TF}|Q|^{\frac43}-|Q|^{\frac23}$.
	\item We have the strictly binding inequality
\begin{equation}\label{R3_eff_model_strict_binding_ineq}
\forall \; 0<\lambda'<\lambda, \qquad J_{\R^3}(\lambda)< J_{\R^3}(\lambda')+J_{\R^3}(\lambda-\lambda').
\end{equation}
	\item For any minimizing sequence $(Q_n)_n$ of $J_{\R^3}(\lambda)$, there exists $\{x_n\}\subset\R^3$ such that $Q_n(\cdot-x_n)$ strongly converges in $H^1(\R^3)$ to a minimizer, up to the extraction of a subsequence.
\end{enumerate}
\end{thm}

An important result about the effective model on $\R^3$ is the following result giving the uniqueness and the non-degeneracy of positive solutions $Q$ to the Euler--Lagrange equation~\eqref{R3_eff_model_EulerLagrange} for any admissible $\mu>0$. The proof of this theorem is the subject of Section~\ref{section_R3_eff_model_existence_and_nondeg}.

\begin{thm}[Uniqueness and non-degeneracy of positive solutions]\label{R3_eff_model_existence_and_nondeg}
Let $c_{TF}>0$. If $\frac{64}{15}c_{TF}\mu\geq1$, then the Euler--Lagrange equation~\eqref{R3_eff_model_EulerLagrange} has no non-trivial solution in $H^1(\R^3)$. For $0<\frac{64}{15}c_{TF}\mu<1$, the Euler--Lagrange equation~\eqref{R3_eff_model_EulerLagrange} has, up to translations, a unique nonnegative solution $Q_\mu\nequiv0$ in $H^1(\R^3)$. This solution is radial decreasing and non-degenerate: the linearized operator
\begin{equation}\label{R3_eff_model_def_Lplus}
L_\mu^+=-\Delta+\frac73 c_{TF} |Q_\mu|^{\frac43}-\frac53 |Q_\mu|^{\frac23}+\mu
\end{equation}
with domain $H^2(\R^3)$ and acting on $L^2(\R^3)$ has the kernel
\begin{equation}
\Ker L_\mu^+ =\vect\left\{\partial_{x_1}Q_\mu, \partial_{x_2}Q_\mu, \partial_{x_3}Q_\mu\right\}.
\end{equation}
\end{thm}
Note that the condition $\frac{64}{15}c_{TF}\mu\geq1$ comes directly from Pohozaev's identity, see e.g.~\cite{BerLio-83}.
\begin{rmq*}
The linearized operator $L_\mu$ for the equation~\eqref{R3_eff_model_EulerLagrange} at $Q_\mu$ is
$$L_\mu h=-\Delta h+\left(c_{TF}|Q_\mu|^{\frac43}-|Q_\mu|^{\frac23}\right)h+\left(\frac23c_{TF}|Q_\mu|^{\frac43}-\frac13|Q_\mu|^{\frac23}\right)(h+\bar{h})+\mu h.$$
Note that it is not $\C$-linear. Separating its real and imaginary parts, it is convenient to rewrite it as
$$L_\mu=\begin{pmatrix}L_\mu^+&0\\0&L_\mu^-\end{pmatrix},$$
where $L_\mu^+$ is as in~\eqref{R3_eff_model_def_Lplus} while $L_\mu^-$ is the operator
\begin{equation}\label{R3_eff_model_def_Lmoins}
L_\mu^-=-\Delta+c_{TF} |Q_\mu|^{\frac43}- |Q_\mu|^{\frac23}+\mu=H_{Q_\mu}+\mu.
\end{equation}
The result about the lowest eigenvalue of the operator $H_Q$ in~\cth{R3_eff_model_existence_thm} exactly gives that $\Ker L_\mu^-=\vect\left\{Q_\mu\right\}$. Hence,~\cth{R3_eff_model_existence_and_nondeg} implies that
$$\Ker L_\mu=\vect\left\{\begin{pmatrix}0\\Q_\mu\end{pmatrix},\begin{pmatrix}\partial_{x_1}Q_\mu\\0\end{pmatrix}, \begin{pmatrix}\partial_{x_2}Q_\mu\\0\end{pmatrix}, \begin{pmatrix}\partial_{x_3}Q_\mu\\0\end{pmatrix}\right\}.$$
\end{rmq*}

The natural step one would like to perform now is to deduce the uniqueness of minimizers from the uniqueness of Euler--Lagrange positive solutions as it has been done for many equations \cite{Lieb-77,TodMor-99,Lenzmann-09,FraLen-13,FraLenSil-15,Ricaud-16}. An argument of this type relies on the fact that $\mu\mapsto M(\mu):=\norm{Q_\mu}^2_{L^2(\R^3)}$ is a bijection, which is an easy result for models with trivial scalings like the nonlinear Schr\"odinger equation with only one power nonlineartity. However, for the effective problem of this section, we are unable to prove that this mapping is a bijection, proving the injection property being the issue.

In~\cite{KilOhPocVis-17}, Killip, Oh, Pocovnicu and Visan study extensively a similar problem with another non-linearity including two powers, namely the \emph{cubic-quintic NLS on $\R^3$} which is associated with the energy
\begin{equation}\label{cubic_quintic_NLS_energy}
\int_{\R^3}{\frac12|\nabla u|^2+\frac16|u|^6-\frac14 |u|^4}.
\end{equation}
They discussed at length the question of uniqueness of minimizers and could also not solve it for their model. An important difference between~\eqref{cubic_quintic_NLS_energy} and effective problem of this section is that the map $\mu\mapsto M(\mu)$ is for sure not a bijection in their case. But it is conjectured to be one if one only retains \emph{stable} solutions \cite[Conjecture 2.6]{KilOhPocVis-17}.

If we cannot prove uniqueness of minimizers, we can nevertheless prove that for any mass $\lambda>0$ there is a finite number of $\mu$'s in $(0; \frac{15}{64 c_{TF}})$ for which the unique positive solution to the associated Euler--Lagrange problem has a mass equal to $\lambda$ and, consequently, that there is a finite number of minimizers of the TFDW problem for any given mass constraint.
\begin{prop}\label{R3_eff_model_finite_number_sol_EL}
Let $\lambda>0$ and $c_{TF}>0$. There exist finitely many $\mu$'s for which the mass $M(\mu)$ of $Q_\mu$ is equal to $\lambda$.
\end{prop}
\begin{proof}[Proof of~\cpr{R3_eff_model_finite_number_sol_EL}]
By~\cth{R3_eff_model_existence_thm}, we know that for any mass constraint $\lambda\in(0,+\infty)$, there exist at least one minimizer to the corresponding $J_{\R^3}(\lambda)$ minimization problem. Therefore, for any $\lambda\in(0,+\infty)$, there exists at least one $\mu$ such that the unique positive solution $Q_\mu$ to the associated Euler--Lagrange equation is a minimizer of $J_{\R^3}(\lambda)$ and thus is of mass $M(\mu)=\lambda$. We therefore obtain that $\left(0;\frac{15}{64 c_{TF}}\right)\ni\mu\mapsto M(\mu)\in(0;+\infty)$ is onto. Moreover, this map is real-analytic since the non-degeneracy in \cth{R3_eff_model_existence_and_nondeg} and the analytic implicit function theorem give that $\mu\mapsto Q_\mu$ is real analytic. The map $M$ being onto and real-analytic, then for any $\lambda\in(0;+\infty)$ there exists a finite number of $\mu$'s, which are all in $\left(0;\frac{15}{64 c_{TF}}\right)$, such that the mass $M(\mu)$ of the unique positive solution $Q_\mu$ is equal to $\lambda$.
\end{proof}

We have performed some numerical computations of the solution $Q_\mu$ and the results strongly support the uniqueness of minimizers since $M$ was found to be increasing, see Figure~\ref{figure_numerics_M_strict_croissante}.
\begin{figure}[h]
\centering
    \includegraphics[width=0.7\columnwidth,draft=false]{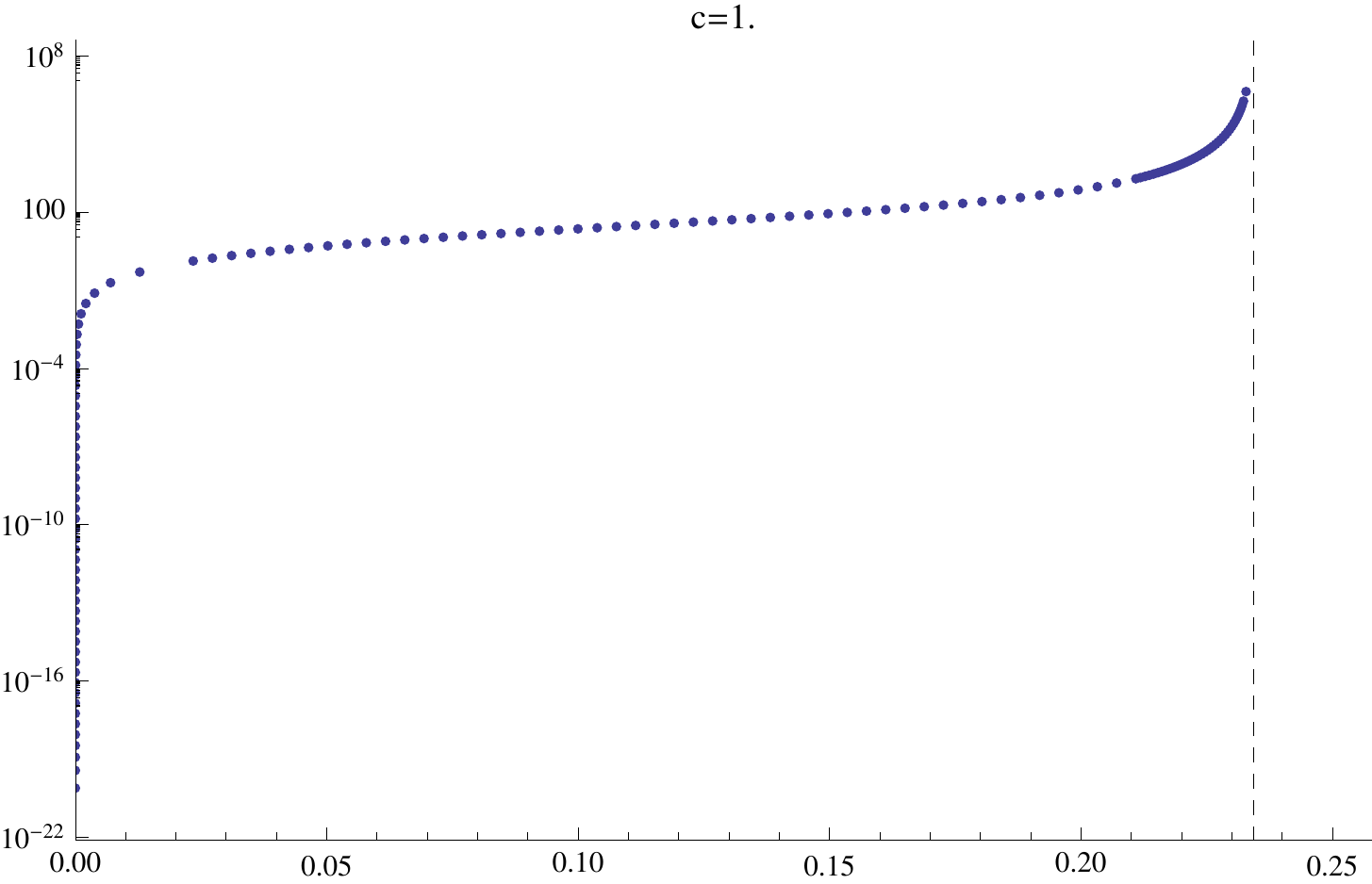}
\caption{Plot of $\mu\mapsto\ln\left(M(\mu)\right)$ on $\left(0;\frac{15}{64 c_{TF}}\right)$.}
\label{figure_numerics_M_strict_croissante}
\end{figure}

\begin{cjt}\label{R3_eff_model_conjecture_uniqueness_and_monotony_M}
The function
\begin{equation}
\begin{aligned}
\left(0;\frac{15}{64 c_{TF}}\right)&\to(0;+\infty)\\
\mu&\mapsto M(\mu)
\end{aligned}
\end{equation}
is strictly increasing and one-to-one. Consequently, for any $0< \mu<\frac{15}{64 c_{TF}}$, there exists a unique minimizer $Q_\mu$ of $J_{\R^3}(\lambda)$, up to a phase and a space translation.
\end{cjt}
\begin{rmq*}
It should be possible to show that the energy $\mu\mapsto {\mathscr J}_{\R^3}(Q_\mu)$ is strictly decreasing close to $\mu=0$ and $\mu=\mu_*$, and real-analytic on $(0,\mu_*)$. Using the concavity of $\lambda\mapsto J_{\R^3}(\lambda)$ (see \clm{R3_eff_model_apriori_properties_J}) one should be able to prove that the function $\lambda\mapsto \mu(\lambda)$ is increasing and continuous, except at a countable set of points where it can jump. From the analyticity there must be a finite number of jumps and we conclude that $\lambda\mapsto J_{\R^3}(\lambda)$ has a unique minimizer for all $\lambda$ except at these finitely many points.
\end{rmq*}
\begin{rmq*}
Following the method of \cite{KilOhPocVis-17}, one can prove there exist $C,C'>0$ such that $M(\mu)=C\mu^{\frac32}+o(\mu^{\frac32})_{\mu\to0^+}$ and $M(\mu)=C'(\mu-\mu_*)^{-3}+o\left((\mu-\mu_*)^{-3}\right)_{\mu\to\mu_*^-}$ where $\mu_*=\frac{15}{64 c_{TF}}$.
\end{rmq*}

This conjecture on $M$ is related to the stability condition on $({L_\mu^+})^{-1}$ that appears in works like~\cite{Weinstein-85, GriShaStra-87}. Indeed, differentiating the Euler--Lagrange equation~\eqref{R3_eff_model_EulerLagrange} with respect to $\mu$, we obtain that $L_\mu^+(\frac{\dd Q_\mu}{\dd \mu})=-Q_\mu$ which thus leads to
$$\frac{\dd}{\dd \mu}\int{Q_\mu}^2=2\pscal{Q_\mu,\frac{\dd Q_\mu}{\dd \mu}}=-2\pscal{Q_\mu,\left(L_\mu^+\right)^{-1} Q_\mu}.$$
Thus our conjecture is that $\ensuremath{\langle Q_\mu,\left(L_\mu^+\right)^{-1} Q_\mu \rangle}<0$ for all $0<\mu<\frac{15}{64 c_{TF}}$ and this corresponds to the fact that all the solutions are local strict minimizers.

\begin{thm}\label{R3_eff_model_consequence_conjecture_uniqueness_and_monotony_M}
If~\ccjt{R3_eff_model_conjecture_uniqueness_and_monotony_M} holds then, for $c$ large enough, there are exactly $N^3$ nonnegative minimizers for the periodic TFDW problem $E_{N\cdot\K,N^3\lambda}(c)$.
\end{thm}
The proof of~\cth{R3_eff_model_consequence_conjecture_uniqueness_and_monotony_M} is the subject of Section~\ref{section_number_minimizers_under_conjecture}.

\subsection{Numerical simulations}\label{num_sim}
The occurrence of symmetry breaking is an important question in practical calculations. Concerning the general behavior of DFT on this matter, we refer to the discussion in~\cite{SheLeeHea-99} and the references therein.

Our numerical simulations have been run with a constant $c_W=0.186$ in front of the gradient term (see \cite{Lieb-81b} for the choice of this value) and using the software \emph{PROFESS v.3.0} \cite{PROFESS3} which is based on pseudo-potentials (see~\crm{rmq_pseudo} below): we have used a (BCC) Lithium crystal of side-length $4$\AA{} (in order to be physically relevant as the two first alkali metals Lithium and Sodium organize themselves on BCC lattices with respective side length $3.51$\AA{} and $4.29$\AA{}) for which one electron is treated while the two others are included in the pseudo-potential, simulating therefore a lattice of pseudo-atoms with pseudo-charge $Z=\lambda=1$. The relative gain of energy of $2$-periodic minimizers compared to $1$-periodic ones is plotted in Figure~\ref{figure_symm_breaking_joined}. Symmetry breaking occurs at about $\frac34 c\approx 2.48$.
\begin{figure}[h]
\centering
\subcaptionbox{$0\leq \frac34c \leq 4$\label{figure_symm_breaking_0_5}}
{%
\begin{tikzpicture}
\begin{axis}[
scale only axis, 
height=2.5cm,
width=\textwidth/2-width("$?0.5\%$")-0.2cm, 
xmin = -0.1,
xmax = 4.1, 
grid=major,
major grid style={dashed,gray!50},
yticklabel=\pgfmathparse{\tick}\pgfmathprintnumber{\pgfmathresult}\,\%,
yticklabel style = {font=\scriptsize,xshift=0.5ex}]
\addplot[
color = black,
fill = black,
mark = *,
only marks] coordinates {
(0,0)
(0.7386,0)
(1.4771,0)
(2.2157,0)
(2.4372,0)
(2.4520,0)
(2.4668,0)
(2.4742,0)
(2.4816,-0.00725)
(2.4889,-0.02072)
(2.4963,-0.03483)
(2.5111,-0.06487)
(2.5850,-0.25329)
(2.6588,-0.51335)
(2.7327,-0.81314)
(2.8065,-1.18702)
(2.9542,-2.05948)
(3.6928,-7.90875)
};
\end{axis}
\end{tikzpicture}%
}%
\subcaptionbox{Zoom: $2.435\leq \frac34c \leq 2.515$\label{figure_symm_breaking_zoomed}}
{%
\begin{tikzpicture}
\begin{axis}[
scale only axis, 
height=2.5cm,
width=\textwidth/2-width("$?0.5\%$")-0.2cm, 
xmin = 2.435,
xmax = 2.515, 
minor tick num=1,
xminorticks=false,
grid=both,
major grid style={dashed,gray!50},
minor grid style={dotted,gray!25},
yticklabel=\pgfmathparse{\tick}\pgfmathprintnumber{\pgfmathresult}\,\%,
yticklabel style = {font=\scriptsize,xshift=0.5ex},
scaled ticks=false,
tick label style={/pgf/number format/fixed}]
\addplot[
color = black,
fill = black,
mark = *,
only marks] coordinates {
(0,0)
(0.7386,0)
(1.4771,0)
(2.2157,0)
(2.4372,0)
(2.4520,0)
(2.4668,0)
(2.4742,0)
(2.4816,-0.00725)
(2.4889,-0.02072)
(2.4963,-0.03483)
(2.5111,-0.06487)
(2.5850,-0.25329)
(2.6588,-0.51335)
(2.7327,-0.81314)
(2.8065,-1.18702)
(2.9542,-2.05948)
(3.6928,-7.90875)
};
\end{axis}
\end{tikzpicture}%
}%
\caption{Relative gain of energy $\frac{8E_{\K,\lambda}(c)-E_{2\cdot\K,8\lambda}(c)}{8E_{\K,\lambda}(c)}$.}\label{figure_symm_breaking_joined}
\end{figure}
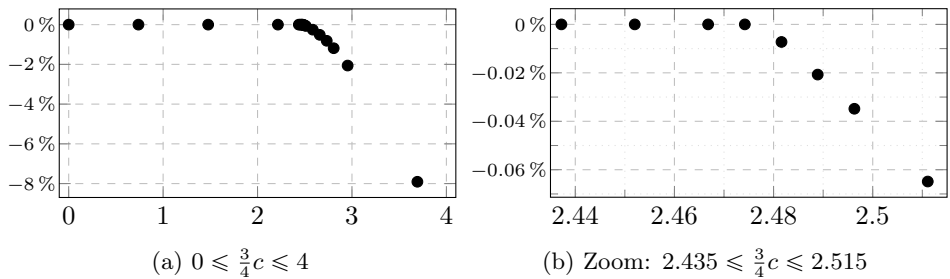
More precisely, minimizing the $2\cdot\K$ problem and the $1\cdot\K$ problem result in the same minimum energy (up to a factor 8) if $\frac34c\lesssim2.474$ while, for $\frac34c\gtrsim2.482$, we have found (at least) one \mbox{$2$-periodic} function for which the energy is lower than the minimal energy for the $1\cdot\K$ problem. Note that changing $c_W$ would affect the critical value of the Dirac constant at which symmetry breaking occurs but the value of $c_W$ does not affect the mathematical proofs (which are presented with $c_W=1$ for convenience).

The plots of the computed minimizers presented in Figure~\ref{figure_symm_breaking_density} visually confirm the symmetry breaking. They also suggest that the electronic density is very much concentrated. However, since the computation uses pseudo-potentials, only one outer shell electron is computed and the density is sharp on an annulus for these values of $c$.
\begin{figure}[h]
\centering
\subcaptionbox{$\frac34c=\frac34 3.35\sqrt[3]{\frac3\pi}\approx2.474$\label{figure_symm_breaking_density_before}}
{%
    \includegraphics[width=0.32\columnwidth,draft=false]{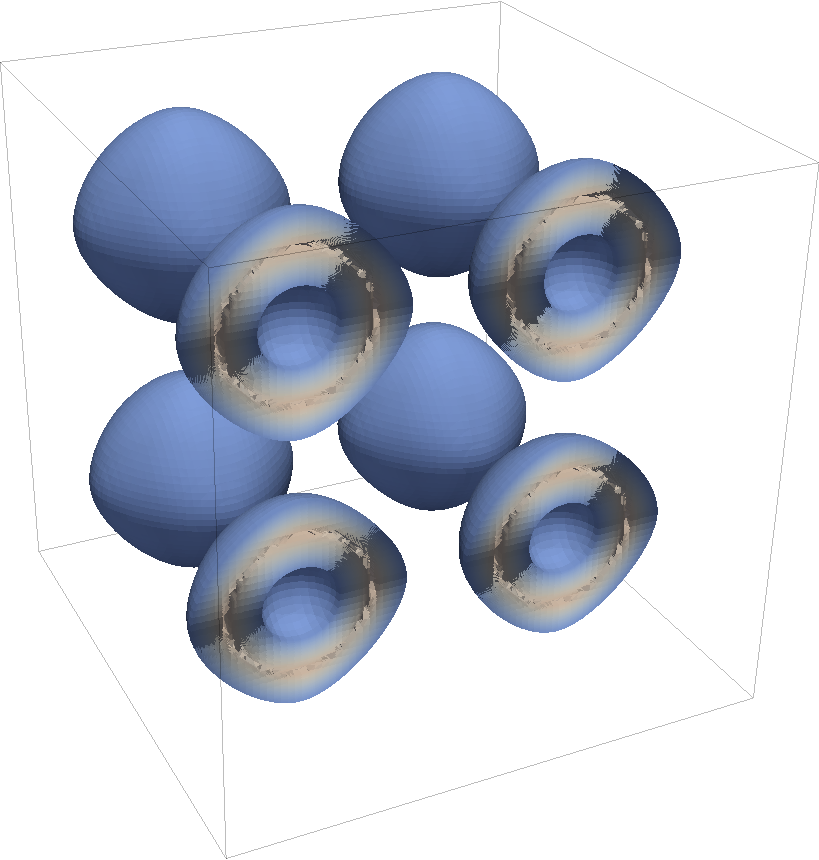}
    \hspace*{\fill}
}%
\subcaptionbox{$\frac34c=\frac34 3.36\sqrt[3]{\frac3\pi}\approx2.482$\label{figure_symm_breaking_density_just_after}}
{%
    \includegraphics[width=0.32\columnwidth,draft=false]{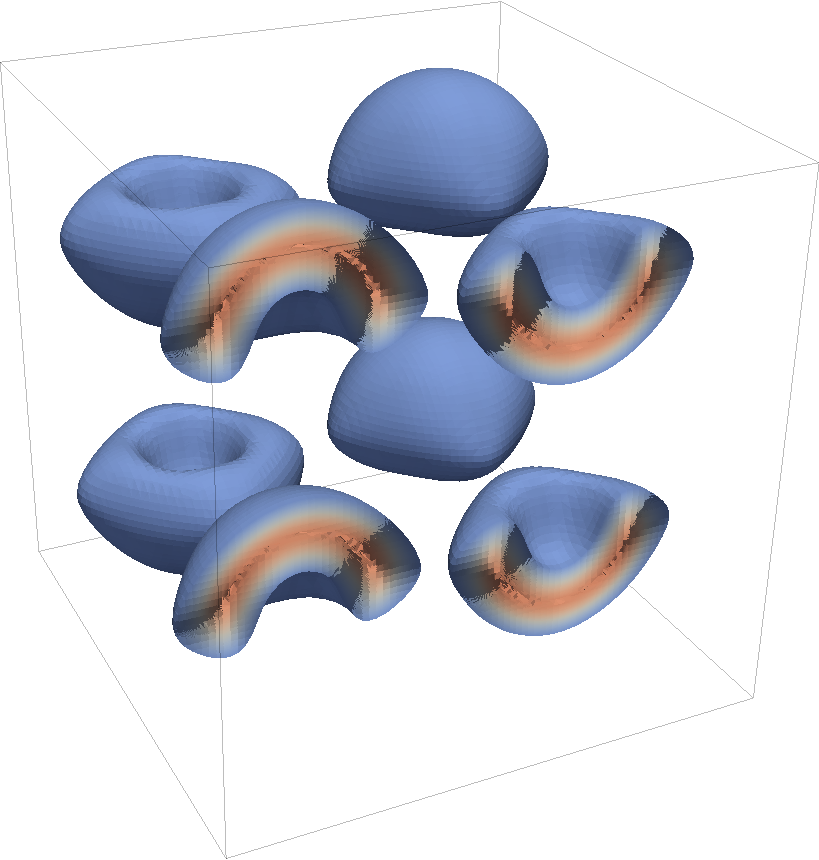}
}%
\subcaptionbox{$\frac34c=\frac34 3.50\sqrt[3]{\frac3\pi}\approx2.585$\label{figure_symm_breaking_density_after}}
{%
    \includegraphics[width=0.32\columnwidth,draft=false]{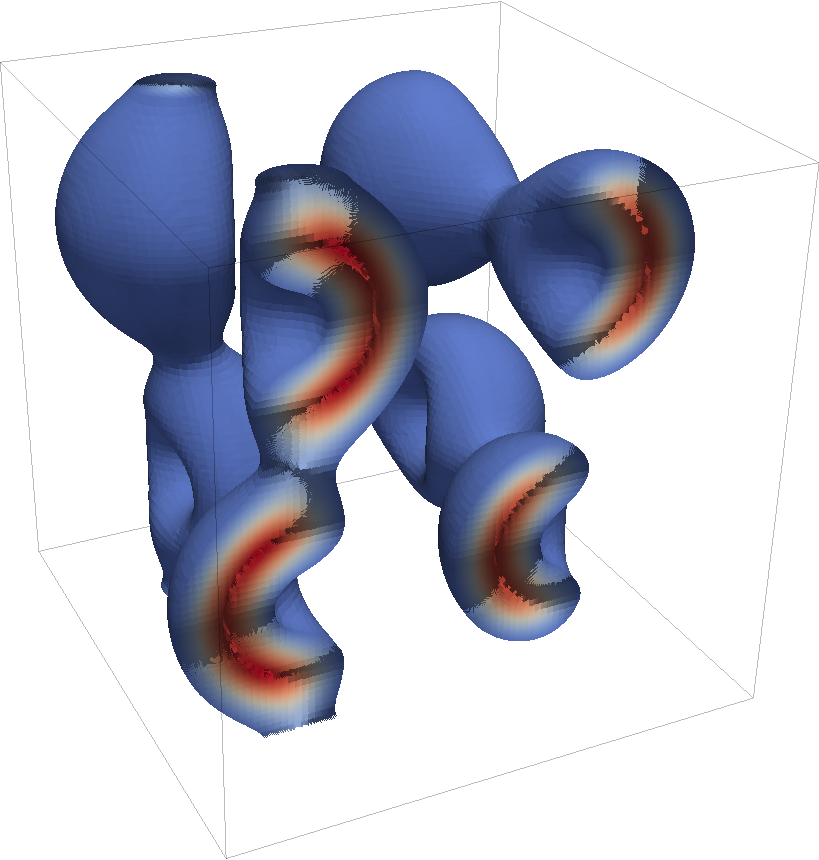}
}%
\caption{Electron density for $Z=1$ and length side $4$\AA. Same "dark-blue to white to dark-red" density scale for (a), (b) and (c).\\(a) \hspace{0.3cm}The computed $2$-periodic minimizer is still $1$-periodic.\\(b-c) The computed $2$-periodic minimizer is not $1$-periodic.}\label{figure_symm_breaking_density}
\end{figure}

The numerical value of the critical constant $\frac34 c\approx 2.48$ obtained in our numerical simulations is outside the usual values $\frac34 c\in[0.93;1.64]$ chosen in the literature. However, it is of the same order of magnitude and one cannot exclude that symmetry breaking would happen inside this range for different systems, meaning for different values of $Z$ and/or of the size of the lattice.

\begin{rmq}[Pseudo-potentials]\label{rmq_pseudo}
The software \emph{PROFESS v.3.0} that we used in our simulations is based on pseudo-potentials~\cite{Johnson-73}. This means that only $n$ outer shell electrons among the $N$ electrons of the unit cell are considered. The $N-n$ other ones are described through a pseudo-potential, together with the nucleus. Mathematically, this means that we have $\lambda=n$ and that the nucleus-electron interaction $-N\int_\K G_\K|w|^2$ is replaced by $-\int_\K G_{\textrm{ps}}|w|^2$ where the $\K$-periodic function $G_{\textrm{ps}}(x)$ behaves like $n/|x|$ when $|x|\to0$. All our results apply to this case as well. More precisely, we only need that $G_{\textrm{ps}}(x)-n/|x|$ is bounded on $\K$. We emphasize that the electron-electron interaction $D_\K$ is not changed by this generalization, and still involves the periodic Coulomb potential $G_\K$.
\end{rmq}

\section{The effective model in \texorpdfstring{$\R^3$}{R3}}\label{section_R3_eff_model}
This section is dedicated to the proof of~\cth{R3_eff_model_existence_thm} and~\cth{R3_eff_model_existence_and_nondeg}. We first give a lemma on the functional ${\mathscr J}_{\R^3}$, which has been defined in \eqref{R3_eff_model_functional}.
\begin{lemme}\label{R3_eff_model_lower_bound_on_NRJ}
For $c_{TF},\lambda>0$ and $u\in H^1(\R^3)$ such that $\norm{u}_2^2=\lambda$, we have
\begin{equation}\label{R3_eff_model_general_lower_bound_on_NRJ_c_quadratic_minoration_energy}
{\mathscr J}_{\R^3}(u)\geq\norm{\nabla u}^2_{L^2(\R^3)}-\frac{15}{64}\frac{\lambda}{c_{TF}}.
\end{equation}
\end{lemme}
\begin{proof}[Proof of~\clm{R3_eff_model_lower_bound_on_NRJ}]
It follows from
$$\frac35c_{TF}|u|^{\frac{10}3}-\frac34|u|^{\frac83} = \left(\sqrt{\frac35c_{TF}|u|^{\frac{10}3}}-\sqrt{\frac34|u|^{\frac83}}\right)^2 \geq -\frac{15\lambda}{64c_{TF}}|u|^2.$$
\end{proof}
We deduce from this some preliminary properties for the effective model in $\R^3$.
\begin{lemme}[A priori properties of $J_{\R^3}(\lambda)$]\label{R3_eff_model_apriori_properties_J}
Let $c_{TF}$ and $\lambda$ be positive constants.
We have
\begin{equation}\label{R3_eff_model_apriori_bounds}
-\frac{15}{64}\frac{\lambda}{c_{TF}} < J_{\R^3}(\lambda)<0.
\end{equation}
The function, $\lambda\mapsto J_{\R^3}(\lambda)$ is continuous on $[0;+\infty)$ and negative, concave and strictly decreasing on $(0;+\infty)$.
\end{lemme}
\begin{proof}[Proof of~\clm{R3_eff_model_apriori_properties_J}]
The negativity of $J_{\R^3}(\lambda)$ is obtained by taking $\nu$ large enough in the computation of ${\mathscr J}_{\R^3}(\nu^{-\frac32}u(\nu^{-1}\cdot))$.
\clm{R3_eff_model_lower_bound_on_NRJ} gives the lower bound in~\eqref{R3_eff_model_apriori_bounds}, which implies the continuity at $\lambda=0$. Moreover, after scaling, we have
\begin{equation*}
J_{\R^3}(\lambda)=\lambda \underset{\hspace{1.4cm}=:F(\lambda^{-2/3})}{\underbrace{\inf\limits_{\substack{u\in H^1(\R^3) \\ \norm{u}_{L^2(\R^3)}^2=1}}\left\{\lambda^{-\frac23}\norm{\nabla u}_{L^2(\R^3)}^2+\frac35c_{TF}\norm{u}^{\frac{10}3}_{L^{\frac{10}3}(\R^3)}-\frac34\norm{u}^{\frac83}_{L^{\frac83}(\R^3)}\right\}}}
\end{equation*}
where $t\mapsto F(t)$ is concave on $[0;+\infty)$, since $a\mapsto\inf_u (a f(u)+g(u))$ is concave for all $f,g$, hence continuous on $(0;+\infty)$ on which it is also negative (because $J_{\R^3}$ is negative) and non-decreasing. The continuity of $F$ gives that $\lambda\mapsto J_{\R^3}(\lambda)$ is continuous as well. Moreover, if $f$ is a concave non-decreasing negative function, then $\lambda\mapsto \lambda f(\lambda^{-2/3})$ is concave strictly decreasing on $(0,\infty)$, which proves that our energy $J$ is concave. To prove that, one can regularize $f$ by means of a convolution and then compute its first two derivatives.
\end{proof}

\subsection{Proof of~\texorpdfstring{\cth{R3_eff_model_existence_thm}}{existence of J(R3) minimizers}}\label{section_R3_eff_model_existence}
We divide the proof into several steps for clarity.
\addtocontents{toc}{\SkipTocEntry} 
\subsubsection*{\textbf{Step 1: Large binding inequality}}\label{Proof_THM_existence_R3_step_large_binding}
\begin{lemme}\label{R3_eff_model_large_binding}
Let $c_{TF}\geq0$ be a constant. Then
\begin{equation}\label{R3_eff_model_large_binding_ineq}
J_{\R^3}(\lambda)\leq J_{\R^3}(\lambda')+J_{\R^3}(\lambda-\lambda'), \qquad \forall \; 0\leq\lambda'\leq\lambda
\end{equation}
\end{lemme}
\begin{proof}[Proof of~\clm{R3_eff_model_large_binding}]
The inequality~\eqref{R3_eff_model_large_binding_ineq} is obtained by computing ${\mathscr J}_{\R^3}(\phi+\chi)$ where $\phi$ and $\chi$ are two bubbles of disjoint compact supports and of respective masses $\lambda'$ and $\lambda-\lambda'$.
\end{proof}
\begin{rmq*}
The strict inequality in~\eqref{R3_eff_model_large_binding_ineq}, which is important for applying Lions' concentration-compactness method, actually holds and is proved later in~\cpr{R3_eff_model_strict_binding}.
\end{rmq*}

\addtocontents{toc}{\SkipTocEntry} 
\subsubsection*{\textbf{Step 2: For any $\lambda>0$, $J_{\R^3}(\lambda)$ has a minimizer.}}\label{Proof_THM_existence_R3_step_existence_minimizers}
This is a classical result to which we will only give a sketch of proof (for a detailed proof, see~\cite{Ricaud-PhD}).
First, by rearrangement inequalities, we have ${\mathscr J}_{\R^3}(v)\geq{\mathscr J}_{\R^3}(v^*)$ for every $v\in H^1(\R^3)$. Therefore, one can restrict the minimization to nonnegative radial decreasing functions. By the compact embedding $H_{rad}^1(\R^3) \hookrightarrow L^p(\R^3)$, for $2<p<6$, we find
\begin{equation}\label{Proof_THM_existence_R3_step_existence_minimizers_sequenceOFineq}
{J}_{\R^3}(\lambda')\leq{\mathscr J}_{\R^3}(Q)\leq\lim\inf{\mathscr J}_{\R^3}(Q_n)={J}_{\R^3}(\lambda)
\end{equation}
for a minimizing sequence $Q_n\wto Q$ and where $\lambda':=\norm{Q}_{L^2(\R^3)}^2\leq\lambda$. Then, $J_{\R^3}$ being strictly decreasing by~\clm{R3_eff_model_apriori_properties_J}, $\lambda'=\lambda$ and the limit is strong in $L^2(\R^3)$, hence in $H^1(\R^3)$ by classical arguments. This proves that the limit $Q$ is a minimizer.

\addtocontents{toc}{\SkipTocEntry} 
\subsubsection*{\textbf{Step 3: Any minimizer is in} \texorpdfstring{$\boldsymbol{H^2(\R^3)}$}{\textbf{H2}} \textbf{and solves the E-L equation~\eqref{R3_eff_model_EulerLagrange}}}\label{Proof_THM_existence_R3_step_EL_and_H2}
The proof that any minimizer solves the Euler--Lagrange equation is classical and implies, together with $u\in H^1(\R^3)$, that $u\in H^2(\R^3)$ by elliptic regularity. Moreover, we have
\begin{equation}\label{R3_eff_model_EulerLagrange_mu_formulae}
\mu=-\frac{\norm{\nabla Q}_{L^2(\R^3)}^2+c_{TF}\norm{Q}^{10/3}_{L^{10/3}(\R^3)}-\norm{Q}^{8/3}_{L^{8/3}(\R^3)}}{\lambda}.
\end{equation}

\addtocontents{toc}{\SkipTocEntry} 
\subsubsection*{\textbf{Step 4: Strict binding inequality}}\label{Proof_THM_existence_R3_step_strict_binding}
\begin{prop}\label{R3_eff_model_strict_binding}
Let $c_{TF}>0$ and $\lambda>0$.
\begin{equation*}\label{R3_eff_model_strict_binding_ineq_in_section}\tag{\ref{R3_eff_model_strict_binding_ineq}}
\forall \; 0<\lambda'<\lambda, \; J_{\R^3}(\lambda)< J_{\R^3}(\lambda')+J_{\R^3}(\lambda-\lambda').
\end{equation*}
In particular, for any integer $N\geq2$,
\begin{equation}\label{R3_eff_model_strict_binding_ineq_Nlambda}
J_{\R^3}(N^3\lambda)<N^3J_{\R^3}(\lambda)<0.
\end{equation}
\end{prop}
\begin{proof}[Proof of~\cpr{R3_eff_model_strict_binding}]By the same scaling as in~\clm{R3_eff_model_apriori_properties_J}, we have
\begin{equation}\label{R3_eff_model_strict_binding_definition_F_lambda}
J_{\R^3}(\lambda)=\lambda \inf\limits_{\substack{u\in H^1(\R^3) \\ \norm{u}_{L^2(\R^3)}^2=1}}\underset{\hspace{1.2cm}=:{\mathscr F}_{\lambda}(u)}{\underbrace{\left\{\lambda^{-\frac23}\norm{\nabla u}_{L^2(\R^3)}^2+\frac35c_{TF}\norm{u}^{\frac{10}3}_{L^{\frac{10}3}(\R^3)}-\frac34\norm{u}^{\frac83}_{L^{\frac83}(\R^3)}\right\}}}.
\end{equation}
Let $\lambda>\lambda'>0$. By~\hyperref[Proof_THM_existence_R3_step_existence_minimizers]{Step~2}, the minimization problem
$$ \inf\limits_{\substack{u\in H^1(\R^3) \\ \norm{u}_{L^2(\R^3)}^2=1}}\left\{\norm{\nabla u}_{L^2(\R^3)}^2+\frac35c_{TF}{\lambda'}^{\frac23}\norm{u}^{\frac{10}3}_{L^{\frac{10}3}(\R^3)}-\frac34{\lambda'}^{\frac23}\norm{u}^{\frac83}_{L^{\frac83}(\R^3)}\right\}$$
has a minimizer $Q_{\lambda'}$ which, by~\hyperref[Proof_THM_existence_R3_step_EL_and_H2]{Step~3}, is in $H^2(\R^3)$ thus continuous. In particular, $\norm{\nabla Q_{\lambda'}}_{L^2(\R^3)}>0$ thus ${\mathscr F}_{\lambda'}(Q_{\lambda'})> {\mathscr F}_{\lambda}(Q_{\lambda'})$, where ${\mathscr F}_{\lambda}$ is defined in~\eqref{R3_eff_model_strict_binding_definition_F_lambda}. Therefore
\begin{align*}
J_{\R^3}(\lambda')=\lambda' {\mathscr F}_{\lambda'}(Q_{\lambda'})>\lambda' {\mathscr F}_{\lambda}(Q_{\lambda'})&=\frac{\lambda'}{\lambda} \mathscr J_{\R^3}(Q_{\lambda'}(\lambda^{-1/3}\cdot))\geq \frac{\lambda'}{\lambda} J_{\R^3}(\lambda),
\end{align*}
and we finally obtain
\begin{equation*}
J_{\R^3}(\lambda-\lambda')+J_{\R^3}(\lambda')>\frac{\lambda-\lambda'}{\lambda}J_{\R^3}(\lambda)+\frac{\lambda'}{\lambda}J_{\R^3}(\lambda)=J_{\R^3}(\lambda),
\end{equation*}
as we wanted.
\end{proof}

\addtocontents{toc}{\SkipTocEntry} 
\subsubsection*{\textbf{Step 5: }\texorpdfstring{$\boldsymbol{-\mu<0}$}{-mu<0}}\label{Proof_THM_existence_R3_step_mu_strictneg}
Let us choose $v$ in the minimization domain of $J_{\R^3}(1)$. Then, defining the positive number
$$\alpha_0=\frac38\frac{\norm{v}^{8/3}_{8/3}\lambda^{1/3}}{\norm{\nabla v}_{2}^2+\frac35c_{TF}\norm{v}^{10/3}_{10/3}\lambda^{2/3}},$$
we can obtain for any $\lambda>0$ an upper bound on $J_{\R^3}(\lambda)$. Namely
\begin{equation}\label{Proof_THM_existence_R3_step_mu_strictneg_first_ineq}
J_{\R^3}(\lambda)\leq{\mathscr J}_{\R^3}\left(\sqrt{\lambda}{{\alpha_0}^{3/2}v(\alpha_0}\cdot)\right)= -\frac9{64}\lambda^{5/3}\frac{\norm{v}^{16/3}_{8/3}}{\norm{\nabla v}_{2}^2+\frac35c_{TF}\norm{v}^{10/3}_{10/3}\lambda^{2/3}}.
\end{equation}
Moreover, for all $\epsilon$ and for $Q$ a minimizer to $J_{\R^3,c}(\lambda)$, we have
\begin{align*}
{\mathscr J}_{\R^3}((1-\epsilon)Q)&={\mathscr J}_{\R^3}(Q)+2\epsilon\lambda\mu+O(\epsilon^2),
\end{align*}
which leads, together with~\eqref{R3_eff_model_large_binding_ineq} and the fact that $Q$ is a minimizer of $J_{\R^3}(\lambda)$, to
$$2\epsilon\lambda\mu+O(\epsilon^2)\geq J_{\R^3}((1-\epsilon)^2\lambda)- J_{\R^3}(\lambda)\geq - J_{\R^3}(\epsilon(2-\epsilon)\lambda),$$
for any $\epsilon\in(0;2)$. Using this last inequality together with the upper bound~\eqref{Proof_THM_existence_R3_step_mu_strictneg_first_ineq}, we get for any $\epsilon\in(0;1)$ that
\begin{align*}
2\lambda\mu&\geq \frac9{64}\epsilon^{2/3}(2-\epsilon)^{5/3}\lambda^{5/3}\frac{\norm{v}^{16/3}_{8/3}}{\norm{\nabla v}_{2}^2+\frac35c_{TF}\norm{v}^{10/3}_{10/3}\epsilon^{2/3}(2-\epsilon)^{2/3}\lambda^{2/3}}+O(\epsilon)
\end{align*}
which leads to $\mu>0$ by taking $\epsilon$ small enough.

\addtocontents{toc}{\SkipTocEntry} 
\subsubsection*{\textbf{Step 6: Positivity of nonnegative minimizers}}\label{Proof_THM_existence_R3_step_Q_strict_pos}
Let $Q\geq0$ be a minimizer. By~\hyperref[Proof_THM_existence_R3_step_EL_and_H2]{Step~3}, $0 \nequiv Q \in H^2(\R^3)\subset C(\R^3)$ and $W:=c_{TF}|Q|^{\frac43}-|Q|^{\frac23}+\mu$ is in $\in L^\infty(\R^3)$. Therefore, the Euler--Lagrange equation gives $Q>0$ thanks to~\cite[Theorem 9.10]{LieLos-01}.

\addtocontents{toc}{\SkipTocEntry} 
\subsubsection*{\textbf{Step 7: nonnegative minimizers are radial strictly decreasing up to translations}}\label{Proof_THM_existence_R3_step_pos_minimizers_radial_strictDecreas}
This step is a consequence of~\hyperref[Proof_THM_existence_R3_step_Q_strict_pos]{Step~6} and is the subject of the following proposition.
\begin{prop}\label{R3_eff_model_radial_strict_decr_minimizer}
Let $\lambda>0$. Any positive minimizer to $J_{\R^3}(\lambda)$ is radial strictly decreasing, up to a translation.
\end{prop}
\begin{proof}[Proof of~\cpr{R3_eff_model_radial_strict_decr_minimizer}] Let $0\leq Q\in H^1(\R^3; \R)$ be a minimizer of $J_{\R^3}(\lambda)$. We denote by $Q^*$ its Schwarz rearrangement which is, as mentioned in first part of~\hyperref[Proof_THM_existence_R3_step_existence_minimizers]{Step~2}, also a minimizer and, consequently, $\int_{\R^3} |\nabla Q^*|^2 = \int_{\R^3} |\nabla Q|^2$. Moreover, by~\hyperref[Proof_THM_existence_R3_step_EL_and_H2]{Step~3} and~\hyperref[Proof_THM_existence_R3_step_Q_strict_pos]{Step~6}, $Q>0$ and $Q^*>0$ are in $H^2(\R^3; \R)$ and solutions of the Euler--Lagrange equation~\eqref{R3_eff_model_EulerLagrange}. They are therefore real-analytic (see e.g.~\cite{Morrey-58}) which implies that $\left|\{x|Q(x)=t\}\right|=\left|\{x|Q^*(x)=t\}\right|=0$ for any $t$. In particular, the radial non-increasing function $Q^*$ is in fact radial strictly decreasing. We then use~\cite[Theorem 1.1]{BroZie-88} to obtain $Q^*=Q$ a.e., up to a translation. Finally, $Q$ and $Q^*$ being continuous, the equality holds in fact everywhere.
\end{proof}

\addtocontents{toc}{\SkipTocEntry} 
\subsubsection*{\textbf{Step 8: }\texorpdfstring{$\boldsymbol{-\mu}$}{-mu} \textbf{is the lowest eigenvalue of}~\texorpdfstring{$\boldsymbol{H_Q}$,}{H,} \textbf{is simple, and} \texorpdfstring{$\boldsymbol{Q=z|Q|}$}{Q=z|Q|}}\label{Proof_THM_existence_R3_step_mu_lowest_simple_Q_up_phase_factor}
It is classical that the first eigenvalue of a Schr\"odinger operator $-\Delta+V$ is non-degenerate and that any nonnegative eigenfunction must be the first, see e.g.~\cite[Chapter 11]{LieLos-01}.

\addtocontents{toc}{\SkipTocEntry} 
\subsubsection*{\textbf{Step 9: Minimizing sequences are precompact up to a translations.}}\label{Proof_R3_eff_model_existence_thm_Step9}
Since the strict binding inequality~\eqref{R3_eff_model_strict_binding_ineq_in_section} holds, this follows from a result of Lions in~\cite[Theorem I.2]{Lions-84b}.

\medskip
This concludes the proof of~\cth{R3_eff_model_existence_thm}.

\qed

\subsection{Proof of~\texorpdfstring{\cth{R3_eff_model_existence_and_nondeg}}{uniqueness and non-degeneracy of positive solutions to E--L equation on R3}}\label{section_R3_eff_model_existence_and_nondeg}
The uniqueness of radial solutions has been proved by Serrin and Tang in~\cite{SerTan-00}. However, we need the non-degeneracy of the solution. Both uniqueness and non-degeneracy can be proved following line by line the method in~\cite[Thm. 2]{LewRot-15} (the argument is detailed in \cite{Ricaud-PhD}). One slight difference is the application of the moving plane method to prove that positive solutions are radial. Contrarily to~\cite{LewRot-15} we cannot use~\cite[Thm. 2]{GidNiNir-81} because our function
\begin{equation}\label{R3_eff_model_def_F}
F_\mu(y)=-c_{TF}y^{\frac73}+y^{\frac53}-\mu y
\end{equation}
is not $C^2$. However, given that nonnegative solutions are positive, one can show that they are $C^\infty$ and, therefore, we can apply~\cite[Thm. 1.1]{Li-91}.
\qed

\section{Regime of small \texorpdfstring{$c$}{c}: uniqueness of the minimizer to \texorpdfstring{$E_{\K,\lambda}(c)$}{the periodic TFDW model}}\label{section_small_c}
We first give some useful properties of $G_\K$ in the following lemma.
\begin{lemme}[The periodic Coulomb potential $G_\K$]\label{G_K_in_Lp}
The function $G_{\K}-|\cdot|^{-1}$ is bounded on $\K$. Thus, there exits $C$ such that for any $x\in\K\setminus\{0\}$, we have
\begin{equation}\label{equation_G_K_in_Lp}
0\leq G_\K(x) \leq \frac{C}{|x|}.
\end{equation}
In particular, $G_\K\in L^p(\K)$ for $1\leq p<3$. The Fourier transform of $G_\K$ is
\begin{equation}\label{G_K_fourier}
\widehat{G}_\K(\xi)=4\pi\sum\limits_{k\in{\mathscr L}^*_{\K}\setminus\{0\}} \frac{\delta_k(\xi)}{|k|^2}+\delta_0(\xi)\int_\K G_\K(x)\dd x
\end{equation}
where ${\mathscr L}^*_{\K}$ is the reciprocal lattice of ${\mathscr L}_\K$. Hence, for any $f\nequiv0$ for which $D_\K(f,f)$ is defined, we have $D_\K(f,f)>0$.
\end{lemme}
\begin{proof}[Proof of~\clm{G_K_in_Lp}]
The first part follows from the fact that
$$\lim\limits_{x\to0}{G_\K(x)-|x|^{-1}}=M\in\R,$$
see~\cite[VI.2]{LieSim-77b}. The expression of the Fourier transform is a direct computation.
\end{proof}

\subsection{Existence of minimizers to \texorpdfstring{$E_{\K,\lambda}(c)$}{the periodic TFDW model}}\label{Section_existence_minimizer_E}
In order to prove~\cth{main_result_2}, we need the existence of minimizers to $E_{\K,\lambda}(c)$, for any $c\geq0$, which is done in this section.
\begin{prop}[Existence of minimizers to $E_{\K,\lambda}(c)$]\label{K_complete_model_EulerLagrange}
Let $\K$ be the unit cube and, $c_{TF}>0$, $\lambda>0$ and $c\geq0$ be real constants.
\begin{enumerate}[label=\roman*.,leftmargin=1.4em]
	\item There exists a nonnegative minimizer to $E_{\K,\lambda}(c)$ and any minimizing sequence $(w_n)_n$ strongly converges in $H^1_{\textrm{per}}({\K})$ to a minimizer, up to extraction of a subsequence.
	\item Any minimizer $w_c$ is in $H^2_{\textrm{per}}({\K})$, is non-constant and solves the Euler--Lagrange equation
\begin{equation}\label{K_complete_model_EulerLagrange_equation}
\left(-\Delta +c_{TF}|w_c|^{\frac43}-c|w_c|^{\frac23}-G_\K +(|w_c|^2\star G_\K)\right)w_c=-\mu_{w_c} w_c,
\end{equation}
with
\begin{equation}\label{K_complete_model_EulerLagrange_mu_formulae}
\mu_{w_c}=-\frac{\norm{\nabla w_c}_2^2+c_{TF}\norm{w_c}^{10/3}_{10/3}-c\norm{w_c}^{8/3}_{8/3}+D_\K(|w_c|^2,|w_c|^2)-\pscal{G_\K,|w_c|^2}_{L^2({\K})}}{\lambda}.
\end{equation}
	\item Up to a phase factor, a minimizer $w_c$ is positive and the unique ground-state eigenfunction of the self-adjoint operator, with domain $H^2_{\textrm{per}}(\K)$,
$$H_{w_c}:=-\Delta +c_{TF}|w_c|^{\frac43}-c|w_c|^{\frac23}-G_\K +(|w_c|^2\star G_\K).$$
\end{enumerate}
\end{prop}
Since the problem is posed on a bounded domain, this is a classical result to which we only give a sketch of proof. For a detailed proof, see \cite{Ricaud-PhD}. Note that for shortness, we have denoted $\norm{\cdot}_p=\norm{\cdot}_{L^p(\K)}$.

\begin{proof}[Sketch of proof of~\cpr{K_complete_model_EulerLagrange}]
In order to prove \emph{i.}, we need the following result that will be useful all along the paper, and is somewhat similar to \clm{R3_eff_model_lower_bound_on_NRJ}.
\begin{lemme}\label{K_complete_model_lower_bound_on_NRJ}
There exist positive constants $a<1$ and $C$ such that for any $c\geq0$, $c_{TF}, \lambda>0$ and any $u\in H^1_{\textrm{per}}(\mathbb{K})$ with $\norm{u}_2^2=\lambda$, we have
\begin{equation}
{\mathscr E}_{\K,c}(u)\geq a\norm{\nabla u}_{L^2(\K)}^2-\frac{15}{64}\frac{\lambda}{c_{TF}}{c}^2-\lambda C.
\end{equation}
\end{lemme}
\begin{proof}[Proof of~\clm{K_complete_model_lower_bound_on_NRJ}]
As in \clm{R3_eff_model_lower_bound_on_NRJ} (but on $\K$) we have
$$\frac35c_{TF}\norm{u}_{L^{\frac{10}3}(\K)}^{\frac{10}3}-\frac34c\norm{u}_{L^{\frac83}(\K)}^{\frac83} \geq -\frac{15}{64}\frac{\lambda}{c_{TF}}{c}^2.$$
Moreover, for any $\epsilon>0$, we have
$$\Big|\int_{\K}{G_\K |u|^2}\Big| \leq \epsilon\norm{u}^2_{L^6(\K)}+\lambda C_\epsilon.$$
Indeed $G_\K=\mathds{1}_{\{|\cdot|< r\}} G_\K+\mathds{1}_{\K\setminus\{|\cdot|< r\}} G_\K \in L^{\frac32}(\K)+L^\infty(\K)$, by~\eqref{equation_G_K_in_Lp}, and $r$ can be chosen such that $\norm{\mathds{1}_{\{|\cdot|< r\}} G_\K}_{L^{\frac32}(\K)}\leq\epsilon$ to obtain the claimed inequality. The above results, together with Sobolev embeddings and $D_\K(u^2,u^2)\geq0$, gives
\begin{align*}
{\mathscr E}_{\K,c}(u)&= \norm{\nabla u}_{L^2(\K)}^2+\frac35c_{TF}\norm{u}_{L^{\frac{10}3}(\K)}^{\frac{10}3}-\frac34c\norm{u}_{L^{\frac83}(\K)}^{\frac83}+\frac12D_\K(u^2,u^2)-\int_\K{G_\K u^2}\\
	&\geq \norm{\nabla u}_{L^2(\K)}^2-\frac{15}{64}\frac{\lambda}{c_{TF}}{c}^2-\epsilon \norm{u}^2_{L^6(\K)}-\lambda C_\epsilon\\
	&\geq (1- \epsilon S)\norm{\nabla u}_{L^2(\K)}^2-\frac{15}{64}\frac{\lambda}{c_{TF}}{c}^2-\lambda (C_\epsilon+\epsilon S)
\end{align*}
for any $\epsilon>0$ and where $S$ is the constant from the Sobolev embedding. Choosing $\epsilon$ such that $\epsilon S <1$ concludes the proof.
\end{proof}

The above result together with the fact that $H^1(\K)$ is compactly embedded in $L^p(\K)$ for $1\leq p<6$ (since the cube $\K$ is bounded) and with Fatou's Lemma implies the existence of a minimizer and the strong convergence in $H^1({\K})$ of any minimizing sequence. Moreover, the convexity inequality for gradients (see~\cite[Theorem 7.8]{LieLos-01}) implies the existence of a nonnegative minimizer and concludes the proof of \emph{i}.

To prove that any minimizer $w_c$ is in $H^2_{\textrm{per}}({\K})$, we write
$$-\Delta w_c =-c_{TF}|w_c|^{\frac43}w_c+c|w_c|^{\frac23}w_c+G_\K w_c-(|w_c|^2\star G_\K)w_c-\mu_c w_c$$
and prove that the right hand side is in $L^2(\K)$, which will give $w_c \in H^2_{\textrm{per}}(\K)$ by elliptic regularity for the periodic Laplacian. We note that $|w_c|^{\frac43}w_c$ and $|w_c|^{\frac23}w_c$ are in $L^2(\K)$, by Sobolev embeddings, since $w_c\in H^1_{\textrm{per}}(\K)$ which also gives, together with $G_\K\in L^2(\K)$ by~\clm{G_K_in_Lp}, that $|w_c|^2\star G_\K \in L^\infty(\K)$. It remains to prove that $G_\K w_c\in L^2(\K)$: equation \eqref{equation_G_K_in_Lp} and the periodic Hardy inequality on $\K$ give
$$\norm{G_\K w_c}_{L^2(\K)}\leq C\norm{|\cdot|^{-1}w_c}_{L^2(\K)}\leq C'\norm{w_c}_{H^1_{\textrm{per}}(\K)}.$$
Finally, since $G_\K$ is not constant, the constant functions are not solutions of the Euler--Lagrange equation hence are not minimizers. This concludes the proof of \emph{ii}.

Let $w_c$ be a nonnegative minimizer, then $0\nequiv w_c\geq0$ is in $H^2(\K)\subset L^\infty(\K)$ and is a solution of $\left(-\Delta +C \right)u=\left(f + G_\K + C\right)u$, with $G_\K$ bounded below and
$$f=-c_{TF}|w_c|^{\frac43}+c|w_c|^{\frac23} -(|w_c|^2\star G_\K)-\mu_{w_c} \in L^\infty(\K),$$
thus $\left(-\Delta +C \right)w_c\geq0$ for $C\gg1$. Hence, $w_c>0$ on $\K$ since the periodic Laplacian is positive improving \cite[Theorem 9.10]{LieLos-01}. Consequently, $w_c>0$ verifies $H_{w_c}w_c=-\mu_{w_c} w_c$ and this implies that for any $u\in H^1_{\textrm{per}}(\K)$ it holds
$$\pscalSM{u,(H_{w_c}+\mu_{w_c})u}_{L^2(\K)}=\pscalSM{{w_c}^2 ,\left|\nabla(u{w_c}^{-1})\right|^2}_{L^2(\K)}\geq0.$$
This vanishes only if there exists $\alpha\in\C$ such that $u=\alpha w_c$ ae. It proves $w_c$ is the unique ground state of $H_{w_c}$ and concludes the proof of~\cpr{K_complete_model_EulerLagrange}.
\end{proof}

From this existence result, we deduce the following corollary.
\begin{cor}\label{K_complete_model_convergence_E_of_c}
On $[0, +\infty)$, $c\mapsto E_{\K,\lambda}(c)$ is continuous and strictly decreasing.
\end{cor}
\begin{proof}[Proof of~\ccr{K_complete_model_convergence_E_of_c}]Let $0\leq c_1< c_2$ and, let $w_1$ and $w_2$ be corresponding minimizers, which exist by~\cpr{K_complete_model_EulerLagrange}. On one hand, we have
\begin{align*}
E_{\K,\lambda}(c_2)\leq {\mathscr E}_{\K,c_2}(w_1)&=E_{\K,\lambda}(c_1)-\frac34(c_2-c_1)\norm{w_1}^{\frac83}_{L^{\frac83}(\K)}\\
	&<E_{\K,\lambda}(c_1)\leq{\mathscr E}_{\K,c_1}(w_2)=E_{\K,\lambda}(c_2)+\frac34(c_2-c_1)\norm{w_2}^{\frac83}_{L^{\frac83}(\K)}.
\end{align*}
This gives that $E_{\K,\lambda}(c)$ is strictly decreasing on $[0,+\infty)$ but also the left-continuity for any $c_2>0$. Moreover, $c_2\mapsto\norm{w_2}_{H^1(\K)}$ is uniformly bounded on any bounded interval since
\begin{equation}\label{maj_NRJ_Kin_by_fonctional}
E_{\K,\lambda}(0)\geq E_{\K,\lambda}(c_2)={\mathscr E}_{\K,c_2}(w_2)\geq a\norm{\nabla w_2}_{L^2(\K)}^2-\frac{15}{64}\frac{\lambda}{c_{TF}}{c_2}^2-\lambda C
\end{equation}
by \clm{K_complete_model_lower_bound_on_NRJ}. Hence, by the Sobolev embedding, we have
$$E_{\K,\lambda}(c_2)<E_{\K,\lambda}(c_1)\leq E_{\K,\lambda}(c_2)+\frac34(c_2-c_1)C_1\lambda^{5/6}\norm{w_2}_{H^1(\K)},$$
which gives the right-continuity and concludes the proof of~\ccr{K_complete_model_convergence_E_of_c}.
\end{proof}

\subsection{Limit case \texorpdfstring{$c=0$}{c=0}: the TFW model}
In order to prove~\cth{main_result_2}, we need some results on the TFW model which corresponds to the TFDW model for $c=0$. For clarity, we denote
\begin{equation}\label{K_TFW_NRJ_u}
{\mathscr E}^{TFW}_{\K}(w):={\mathscr E}_{\K,0}(w)=\int\limits_\K{|\nabla w|^2}+\frac35c_{TF}\int\limits_\K{|w|^{\frac{10}3}}+\frac12D_\K(|w|^2,|w|^2)-\int\limits_\K{G_\K |w|^2},
\end{equation}
and similarly $E^{TFW}_{\K,\lambda}:=E_{\K,\lambda}(0)$.

By~\cpr{K_complete_model_EulerLagrange}, there exist minimizers to $E^{TFW}_{\K,\lambda}$, and we now prove the uniqueness of minimizer for the TFW model.
\begin{prop}\label{K_TFW_model_EulerLagrange_and_properties_of_minimizer}
The minimization problem $E^{TFW}_{\K,\lambda}$ admits, up to phase, a unique minimizer $w_0$ which is non constant and positive. Moreover, $w_0$ is the unique ground-state eigenfunction of the self-adjoint operator
$$H:=-\Delta+c_{TF}|w_0|^{\frac43}-G_\K+(|w_0|^2\star G_\K),$$
with domain $H^2_{\textrm{per}}(\K)$, acting on $L^2_{\textrm{per}}(\K)$, and with ground-state eigenvalue
\begin{equation}\label{K_complete_TFW_model_EulerLagrange_mu_formulae}
-\mu_0=\frac{\norm{\nabla w_0}_2^2+c_{TF}\norm{w_0}^{10/3}_{10/3}+D_\K(w_0^2,w_0^2)-\pscal{G_\K,w_0^2}_{L^2(\K)}}{\lambda}.
\end{equation}
\end{prop}
\begin{proof}[Proof of~\cpr{K_TFW_model_EulerLagrange_and_properties_of_minimizer}]
By~\cpr{K_complete_model_EulerLagrange}, we only have to prove the uniqueness. It follows from the convexity of the $\rho\mapsto|\nabla\sqrt{\rho}|^2$ (see \cite[Proposition 7.1]{Lieb-81b}) and the strict convexity of $\rho\mapsto D_\K(\rho,\rho)$.
\end{proof}

\subsection{Proof of~\texorpdfstring{\cth{main_result_2}: uniqueness}{uniqueness} in the regime of small \texorpdfstring{$c$}{c}}
We first prove one convergence result and a uniqueness result under a condition on $\min\limits_\K \rho$.
\begin{lemme}\label{cvgce_minimizers_to_TFW}
Let $\{c_n\}_n\subset\R_+$ be such that $c_n\to\bar{c}$. If $\{w_{c_n}\}_{n}$ is a sequence of respective positive minimizers to $E_{\K,\lambda}(c_n)$ and $\{\mu_{w_{c_n}}\}_{n}$ the associated Euler--Lagrange multipliers, then there exists a subsequence $c_{n_k}$ such that the convergence
$$\big(w_{c_{n_k}}, \mu_{w_{c_{n_k}}}\big)\underset{k\to\infty}{\longrightarrow}\left(\bar{w}, \mu_{\bar{w}}\right)$$
holds strongly in $H^2_{\textrm{per}}(\K)\times\R$, where $\bar{w}$ is a positive minimizer to $E_{\K,\lambda}(\bar{c})$ and $\mu_{\bar{w}}$ is the associated multiplier.

Additionally, if $E_{\K,\lambda}(\bar{c})$ has a unique positive minimizer $\bar{w}$ then the result holds for the whole sequence $c_n\to\bar{c}$:
$$\big(w_{c_n}, \mu_{w_{c_n}}\big)\underset{n\to\infty}{\longrightarrow}\left(\bar{w}, \mu_{\bar{c}}\right).$$
\end{lemme}
We will only use the case $\bar{c}=0$, for which we have proved the uniqueness of the positive minimizer, but we state this lemma for any $\bar{c}\geq0$.
\begin{proof}[Proof of~\clm{cvgce_minimizers_to_TFW}]
We first prove the convergence in $H^1_{\textrm{per}}(\K)\times\R$.
By the continuity of $c\mapsto E_{\K,\lambda}(c)$ proved in~\ccr{K_complete_model_convergence_E_of_c}, $\{w_{c_n}\}_{n\to\infty}$ is a positive minimizing sequence of $E_{\K,\lambda}(\bar{c})$. Thus, by~\cpr{K_complete_model_EulerLagrange}, up to a subsequence (denoted the same for shortness), $w_{c_n}$ converges strongly in $H^1_{\textrm{per}}(\K)$ to a minimizer $\bar{w}$ of $E_{\K,\lambda}(\bar{c})$.

Moreover, for any $c$, $(w_c,\mu_{w_c})$ is a solution of the Euler--Lagrange equation
$$\left(-\Delta +c_{TF}{w_c}^{\frac43}-c{w_c}^{\frac23}-G_\K +({w_c}^2\star G_\K)\right)w_c=-\mu_{w_c} w_c.$$
Thus, as $c_n$ goes to $\bar{c}$, $\mu_{w_{c_n}}$ converges to $\mu\in\overline{\R}$ satisfying
$$-\Delta\bar{w}+c_{TF}\bar{w}^{\frac73}-\bar{c}\bar{w}^{\frac53}-G_\K\bar{w}+(\bar\rho\star G_\K)\bar{w}=-\mu\bar{w}.$$
In particular, $\mu=\mu_{\bar{w}}$. At this point, we proved the convergence in $H^1_{\textrm{per}}(\K)\times\R$: $$\left(w_{c_n}, \mu_{w_{c_n}}\right)\underset{n\to\infty}{\longrightarrow}\left(\bar{w}, \mu_{\bar{w}}\right).$$

If, additionally, the positive minimizer $\bar{w}$ of $E_{\K,\lambda}(\bar{c})$ is unique, then any positive minimizing sequence must converge in $H^1_{\textrm{per}}(\K)$ to $\bar{w}$, so the whole sequence $\{w_{c_n}\}_{n\to\infty}$ in fact converges to the unique positive minimizer $\bar{w}$.

We turn to the proof of the convergence in $H^2_{\textrm{per}}(\K)$. For any $c_n\geq0$, by~\cpr{K_complete_model_EulerLagrange}, $w_{c_n}$ is in $H^2_{\textrm{per}}(\K)$ thus we have
\begin{align*}
\left(-\Delta-G_\K+\beta\right)\left(w_{c_n}-\bar{w}\right)=&-c_{TF}({w_{c_n}}^{\frac73}-\bar{w}^{\frac73})+(c_n-\bar{c}) {w_{c_n}}^{\frac53} + \bar{c} \left({w_{c_n}}^{\frac53}-{\bar{w}}^{\frac53}\right)\\
	&-\left(({w_{c_n}}^2-{\bar{w}}^2)\star G_\K\right)w_{c_n}-\left({\bar{w}}^2\star G_\K\right)\left(w_{c_n}-\bar{w}\right)\\
	&-(\mu_{w_{c_n}}-\mu_{\bar{w}})w_{c_n} +(\beta-\mu_{\bar{w}})\left(w_{c_n}-\bar{w}\right)=:\epsilon_n.
\end{align*}
The right side $\epsilon_n$ converges to $0$ in $L^2_{\textrm{per}}(\K)$. Moreover, by the Rellich-Kato theorem, the operator $-\Delta-G_\K$ is self-adjoint on $H^2_{\textrm{per}}(\K)$ and bounded below, hence we conclude that
\begin{multline*}
\norm{w_{c_n}-\bar{w}}_{H^2(\K)}=\norm{\left(-\Delta-G_\K+\beta\right)^{-1}\epsilon_n}_{H^2(\K)}\\
	\leq\normt{\left(-\Delta-G_\K+\beta\right)^{-1}}_{L^2(\K)\to H^2_{\textrm{per}}(\K)}\norm{\epsilon_n}_{L^2(\K)}\underset{n\to\infty}{\longrightarrow}0.
\end{multline*}
This concludes the proof of~\clm{cvgce_minimizers_to_TFW}.
\end{proof}

\begin{prop}[Conditional uniqueness]\label{K_complete_model_uniqueness_EulerLagrange_solutions_conditional}
Let $\K$ be the unit cube, $N\geq1$ be an integer, $c_{TF}>0$, $c\geq0$ and $\mu\in\R$ be  constants. Let $w>0$ be such that $w\in H^1(N\cdot\K)$ and $w$ is a $N\cdot\K-$periodic solution of
\begin{equation}\label{Euler_TFDW}
\left(-\Delta +c_{TF}w^{\frac43}-cw^{\frac23}+(w^2\star G_\K)-G_\K\right) w=-\mu w.
\end{equation}
If $\min\limits_{N\cdot\K} w > \left(\frac{c}{c_{TF}}\right)^{\frac32}$, then $w$ is the unique minimizer of $E_{N\cdot\K,\int_{N\cdot\K}{|w|^2}}(c)$.
\end{prop}
\begin{proof}[Proof of~\cpr{K_complete_model_uniqueness_EulerLagrange_solutions_conditional}]
First, the hypothesis give $w\in H^2_{\textrm{per}}(N\cdot\K)$, by the same proof as in \cpr{K_complete_model_EulerLagrange}. Moreover, we have the following lemma.
\begin{lemme}\label{K_complete_model_positivity_2ndTL} Let $\rho>0$ and $\rho'\geq0$ such that $\sqrt\rho\in H^2_{\textrm{per}}(\K)$ and $\sqrt{\rho'}\in H^1_{\textrm{per}}(\K)$. Then $$\int_{ \K}{\left|\nabla\sqrt{\rho'}\right|^2}-\int_{ \K}{\left|\nabla\sqrt{\rho}\right|^2}+\int_{ \K}{\frac{\Delta\sqrt{\rho}}{\sqrt{\rho}}(\rho'-\rho)}\geq0.$$
\end{lemme}
\begin{proof}[Proof of~\clm{K_complete_model_positivity_2ndTL}]
Using the fact that
$$\sqrt{\rho}\Delta\sqrt{\rho}=\frac{\sqrt{\rho}}2\nabla\left[\sqrt{\rho}\nabla(\ln\rho)\right]=\frac12\rho\Delta(\ln\rho)+\frac14\rho\left|\nabla(\ln\rho)\right|^2$$
and defining $h=\rho'-\rho$, one obtains
$$\int_{N\cdot\K}{\left|\nabla\sqrt{\rho+h}\right|^2}-\int_{N\cdot\K}{\left|\nabla\sqrt{\rho}\right|^2}+\int_{N\cdot\K}{\frac{\Delta\sqrt{\rho}}{\sqrt{\rho}}h}=\frac14\int_{N\cdot\K}{\left|\frac{h\nabla\rho}{\rho\sqrt{\rho+h}}-\frac{\nabla h}{\sqrt{\rho+h}}\right|^2}\geq0.$$
\end{proof}
Let $w'$ be in $H^1_{\textrm per}(N\cdot\K)$ such that $\int_{N\cdot\K}{w^2}=\int_{N\cdot\K}{|w'|^2}$ and $|w'|\nequiv w$. Defining $\rho=w^2$ and $\rho'=|w'|^2$, this means that $\int_{N\cdot\K}{h}=0$ where $h:=\rho'-\rho\nequiv0$. We have
\begin{align*}
&{\mathscr E}_{N\cdot\K,c}(|w'|)-{\mathscr E}_{N\cdot\K,c}(w)\\
&=\pscal{\left(-\Delta+c_{TF}w^{\frac43}-c w^{\frac23}+w^2\star G_{N\cdot\K}-G_{N\cdot\K}+\mu\right)w,h w^{-1}}_{L^2(N\cdot\K)}\\
&\;\;+\int_{{N\cdot\K}}{|\nabla\sqrt{\rho+h}|^2}-\int_{{N\cdot\K}}{|\nabla\sqrt{\rho}|^2}+\int_{{N\cdot\K}}{\frac{\Delta\sqrt{\rho}}{\sqrt{\rho}}h}+\frac12D_{N\cdot\K}(h,h)\\
&\;\;+\frac35c_{TF}\left(\int_{{N\cdot\K}}{(\rho+h)^{\frac53}-\rho^{\frac53}-\frac53\rho^{\frac23}h}\right)-\frac34c\left(\int_{{N\cdot\K}}{(\rho+h)^{\frac43}-\rho^{\frac43}-\frac43\rho^{\frac13}h}\right)\\
&>\int_{{N\cdot\K}}{F(\rho')-F(\rho)-F'(\rho)(\rho'-\rho)},
\end{align*}
with $F(X)=\frac35c_{TF}X^{\frac53}-\frac34cX^{\frac43}$. The above inequality comes from~\eqref{Euler_TFDW} together with~\clm{K_complete_model_positivity_2ndTL} and with $D_\K(h,h)>0$ for $h\nequiv0$. Defining now
$$F_X(Y)=F(Y)-F(X)-F'(X)(Y-X),$$
one can check, as soon as $X\geq\sqrt[3]{\frac{c}{c_{TF}}}$, that $F_X'<0$ on $(0,X)$ and $F_X'>0$ on $(X,+\infty)$. Moreover, $F_X'(0)<0$ if $X>\sqrt[3]{\frac{c}{c_{TF}}}$. Thus $F_X$ has a global strict minimum on $\R_+$ at $X$ and this minimum is zero. Consequently, if $\min\limits_{N\cdot\K} w\geq\big(\frac{c}{c_{TF}}\big)^{3/2}$, then ${\mathscr E}_{\K,c}(w')\geq{\mathscr E}_{\K,c}(|w'|)>{\mathscr E}_{\K,c}(w)$ for any $w'\in H^1_{\textrm per}(N\cdot\K)$ such that $|w'|\nequiv w$ and $\int_{N\cdot\K}{|w'|^2}=\int_{N\cdot\K}{w^2}$. This ends the proof of~\cpr{K_complete_model_uniqueness_EulerLagrange_solutions_conditional}.
\end{proof}

We have now all the tools to prove the uniqueness of minimizers for $c$ small.
\begin{proof}[Proof of~\cth{main_result_2}]
We have already proved all the results of~\emph{i.} of \cth{main_result_2} in \cpr{K_complete_model_EulerLagrange} except for the uniqueness that we prove now. Let $(w_c)_{c\to0^+}$ be a sequence of respective positive minimizers to $E_{\K,\lambda}(c)$. By~\cpr{K_TFW_model_EulerLagrange_and_properties_of_minimizer}, $E_{\K,\lambda}(0)$ has a unique minimizer thus, by~\cpr{cvgce_minimizers_to_TFW}, $w_c$ converges strongly in $H^2(\K)$ hence in $L^\infty(\K)$ to the unique positive minimizer $w_0$ to $E_{\K,\lambda}(0)$. Therefore, for $c$ small enough we have
$$\min\limits_\K w_c \geq \frac12\min\limits_\K w_0>\left(\frac{c}{c_{TF}}\right)^{\frac32}$$
and we can apply~\cpr{K_complete_model_uniqueness_EulerLagrange_solutions_conditional} (with $N=1$) to the minimizer $w_c>0$ to conclude that it is the unique minimizer of $E_{\K,\lambda}(c)$.

We now prove~\emph{ii.} of~\cth{main_result_2}. We fix $c$ small enough such that $E_{\K,\lambda}(c)$ has an unique minimizer $w_c$. Then $w_c$ being $\K$-periodic, it is $N\cdot\K-$periodic for any integer $N\geq1$ and verifies all the hypothesis of~\cpr{K_complete_model_uniqueness_EulerLagrange_solutions_conditional} hence it is also the unique minimizer of $E_{N\cdot\K,\int_{N\cdot\K}{|w_c|^2}}(c)=E_{N\cdot\K,N^3\lambda}(c)$.
\end{proof}

\section{Regime of large \texorpdfstring{$c$}{c}: symmetry breaking}\label{section_large_c}
This section is dedicated to the proof of the main result of the paper, namely \cth{main_result}.
We introduce for clarity some notations for the rest of the paper. We will denote the minimization problem for the effective model on the unit cell $\K$ by
\begin{equation}\label{K_eff_model_minimization}
\boxed{J_{\K,\lambda}(c)=\inf\limits_{\substack{v\in H^1_{\textrm{per}}(\K) \\ \norm{v}_{L^2(\K)}^2=\lambda}} {\mathscr J}_{\K,c}(v),}
\end{equation}
where
\begin{equation}\label{K_eff_model_functional}
{\mathscr J}_{\K,c}(v)=\int_{\K}{|\nabla v|^2}+\frac35c_{TF}\int_{\K}{|v|^{\frac{10}3}}-\frac34 c \int_{\K}{|v|^{\frac83}}.
\end{equation}

The first but important result is that we have for $J_{\K,\lambda}$ the existence result equivalent to the existence result of \cpr{K_complete_model_EulerLagrange} for $E_{\K,\lambda}$.

The minima of the effective model and of the TFDW model also verify the following a priori estimates which will be useful all along this section.
\begin{lemme}[A priori estimates on minimal energy]\label{K_models_apriori_estimates_J}
Let $\K$ be the unit cube and $c_{TF}>0$ be a constant. There exists $C>0$ such that for any $c>0$ we have
\begin{align}
- \lambda C-\frac{15}{64}\frac{\lambda}{c_{TF}}{c}^2\leq{}& E_{\K,\lambda}(c)\label{K_complete_model_apriori_bounds}\\
\intertext{and}
-\frac{15}{64}\frac{\lambda}{c_{TF}}{c}^2\leq{}& J_{\K,\lambda}(c)\leq -\frac34\frac{\lambda^{\frac43}}{|\K|^{\frac13}}c+\frac35c_{TF}\frac{\lambda^{\frac53}}{|\K|^{\frac23}}.\label{K_eff_model_apriori_bounds}
\end{align}
Moreover, for all $K$ such that $0<K< -J_{\R^3,\lambda}$, there exists $c_*>0$ such that for all $c\geq c_*$ we have
\begin{equation}
-\frac{15}{64}\frac{\lambda}{c_{TF}}{c}^2\leq J_{\K,\lambda}(c)\leq -{c}^2 K<0.\label{K_eff_model_apriori_upper_bound}
\end{equation}
\end{lemme}
\begin{proof}[Proof of~\clm{K_models_apriori_estimates_J}]
The inequality \eqref{K_complete_model_apriori_bounds} has been proved in \clm{K_complete_model_lower_bound_on_NRJ}, the proof of which also leads to the inequality
\begin{equation}\label{K_eff_model_lower_bound_on_NRJ}
{\mathscr J}_{\K,c}(v)\geq\norm{\nabla v}^2_{L^2(\K)}-\frac{15}{64}\frac{\lambda}{c_{TF}}{c}^2,
\end{equation}
hence the lower bound in~\eqref{K_eff_model_apriori_bounds}. The upper bound in~\eqref{K_eff_model_apriori_bounds} is a simple computation of ${\mathscr J}_{\K,c}(\bar{v})$ for the constant function $\bar{v}=\sqrt{\frac{\lambda}{|\K|}}$, defined on $\K$, which belongs to the minimizing domain.

To prove \eqref{K_eff_model_apriori_upper_bound}, let $K$ be such that $0<K< -J_{\R^3,\lambda}$.
Fix $f\in C^\infty_c(\R^3)$ such that $K=-{\mathscr J}_{\R^3}(f)>0$. Such a $f$ exists since $J_{\R^3,\lambda}<0$ and $C^\infty_c(\R^3)$ is dense in $H^1(\R^3)$. Thus, there exists $c_*>0$ such that for any $c\geq c_*$, the support of $f_c:=c^{3/2}f(c\cdot)$ is strictly included in $\K$. This implies, for any $c\geq c_*$, that
\begin{align*}
J_{\K,\lambda}(c)\leq{\mathscr J}_{\K,c}(f_c)&=\int_{\R^3}{|\nabla f_c|^2}+\frac35c_{TF}\int_{\R^3}{|f_c|^{\frac{10}3}}-\frac34c \int_{\R^3}{|f_c|^{\frac83}}=c^2{\mathscr J}_{\R^3}(f),
\end{align*}
and this concludes the proof of~\clm{K_models_apriori_estimates_J}.
\end{proof}

We introduce the notation $\K_c$ which will be the dilation of $\K$ by a factor $c>0$. Namely, if $\K$ is the unit cube, then
\begin{equation}\label{Kc_definition}
\K_c:=c\cdot\K:=\left[-\frac{c}2;\frac{c}2\right)^3.
\end{equation}
Moreover, we use the notation $\breve{v}$ to denote the dilation of $v$: for any $v$ defined on $\K$, $\breve{v}$ is defined on $\K_c$ by $\breve{v}(x):=c^{-3/2}v(c^{-1} x)$.

A direct computation gives
\begin{equation}\label{dilatation_of_J}
{\mathscr J}_{\K,c}(v)=c^2{\mathscr J}_{\K_c,1}(\breve{v}),
\end{equation}
for any $v\in H^1_{\textrm{per}}(\K)$. Consequently, $J_{\K,\lambda}(c)=c^2 J_{\K_c,\lambda}(1)$ and $v$ is a minimizer of $J_{\K,\lambda}(c)$ if and only if $\breve{v}$ is a minimizer of $J_{\K_c,\lambda}(1)$.
Finally, when $v$ is a minimizer of $J_{\K,\lambda}(c)$, we have some a priori bounds on several norms of $\breve{v}$ which are given in the following corollary of~\clm{K_models_apriori_estimates_J}.
\begin{cor}[Uniform norm bounds on minimizers of $J_{\K_c,\lambda}(1)$]\label{K_eff_dilated_minimizers_norm_unif_bound}
Let $\K$ be the unit cube and $\lambda$ be positive. Then there exist $C>0$ and $c_*>0$ such that for any $c\geq c_*$, a minimizer $\breve{v}_c$ of $J_{\K_c,\lambda}(1)$ verifies
$$\frac1C\leq \norm{\nabla \breve{v}_c}_{L^2(\K_c)}, \norm{\breve{v}_c}_{L^{10/3}(\K_c)}, \norm{\breve{v}_c}_{L^{8/3}(\K_c)}\leq C.$$
\end{cor}
\begin{proof}[Proof of~\ccr{K_eff_dilated_minimizers_norm_unif_bound}]
By~\eqref{K_eff_model_apriori_bounds} and~\eqref{K_eff_model_lower_bound_on_NRJ}, we obtain for $c$ large enough that any any minimizer $v_c$ of $J_{\K,\lambda}(c)$ verifies
$$\norm{\nabla\breve{v}_c}_{L^2(\K_c)}^2=c^{-2}\norm{\nabla v_c}^2_{L^2(\K)}\leq \frac{15}{64}\frac{\lambda}{c_{TF}}.$$
Applying, on $\K$, H\"older's inequality and Sobolev embeddings to $v_c$, we obtain that there exists $C$ such that
$$\forall c\geq c_*, \quad \norm{\nabla \breve{v}_c}_{L^2(\K_c)}, \norm{\breve{v}_c}_{L^{10/3}(\K_c)}, \norm{\breve{v}_c}_{L^{8/3}(\K_c)}\leq C.$$
By~\eqref{K_eff_model_apriori_upper_bound}, for any $K$ such that $0<K< -J_{\R^3,\lambda}$, there exists $c_\star>0$ such that
$$\forall c\geq c_\star, \; \; 0<\frac43 K\leq -\frac43 J_{\K_c,\lambda}(1) \leq \norm{\breve{v}_c}_{L^{8/3}(\K_c)}^{8/3}$$ and, consequently, such that
$$\forall c\geq c_\star, \; \; \norm{\breve{v}_c}_{L^{10/3}(\K_c)}^{10/3}\geq \frac1{\lambda}\left(\norm{\breve{v}_c}_{L^{8/3}(\K_c)}^{8/3}\right)^2>\frac{16}9\frac{K^2}\lambda>0.$$
We then obtain the lower bound for the gradient by the Sobolev embeddings. This concludes the proof of~\ccr{K_eff_dilated_minimizers_norm_unif_bound}.
\end{proof}

\subsection{Concentration-compactness}\label{section_symm_break_concentration_compactness}
To prove the symmetry breaking stated in \cth{main_result}, we prove the following result using the concentration-compactness method as a key ingredient.
\begin{prop}\label{Main_term_expansion_of_minimum}
Let $\K$ be the unit cube and $\lambda$ be positive. Then
$$\lim\limits_{c\to\infty} {c}^{-2} E_{\K,\lambda}(c)=J_{\R^3,\lambda}=\lim\limits_{c\to\infty} {c}^{-2} J_{\K,\lambda}(c).$$
Moreover, for any sequence ${w}_c$ of minimizers to $E_{\K,\lambda}(c)$, there exists a subsequence $c_n\to\infty$ and a sequence translations $\{x_n\}\subset\R^3$ such that the sequence of dilated functions $\breve{w}_n:={c_n}^{-3/2}{w}_{c_n}({c_n}^{-1}\cdot)$ verifies
\begin{enumerate}[label=\roman*.]
	\item $\mathds{1}_{\K_{c_n}}\breve{w}_n(\cdot+x_n)$ converges to a minimizer $u$ of $J_{\R^3,\lambda}$ strongly in $L^p(\R^3)$ for $2\leq p<6$, as $n$ goes to infinity;
	\item $\mathds{1}_{\K_{c_n}}\nabla\breve{w}_n(\cdot+x_n)\to \nabla u$ strongly in $L^2(\R^3)$.
\end{enumerate}
The same holds for any sequence ${v}_c$ of minimizers of $J_{\K,\lambda}(c)$.
\end{prop}
Before proving~\cpr{Main_term_expansion_of_minimum}, we give and prove several intermediate results, the first of which is the following proposition which will allow us to deduce the results for $E_{\K,\lambda}$ from those for $J_{\K,\lambda}$.
\begin{lemme}\label{cvgce_EK_JK}
Let $\lambda>0$. Then $$\frac{E_{\K,\lambda}(c)}{J_{\K,\lambda}(c)}\underset{c\to\infty}{\longrightarrow}1.$$
\end{lemme}
\begin{proof}[Proof of~\clm{cvgce_EK_JK}]
Let $w_c$ and $v_c$ be minimizers of $E_{\K,\lambda}(c)$ and $J_{\K,\lambda}(c)$ respectively which exist by~\cpr{K_complete_model_EulerLagrange} and the equivalent result for $J_{\K,\lambda}(c)$. Thus
$$\frac12 D_\K({w_c}^2,{w_c}^2)-\int_\K G_\K {w_c}^2\leq E_{\K,\lambda}(c)-J_{\K,\lambda}(c)\leq \frac12 D_\K({v_c}^2,{v_c}^2)-\int_\K G_\K {v_c}^2.$$
By the Hardy inequality on $\K$ and \eqref{equation_G_K_in_Lp}, we have
$$\left|\int_{\K}{G_\K {v_c}^2}\right|\leq \lambda\norm{G_\K v_c}_{L^2(\K)}\leq C\lambda\norm{v_c}_{H^1(\K)}$$
and similarly $\left|\int_{\K}{G_\K {w_c}^2}\right|\lesssim \norm{w_c}_{H^1(\K)}$. Moreover, we claim that
\begin{equation}\label{HLS_on_D_K}
D_\K({v_c}^2,{v_c}^2)\lesssim\norm{v_c}_{H^1(\K)}.
\end{equation}
To prove \eqref{HLS_on_D_K} we define, for each spatial direction $i\in\{1,2,3\}$ of the lattice, the intervals $I^{(-1)}_i:=\left[-1;-1/2\right)$, $I^{(0)}_i:=\left[-1/2;1/2\right)$ and $I^{(+1)}_i:=\left[1/2;1\right)$, and the parallelepipeds $\K^{(\sigma_1,\sigma_2,\sigma_3)}=I^{(\sigma_1)}_1\times I^{(\sigma_2)}_2 \times I^{(\sigma_3)}_3$ which let us rewrite $\K=\K^{(0,0,0)}$ and $\K_2=2\cdot\K:=\left[-1;1\right)^3$ as the union of the $27$ sets
$$\K_2=\bigcup\limits_{\boldsymbol\sigma\in\{-1;0;+1\}^3} \K^{\boldsymbol\sigma}.$$
We thus have by
\eqref{equation_G_K_in_Lp} and the Hardy--Littlewood--Sobolev inequality that
\begin{align*}
\iint\limits_{\substack{\K\times \K\\x-y\in \K^{\boldsymbol\sigma}}}{v_c}^2(x)G_\K(x-y){v_c}^2(y)\dd x\dd y&\lesssim\iint_{\K\times \K}{\frac{v_c^2(x)v_c^2(y)}{|x-y-\boldsymbol\sigma|}\dd y\dd x} \lesssim \norm{v_c}_{L^{\frac{12}5}(\K)}^4.
\end{align*}
Consequently, by H\"older's inequality and Sobolev embeddings, we have
\begin{multline}\label{D_K_HLS}
\left|D_\K({v_c}^2,{v_c}^2)\right|=\Big|\sum_{\boldsymbol\sigma\in\{-1;0;+1\}^3} \iint\limits_{\substack{\K\times \K\\x-y\in \K^{\boldsymbol\sigma}}}{v_c}^2(x)G_\K(x-y){v_c}^2(y)\dd x\dd y\Big|\\
	\lesssim \norm{v_c}_{L^{\frac{12}5}(\K)}^4\lesssim \norm{v_c}_{H^1(\K)}\norm{v_c}_{L^2(\K)}^3.
\end{multline}
This proves \eqref{HLS_on_D_K} which also holds for $w_c$.

Then, on one hand, by~\eqref{maj_NRJ_Kin_by_fonctional} applied to $c_1=0\leq c_2=c$, there exist positive constants $a<1$ and $C$ such that for any $c>0$ we have
\begin{equation*}
a\norm{\nabla w_c}_{L^2(\K)}^2\leq  \frac{15}{64}\frac{\lambda}{c_{TF}}{c}^2+E_{\K,\lambda}(0) +\lambda C.
\end{equation*}
On the other hand, the upper bound in~\eqref{K_eff_model_apriori_upper_bound} together with the \eqref{K_eff_model_lower_bound_on_NRJ} applied to $v_c$, give that there exists $c_*>0$ such that
\begin{align}\label{maj_norm_grad_jD}
\exists \; K>0, \forall \; c\geq c_*, \qquad \norm{\nabla{v_c}}^2_{L^2(\K)} &\leq \left(\frac{15}{64}\frac{\lambda}{c_{TF}}-K\right){c}^2.
\end{align}
Consequently, for $c$ large enough, we have
\begin{equation*}
\left|J_{\K,\lambda}(c)-E_{\K,\lambda}(c)\right| \lesssim c
\end{equation*}
hence, using \eqref{K_eff_model_apriori_upper_bound}, we finally obtain
\begin{equation*}
\left|\frac{E_{\K,\lambda}(c)}{J_{\K,\lambda}(c)}-1\right| \lesssim {c}^{-1}.
\end{equation*}
This concludes the proof of~\clm{cvgce_EK_JK}.
\end{proof}

We now prove that the periodic effective model converges,
$$\lim\limits_{c\to\infty} {c}^{-2} J_{\K,\lambda}(c)=J_{\R^3,\lambda},$$
by proving the two associated inequalities. We first prove the upper bound then use the concentration-compactness method to prove the converse inequality. The strong convergence of minimizers stated in~\cpr{Main_term_expansion_of_minimum} will be a by-product of the method.

\begin{lemme}[Upper bound]\label{R3_eff_model_main_term_expansion_upper_bound}
Let $\K$ be the unit cube and $\lambda$ be positive. Then there exists $\beta>0$ such that
\begin{equation*}\label{K_eff_model_below_expansion_exponential_rate_equation}
J_{\K,\lambda}(c)\leq c^2 J_{\R^3}(\lambda)+o(e^{-\beta c}).
\end{equation*}
\end{lemme}
\begin{proof}[Proof of~\clm{R3_eff_model_main_term_expansion_upper_bound}]
Using the scaling equality \eqref{dilatation_of_J}, this result is obtained by computing ${\mathscr J}_{\K_c,1}(Q_c)$ where
$$Q_c = \frac{\sqrt{\lambda} \chi_c Q}{\norm{\chi_c Q}_{L^2(\R^3)}},$$
for $Q\in H^1(\R^3)$ a minimizer of $J_{\R^3,\lambda}$, with $\chi_c\in C^\infty_c(\R^3)$, $0\leq\chi_c\leq1$, $\chi_c\equiv0$ on $\R^3\setminus \K_{c+1}$, $\chi_c\equiv1$ on $\K_c$ and $\norm{\nabla \chi_c}_{L^\infty(\R^3)}$ bounded. Indeed, by the well-known exponential decay of continuous positive solution to the Euler--Lagrange equations with strictly negative Lagrange multiplier, one obtains the exponential decay when $r$ goes to infinity of the norm $\norm{\nabla Q}_{L^2({}^\complement\!B(0,r))}$ and the norms $\norm{Q}_{L^p({}^\complement\!B(0,r))}$ for $p>0$, and consequently the claimed upper bound.
\end{proof}

\begin{lemme}[Lower bound]\label{R3_eff_model_main_term_expansion_lower_bound}
Let $\K$ be the unit cube and $\lambda$ be positive. Then
$$\liminf\limits_{c\to\infty} c^{-2}J_{\K,\lambda}(c)\geq J_{\R^3,\lambda}.$$
\end{lemme}
\begin{proof}[Sketch of proof of~\clm{R3_eff_model_main_term_expansion_lower_bound}]
See \cite{Ricaud-PhD} for a detailed proof. This result relies on Lions' con\-cen\-tra\-tion-compactness method and on the following result. Since this lemma is well-known, we omit its proof. Similar statements can be found for example in \cite{Gerard-98, BahGer-99,HmiKer-05,KilVis-08,Lewin-VMQM,Ricaud-PhD}.

\begin{lemme}[Splitting in localized bubbles]\label{Splitting_localized_bubbles}
Let $\K$ be the unit cube, $\{\phi_c\}_{c\geq1}$ be a sequence of functions such that $\phi_c\in H^1_{\textrm{per}}(\K_c)$ for all $c$, with $\norm{\phi_c}_{H^1(\K_c)}$ uniformly bounded. Then there exists a sequence of functions $\{\phi^{(1)}, \phi^{(2)}, \cdots\}$ in $H^1(\R^3)$ such that the following holds. For any $\epsilon>0$ and any fixed sequence $0\leq R_k\to\infty$, there exist: $J\geq0$, a subsequence $\{\phi_{c_k}\}$, sequences $\{\xi_k^{(1)}\}, \cdots, \{\xi_k^{(J)}\}, \{\psi_k\}$ in $H^1_{\textrm{per}}(\K_{c_k})$ and sequences of space translations $\{x_k^{(1)}\}, \cdots, \{x_k^{(J)}\}$ in $\R^3$ such that
$$\lim\limits_{k\to\infty}\Big\vert\kern-0.25ex\Big\vert{\phi_{c_k} - \sum\limits_{j=1}^J \xi_k^{(j)}(\cdot-x_k^{(j)}) - \psi_{k}}\Big\vert\kern-0.25ex\Big\vert_{H^1(\K_{c_k})}=0,$$
where
\begin{itemize}[label=$\bullet$]
	\item $\{\xi_k^{(1)}\}, \cdots, \{\xi_k^{(J)}\}, \{\psi_k\}$ have uniformly bounded $H^1(\K_{c_k})$-norms,
	\item $\mathds{1}_{\K_{c_k}}\xi_k^{(j)}\wto \phi^{(j)}$ weakly in $H^1(\R^3)$ and strongly in $L^p(\R^3)$ for $2\leq p<6$,
	\item $\textrm{supp}(\mathds{1}_{\K_{c_k}}\xi_k^{(j)})\subset B(0,R_k)$ for all $j=1,\cdots,J$ and all $k$,
	\item $\textrm{supp}(\mathds{1}_{\K_{c_k}}\psi_k)\subset \K_{c_k}\setminus \bigcup\limits_{j=1}^J B(x_k^{(j)},2R_k)$ for all $k$,
	\item $|x_k^{(i)}-x_k^{(j)}|\geq5R_k$ for all $i\neq j$ and all $k$,
	\item $\int_{\K_{c_k}}|\psi_{k}|^\frac83\leq\epsilon$.
\end{itemize}
\end{lemme}
\begin{rmq*}
The proof of \clm{R3_eff_model_main_term_expansion_lower_bound} relies on the concentration-compactness method. Extracting only one bubble ($J=1$) by a localization method would not allow us to conclude since we have little information on the energy of the remainder $\psi_{k}$. In similar proofs in the literature, it is often possible to conclude after extracting few bubbles, using that ${\mathscr J}(\psi_k)\geq J(\int|\psi_k|^2)$. In our case, $J_{\K_c}(\int|\psi_k|^2)$ depends on $c$ hence the same inequality of course holds but does not allow us to conclude. We therefore need to extract as many bubbles as necessary such as to sufficiently decrease the energy of $\psi_{k}$.
\end{rmq*}
We apply \clm{Splitting_localized_bubbles} to the sequence $\left(\breve{v}_c\right)_{c\geq1}$ of minimizers to $J_{\K_c,\lambda}(1)$ which verifies the hypothesis by the upper bound proved in \ccr{K_eff_dilated_minimizers_norm_unif_bound}. The lower bound in that corollary excludes the case $J=0$. Indeed, in that case we would have $\lim\limits_{k\to\infty}\norm{\phi_{c_k}- \psi_{k}}_{H^1(\K_{c_k})}=0$ and $\int_{\K_{c_k}}|\psi_{k}|^\frac83\leq\epsilon$ hence $\int_{\K_{c_k}}|\phi_{k}|^\frac83\leq2\epsilon$, for $k$ large enough, contradicting the mentioned lower bound. Consequently, there exists $J\geq1$ such that
$$\breve{v}_{c_k}=\psi_k +\epsilon_k + \sum\limits_{j=1}^J \breve{v}_k^{(j)}(\cdot-x_k^{(j)})$$
where $\norm{\epsilon_k}_{H^1(\K_{c_k})}\to0$ and, for a each $k$, the supports of the $\breve{v}_k^{(j)}(\cdot-x_k^{(j)})$'s and $\psi_k$ are pairwise disjoint. The support properties, the Minkowski inequality, Sobolev embeddings and the fact that $\textrm{supp}(\mathds{1}_{\K_{c_k}}\breve{v}_k^{(j)})\subset B(0,R_k)\subset \K_{c_k}$, give that
\begin{align*}
J_{\K_{c_k}}(\lambda)={\mathscr J}_{\K_{c_k}}(\breve{v}_{c_k})&={\mathscr J}_{\K_{c_k}}(\psi_k)+\sum\limits_{j=1}^J {\mathscr J}_{\R^3}(\mathds{1}_{\K_{c_k}} \breve{v}_k^{(j)})+o(1)_{c_k\to\infty}\\
	&\geq -\frac34\epsilon+\sum\limits_{j=1}^J {\mathscr J}_{\R^3}(\mathds{1}_{\K_{c_k}} \breve{v}_k^{(j)})+o(1)_{c_k\to\infty}.
\end{align*}
Moreover, the strong convergence of $\mathds{1}_{\K_{c_k}} \breve{v}_k^{(j)}$ in $L^2$ and the continuity of $\lambda\mapsto J_{\R^3,\lambda}$, proved in~\clm{R3_eff_model_apriori_properties_J}, imply, for all $j=1,\cdots,J$, that
$${\mathscr J}_{\R^3}(\mathds{1}_{\K_{c_k}} \breve{v}_k^{(j)})\geq J_{\R^3}\Big(\normSM{\breve{v}_k^{(j)}}^2_{L^2(\K_{c_k})}\Big)\underset{k\to\infty}{\longrightarrow}J_{\R^3}(\lambda^{(j)}),$$
where, for any $j$, $\lambda^{(j)}:=\normSM{\breve{v}^{(j)}}_{L^2(\R^3)}$ is the mass of the limit of $\mathds{1}_{\K_{c_k}} \breve{v}_k^{(j)}$. We also have denoted $J_{\R^3}(\lambda):=J_{\R^3,\lambda}$ to simplify notations here.
Those inequalities together with the strict binding proved in~\cpr{R3_eff_model_strict_binding} lead to
\begin{align*}
\frac34 \epsilon+\liminf\limits_{k\to\infty} J_{\K_{c_k}}(\lambda)&\geq\sum\limits_{j=1}^J J_{\R^3}(\lambda^{(j)})>J_{\R^3}(\lambda) - J_{\R^3}\Big(\lambda-\sum\limits_{j=1}^J\lambda^{(j)}\Big)\geq J_{\R^3}(\lambda).
\end{align*}
The last inequality comes from the fact that $0\leq\norm{\psi_k}^2_{L^2(\K_{c_k})}=\lambda-\sum\limits_{j=1}^J\lambda^{(j)}+o(1)$ thus $\lambda-\sum\limits_{j=1}^J\lambda^{(j)}\geq0$ and this implies that $J_{\R^3}\Big(\lambda-\sum\limits_{j=1}^J\lambda^{(j)}\Big)\leq0$. This concludes the proof of~\clm{R3_eff_model_main_term_expansion_lower_bound}.
\end{proof}

We can now compute the main term of $E_{\K,\lambda}(c)$ stated in~\cpr{Main_term_expansion_of_minimum}.
\begin{proof}[Proof of~\cpr{Main_term_expansion_of_minimum}]
Propositions~\ref{R3_eff_model_main_term_expansion_upper_bound} and~\ref{R3_eff_model_main_term_expansion_lower_bound} give, for $\lambda>0$, the limit
$$\lim\limits_{c\to\infty} c^{-2}J_{\K,\lambda}(c)= J_{\R^3,\lambda}$$
and \clm{cvgce_EK_JK} gives then the same limit for $E_{\K,\lambda}(c)$. \cpr{R3_eff_model_main_term_expansion_lower_bound} also gives that $\left(\breve{v}_c\right)_{c\geq1}$ has at least a first extracted bubble $0\nequiv\breve{v}\in H^1(\R^3)$ to which $\mathds{1}_{\K_{c_k}}\breve{v}_{c_k}(\cdot+x_k)$ converges weakly in $L^2(\R^3)$. This leads to
\begin{equation}\label{Proof_Main_term_expansion_of_minimum_J_expansion}
J_{\K_{c_k},\lambda}(1)={\mathscr J}_{\K_{c_k},1}(\breve{v}_{c_k}(\cdot+x_k))={\mathscr J}_{\R^3}(\breve{v})+{\mathscr J}_{\K_{c_k},1}(\breve{v}_{c_k}(\cdot+x_k)-\breve{v})+o(1)
\end{equation}
by the following lemma.
\begin{lemme}\label{K_eff_dilated_strongconvergence_results_of_weakly_convergent_sequence}
Let $\K$ be the unit cube and $\{\phi_c\}_{c\geq1}$ be a sequence of functions on $\R^3$ with $\norm{\phi_c}_{H^1(\K_c)}$ uniformly bounded such that $\mathds{1}_{\K_c}\phi_c \underset{c\to\infty}{\wto} \phi$ weakly in $L^2(\R^3)$. Then $\phi \in H^1(\R^3)$ and, up to the extraction of a subsequence, we have
\begin{enumerate}[label=\roman*.]
	\item $\mathds{1}_{\K_c}\nabla \phi_c\wto \nabla \phi$ weakly in $L^2(\R^3)$,
	\item $\norm{\nabla(\phi_c- \phi) }_{L^2(\K_c)}^2=\norm{\nabla \phi_c}_{L^2(\K_c)}^2 - \norm{\nabla \phi}_{L^2(\R^3)}^2+\underset{c\to\infty}{o}(1)$,
	\item $\norm{\phi_c- \phi }_{L^p(\K_c)}^p=\norm{\phi_c}_{L^p(\K_c)}^p - \norm{\phi}_{L^p(\R^3)}^p+\underset{c\to\infty}{o}(1), \textrm{ for } 2\leq p\leq6.$
\end{enumerate}
\end{lemme}
\begin{proof}[Proof of~\clm{K_eff_dilated_strongconvergence_results_of_weakly_convergent_sequence}]
By the mean of a regularization function (as in the proof of \clm{R3_eff_model_main_term_expansion_upper_bound}) together with the uniform boundedness of $\phi_c$ in $H^1(\K_c)$ and the uniqueness of the limit, one obtains that the limit $\phi$ is in $H^1(\R^3)$. Since \emph{i.} is a classical result and \emph{ii.} a direct consequence of it, we only prove here \emph{iii.}.

The weak convergence in $L^2(\R^3)$ of $\mathds{1}_{\K_c}\nabla \phi_c$ gives the convergence a.e. of $\phi_c$ to $\phi$, up to a subsequence, by \cite[Corollary 8.7]{LieLos-01}. Since $|\phi_c- \phi|$ is uniformly bounded in $L^2(\R^3)\cap L^6(\R^3)$, this implies \emph{iii.} by the \emph{Missing  term in Fatou's lemma} Theorem (see \cite[Theorem 1.9]{LieLos-01}).
\end{proof}

To obtain for $E_{\K,\lambda}(c)$ an expansion similar to \eqref{Proof_Main_term_expansion_of_minimum_J_expansion}, we proceed the same way. We first show that the sequence of minimizers $\breve{w}_c$ is uniformly bounded in $H^1_{\textrm{per}}(\K_c)$ using the upper bound in the following lemma, which is equivalent to \ccr{K_eff_dilated_minimizers_norm_unif_bound} for $\breve{v}_c$.
\begin{lemme}[Uniform norm bounds on minimizers of $E_{\K,\lambda}(c)$]\label{K_complete_dilated_model_minimizers_norm_unif_bound}
Let $\K$ be the unit cube, $\lambda, c_{TF}$ and $c$ be positive. Then there exist $C>0$ and $c_*>0$ such that for any $c\geq c_*$, the dilation $\breve{w}_c(x):=c^{-3/2}w_c(c^{-1} x)$ of a minimizer $w_c$ to $E_{\K,\lambda}(c)$ verifies
$$\frac1C\leq \norm{\nabla \breve{w}_c}_{L^2(\K_c)}, \norm{\breve{w}_c}_{L^{10/3}(\K_c)}, \norm{\breve{w}_c}_{L^{8/3}(\K_c)}\leq C.$$
\end{lemme}
\begin{proof}[Proof of~\clm{K_complete_dilated_model_minimizers_norm_unif_bound}]
As seen in the proof of~\clm{cvgce_EK_JK}, $\norm{\nabla w_c}_{L^2(\K)}=O(c)$ hence
$$\norm{\nabla\breve{w}_c}_{L^2(\K_c)}^2={c}^{-2}\norm{\nabla w_c}_{L^2(\K)}^2=O(1)$$
and, using Sobolev embeddings for the two other norms, we have
$$\forall c\geq c_*, \quad \norm{\nabla \breve{w}_c}_{L^2(\K_c)}, \norm{\breve{w}_c}_{L^{10/3}(\K_c)}, \norm{\breve{w}_c}_{L^{8/3}(\K_c)}\leq C'.$$

Let $K$ be such that $0<K< -J_{\R^3,\lambda}$ and $\epsilon>0$, then by~\eqref{K_eff_model_apriori_upper_bound} and~\clm{cvgce_EK_JK}, there exists $C>0$ such that
$${c}^2K - \epsilon\leq -J_{\K,\lambda}(c) - \epsilon \leq -E_{\K,\lambda}(c)\leq c\left(C+\frac34\norm{w_c}^{\frac83}_{L^{\frac83}(\K)}\right)$$
for $c$'s large enough and, consequently that
$$K-\frac{C+\epsilon}{{c}^2}\leq \frac34\norm{\breve{w}_c}^{8/3}_{L^{8/3}(\K_c)}.$$
We conclude this proof of~\clm{K_complete_dilated_model_minimizers_norm_unif_bound} as we did in the proof of~\ccr{K_eff_dilated_minimizers_norm_unif_bound}.
\end{proof}
We now come back to the proof of \cpr{Main_term_expansion_of_minimum}. We apply \clm{Splitting_localized_bubbles} to $\{\breve{w}_c\}$ and, as for $\breve{v}_c$, the lower bound in \clm{K_complete_dilated_model_minimizers_norm_unif_bound} implies that $J\geq1$, namely that there exist at least a first extracted bubble $0\nequiv\breve{w}\in H^1(\R^3)$ such that $\mathds{1}_{\K_{c_k}}\breve{w}_{c_k}(\cdot+y_k)\wto \breve{w}$ weakly in $L^2(\R^3)$.
\clm{K_eff_dilated_strongconvergence_results_of_weakly_convergent_sequence} then leads to
\begin{align*}
{c_k}^{-2}E_{\K,\lambda}(c_k)&={\mathscr J}_{\K_{c_k},1}(\breve{w}_{c_k}(\cdot+y_k))+O({c_k}^{-1})\\
	&={\mathscr J}_{\R^3}(\breve{w})+{\mathscr J}_{\K_{c_k},1}(\breve{w}_{c_k}(\cdot+y_k)-\breve{w})+o(1),
\end{align*}
where the term $O(c^{-1})$ comes from $D_\K({w_c}^2,{w_c}^2)=O(c)$ and $\int_{\K}{G_\K {w_c}^2}=O(c)$ obtained in the proof of \clm{cvgce_EK_JK}.

Since in both cases $J$ and $E$, the left hand side converges to $J_{\R^3}(\lambda)$, the end of the argument will be the same and we will therefore only write it in the case of $E$. Defining $\lambda_1:=\norm{\breve{w}}_{L^2(\R^3)}^2$, which is positive since $\breve{w}\nequiv0$, we thus have
$${c_k}^{-2}E_{\K,\lambda}(c_k)\geq J_{\R^3}(\lambda_1)+J_{\K_{c_k}}\big(\normSM{\breve{w}_{c_k}(\cdot+y_k)-\breve{w}}^2_{L^2(\K_{c_k})}\big)+o(1).$$
Since $\normSM{\breve{w}_c(\cdot+y_k)-\breve{w}}^2_{L^2(\K_c)}= \lambda-\lambda_1 +o(1)$, then for any $\epsilon>0$, we have
$${c_k}^{-2}E_{\K,\lambda}(c_k)\geq J_{\R^3}(\lambda_1)+J_{\K_{c_k}}(\lambda-\lambda_1+\epsilon)+o(1),$$
By the convergence of ${c}^{-2}E_{\K,\lambda}(c)$ for any $\lambda>0$, this leads to
$$J_{\R^3}(\lambda)\geq J_{\R^3}(\lambda_1)+J_{\R^3}(\lambda-\lambda_1+\epsilon)$$
and, sending $\epsilon$ to $0$, the continuity of $\lambda\mapsto J_{\R^3}(\lambda)$, proved in~\clm{R3_eff_model_apriori_properties_J}, gives
$$J_{\R^3}(\lambda)\geq J_{\R^3}(\lambda_1)+J_{\R^3}(\lambda-\lambda_1).$$
We recall that $\lambda_1>0$ hence, if $\lambda_1<\lambda$ then the above large inequality would contradict the strict binding proved in~\cpr{R3_eff_model_strict_binding}, hence $\lambda_1=\lambda$. This convergence of the norms combined with the original weak convergence in $L^2(\R^3)$ gives the strong convergence in $L^2(\R^3)$ of $\mathds{1}_{\K_c}\breve{w}_c(\cdot+y_k)$ to $\breve{w}$ hence in $L^p(\R^3)$ for $2\leq p<6$ by H\"older's inequality, Sobolev embeddings and the facts that $\breve{w}_c$ is uniformly bounded in $H^1_{\textrm{per}}(\K_c)$ and that $\breve{w}\in H^1(\R^3)$. The strong convergence holds in particular in $L^{\frac83}(\R^3)$ thus we have proved that $\breve{w}$ is the first and only bubble.

Finally, for any $\epsilon>0$, we now have, for $k$ large enough, that
\begin{align*}
{c_k}^{-2}E_{\K,\lambda}(c_k)&={\mathscr J}_{\R^3}(\breve{w})+{\mathscr J}_{\K_{c_k},1}(\breve{w}_{c_k}(\cdot+y_k)-\breve{w})+o(1)\\
	&\geq {\mathscr J}_{\R^3}(\breve{w})+J_{\K_{c_k}}(\epsilon)+o(1).
\end{align*}
This leads to $J_{\R^3}(\lambda)\geq {\mathscr J}_{\R^3}(\breve{w})+J_{\R^3}(\epsilon)$, then to $J_{\R^3}(\lambda)\geq {\mathscr J}_{\R^3}(\breve{w})$ by the continuity of $J_{\R^3}(\lambda)$ proved in \clm{R3_eff_model_apriori_properties_J}. Since $\norm{\breve{w}}^2_{L^2(\R^3)}=\lambda$, this concludes the proof of~\cpr{Main_term_expansion_of_minimum} up to the convergence of $\mathds{1}_{\K_{c_n}}\nabla\breve{w}_n(\cdot+x_n)$ and $\mathds{1}_{\K_{c_n}}\nabla\breve{v}_n(\cdot+x_n)$ that we deduce now from the above results. Indeed, by the convergence in $L^p(\R^3)$ of $\breve{w}_n(\cdot+x_n)$ and since $\left|\int_{\K}{G_\K{w_n}^2}\right|+\left|D_\K({w_n}^2,{w_n}^2)\right|=O(c_n)$, we know, except for the gradient term, that all terms of ${c_n}^{-2}{E}_{\K,\lambda}(c_n)$ (resp. ${c_n}^{-2}{J}_{\K,\lambda}(c_n)$) converge thus the gradient term too. Then we apply~\clm{K_eff_dilated_strongconvergence_results_of_weakly_convergent_sequence} to obtain the strong convergence in $L^2(\R^3)$ from this convergence in norm just obtained.
\end{proof}

Let us emphasize that all the results stated in this section still hold true in the case of several charges per cell (for example for the union $N\cdot\K$) with same proofs. The modifications only come from the factor $\int_{\K}{G_\K{w_c}^2}$ being replaced by $\int_\K{\sum_{i=1}^{N_q} G_\K(\cdot-R_i)|w_c|^2}$ --- see~\eqref{K_NRJ_u_on_NK} --- therefore only the proofs of~\cpr{Main_term_expansion_of_minimum}, \clm{cvgce_EK_JK} and \clm{K_complete_dilated_model_minimizers_norm_unif_bound} are slightly changed by a factor $N_q$ in the bounds of the modified term, but their statement is unchanged. Consequently, as mentioned in Section~\ref{section_symmetry_breaking}, the results
$$\lim\limits_{c\to\infty} {c}^{-2} E_{N\cdot\K,N^3\lambda}(c)=J_{\R^3,N^3\lambda} \qquad \textrm{ and } \qquad \lim\limits_{c\to\infty} {c}^{-2} E_{\K,\lambda}(c)=J_{\R^3,\lambda}$$
from~\cpr{Main_term_expansion_of_minimum} and the result
$$J_{\R^3}(N^3\lambda)<N^3J_{\R^3}(\lambda)$$
from~\cpr{R3_eff_model_strict_binding} imply together the symmetry breaking
$$E_{N\cdot\K,N^3\lambda}(c)<N^3 E_{\K,\lambda}(c).$$

We now give two corollaries of \cpr{Main_term_expansion_of_minimum}. We state and prove them in the case of one charge per unit cell but they hold, with similar proof, for several charges.
\begin{cor}[Convergence of Euler--Lagrange multiplier]\label{convergence_EulerLagrange_multiplier}
Let $\{w_c\}$ be a sequence of minimizers to ${E}_{\K,\lambda}(c)$ and $\{\mu_c\}$ the sequence of associated Euler--Lagrange multipliers, as in~\cpr{K_complete_model_EulerLagrange}. Then there exists a subsequence $c_n\to\infty$ such that
$${c_n}^{-2}\mu_{c_n}\underset{n\to\infty}{\longrightarrow}\mu_{\R^3, \{w_{c_n}\}}$$
with $\mu_{\R^3, \{w_{c_n}\}}\!$ the Euler--Lagrange multiplier associated with the minimizer to $J_{\R^3}(\lambda)$ to which the subsequence of dilated functions $\mathds{1}_{\K_{c_n}}\breve{w}_{c_n}(\cdot+x_n)$ converges strongly.

The same holds for sequences $\{v_c\}$ of Euler--Lagrange multipliers associated with minimizers to $ J_{\K,\lambda}(c)$.
\end{cor}
\begin{proof}[Proof of~\ccr{convergence_EulerLagrange_multiplier}]
Let $u$ be the minimizer of $J_{\R^3}(\lambda)$ to which $\mathds{1}_{\K_{c_n}}\breve{w}_{c_n}(\cdot+x_n)$ converges strongly in $L^p(\R^3)$ for $2\leq p<6$, by~\cpr{Main_term_expansion_of_minimum} which also gives that $\mathds{1}_{\K_{c_n}}\nabla\breve{w}_{c_n}(\cdot+x_n)\to\nabla u$ strongly in $L^2(\R^3)$, and $\mu_{\R^3, u}$ the Euler--Lagrange multiplier associated with this $u$ by~\cth{R3_eff_model_existence_thm}. 

By~\clm{K_complete_dilated_model_minimizers_norm_unif_bound} and the formula~\eqref{K_complete_model_EulerLagrange_mu_formulae} giving an expression of $\mu_c$, we then obtain
$$-{c_n}^{-2}\mu_{c_n}\lambda\to \norm{\nabla u}_{L^2(\R^3)}^2+c_{TF}\norm{u}^{10/3}_{L^{10/3}(\R^3)}-\norm{u}^{8/3}_{L^{8/3}(\R^3)}.$$
Therefore, by~\eqref{R3_eff_model_EulerLagrange_mu_formulae} which gives an expression of the Euler--Lagrange parameter $\mu_{\R^3, u}$ associated with this $u$, we have
$${c_n}^{-2}\mu_{c_n}\underset{c\to\infty}{\longrightarrow}\mu_{\R^3, u}.$$
Since $u$ depends on $\{w_{c_n}\}$, we can of course rename $\mu_{\R^3, \{w_{c_n}\}}:=\mu_{\R^3, u}$.
The result for $J_{\K,\lambda}(c)$ is proved the same way.
\end{proof}
\begin{lemme}[$L^\infty$-convergence]\label{K_models_Linfinity_convergence_minimizers}
Let $\{w_c\}_{c}$ be a sequence of minimizers to ${E}_{\K,\lambda}(c)$ and $u$ be the minimizer to $J_{\R^3}(\lambda)$ to which the subsequence of rescaled functions $\mathds{1}_{\K_{c_n}}\breve{w}_{c_n}(\cdot+x_n)$ converges. Then
$$\norm{\breve{w}_{c_n}(\cdot+x_n)-u}_{H^2(\K_{c_n})}\underset{n\to+\infty}{\longrightarrow}0 \quad \textrm{ and } \quad \norm{\mathds{1}_{\K_{c_n}}\breve{w}_{c_n}(\cdot+x_n)-u}_{L^\infty(\K_{c_n})}\underset{n\to+\infty}{\longrightarrow}0.$$
The same result holds for a sequence $\{v_c\}_{c}$ of minimizers to ${J}_{\K,\lambda}(c)$.
\end{lemme}
\begin{proof}[Proof of~\clm{K_models_Linfinity_convergence_minimizers}] 
For shortness, we omit the spatial translations $\{x_n\}$ in this proof. We define $u_c=\zeta_c u$ where $\zeta_c$ is a smooth function such that $0\leq\zeta_c\leq1$, $\zeta_c\equiv0$ on $\R^3\setminus \K_{c}$ and $\zeta_c\equiv1$ on $\K_{c-1}$. Since $u\in H^2(\R^3)$ by \cth{R3_eff_model_existence_thm} and $\norm{\zeta_c}_\infty+\norm{\nabla\zeta_c}_\infty+\norm{\Delta\zeta_c}_\infty<\infty$, we have to prove $\norm{\breve{w}_{c_n}-u_{c_n}}_{H^2(\K_{c_n})}\underset{n\to+\infty}{\longrightarrow}0$. Moreover, by the Rellich-Kato theorem, the operator $-\Delta_{\textrm{per}}-{c}^{-2} G_\K({c}^{-1}\cdot)$ is self-adjoint of domain $H^2_{\textrm{per}}(\K_c)$ and bounded below. Therefore, there exists $0<C<1$ such that, for any $\beta$ large enough and any $c\geq1$, we have
$$\norm{\breve{w}_{c}-u_{c}}_{H^2_{\textrm{per}}(\K_{c})}\leq C\norm{\left(-\Delta_{\textrm{per}}-c^{-2}G_\K({c}^{-1}\cdot)+\beta\right)(\breve{w}_c-u_c)}_{L^2_{\textrm{per}}(\K_{c})}.$$
Thus, denoting ${\mathscr C}_c^-:=\K_c\setminus\K_{c-1}$ and $\mu_{\R^3}$ the Euler--Lagrange parameter associated with $u$, we have by the Euler--Lagrange equations~\eqref{R3_eff_model_EulerLagrange} and~\eqref{K_complete_model_EulerLagrange_equation} that
\begin{align*}
&\norm{\breve{w}_{c}-u_{c}}_{H^2_{\textrm{per}}(\K_c)}\\
&\;\;\;\;\leq C c_{TF}\norm{{\zeta_c}^{\frac37}u-\breve{w}_c}_{L^4(\K_c)}\norm{{\zeta_c}^{\frac47}|u|^{\frac43}+|\breve{w}_c|^{\frac43}}_{L^4(\K_c)}+ \norm{u}_{L^2({\mathscr C}_{c}^-)}\norm{\Delta\zeta_c}_{L^\infty(\K_c)}\\
	&\;\;\;\;+ C \norm{{\zeta_c}^{\frac35}u-\breve{w}_c}_{L^4(\K_c)}\norm{{\zeta_c}^{\frac25}|u|^{\frac23}+|\breve{w}_c|^{\frac23}}_{L^4(\K_c)}+2\norm{\nabla\zeta_c}_{L^\infty(\K_c)}\norm{\nabla u}_{L^2({\mathscr C}_{c}^-)}\\
	&\;\;\;\;+ C |\mu_{\R^3} -{c}^{-2}\mu_c|\norm{\breve{w}_c}_{L^2(\K_c)} + C (\mu_{\R^3}+\beta)\norm{\zeta_c u-\breve{w}_c}_{L^2(\K_c)}\\
	&\;\;\;\;+ C {c}^{-2}\norm{G_\K({c}^{-1}\cdot)}_{L^{5/2}(\K_c)}\norm{u_c}_{L^{10}(\K_c)} + C {c}^{-2}\norm{|u_c|^2\star G_{\K}}_{L^\infty(\K)}\norm{\breve{w}_c}_{L^2(\K_c)},
\end{align*}
for any $c>0$. Therefore, combining that the $L^\infty(\K_{c})$ norms of $\zeta_c$ and of it derivatives are finite, that $\norm{\nabla u}_{L^2({\mathscr C}_c^-)} + \norm{u}_{L^2({\mathscr C}_c^-)}\to0$, that $c^{-2}\norm{G_\K({c}^{-1}\cdot)}_{L^{5/2}(\K_{c})}={c}^{-\frac45}\norm{G_\K}_{L^{5/2}(\K)}\to0$
and that, for any $\alpha>0$ and $2\leq p\leq6$, we have
$$\norm{{\zeta_{c_n}}^\alpha u-\breve{w}_{c_n}}_{L^p(\K_{c_n})}=\norm{(1-{\zeta_{c_n}}^\alpha) u}_{L^p(\K_{c_n})}+\norm{u-\breve{w}_{c_n}}_{L^p(\K_{c_n})}\to0,$$
all together with~\ccr{convergence_EulerLagrange_multiplier}, we conclude that
$$\norm{\breve{w}_{c_n}-u_{c_n}}_{H^2_{\textrm{per}}(\K_{c_n})}\underset{n\to+\infty}{\longrightarrow}0.$$

The proof for $v_c$ is similar but easier and shorter, we thus omit it.

We then conclude the proof of~\clm{K_models_Linfinity_convergence_minimizers} using that for any $c^*>0$, there exists $C$ such that for any $c\in[c^*; \infty)$ and $f\in H^2(\K_c)$, we have
$\norm{f}_{L^\infty(\K_c)}\leq C \norm{f}_{H^2(\K_c)}$ which can be proved by means of Fourier series.
\end{proof}

\subsection{Location of the concentration points}\label{section_minimizer_localization}
In this section we consider the union of $N^3$ cubes $\K$, each containing one charge $q=1$ --- that we can assume to be at the center of the cube $\K$ --- forming together the cube $\K_N:=N\cdot\K$. The energy of the unit cell $\K_N$ is then
\begin{equation}\label{K_NRJ_u_on_NK}
\begin{aligned}
{\mathscr E}_{\K_N,c}(w)={\mathscr J}_{\K_N,c}(w)+\frac12D_{\K_N}(|w|^2,|w|^2)-\int_{\K_N}{\sum\limits_{i=1}^{N^3} G_{\K_N}(\cdot-R_i)|w|^2},
\end{aligned}
\end{equation}
where $\{R_i\}_{1\leq i\leq N^3}$ denote the positions of the $N^3$ charges.

In this section, we prove a localization type result (\cpr{K_complete_model_localization_minimizer}) --- that any minimizer concentrates around the position of a charge of the lattice --- and a lower bound on the number of distinct minimizers (\cpr{K_complete_model_number_minimizers_lowerbound}).

\begin{prop}[Minimizers' concentration point]\label{K_complete_model_localization_minimizer}
Let $\{R_j\}_{1\leq j\leq N^3}$ be the respective positions of the $N^3$ charges inside $\K_N$. Then the sequence $\{x_n\}\subset c_n\cdot\K_N$ of translations associated with the subsequence $\{w_{c_n}\}$ of minimizers to $E_{\K_N,N^3\lambda}(c_n)$ such that the rescaled sequence $\mathds{1}_{\K_{c_n}}\breve{w}_{c_n}(\cdot+x_n)$ converges to $Q$, a minimizer to $J_{\R^3,N^3\lambda}$, verifies
$$x_n= {c_n}R_i+o(1)$$
as $n\to\infty$, for one $i$.
Consequently, for $2\leq p<+\infty$,
$$\norm{\breve{w}_{c_n}(\cdot+c_n R_i)-Q}_{L^p(\K_{c_n})}\underset{n\to+\infty}{\longrightarrow}0.$$
\end{prop}
\begin{proof}[Proof of~\cpr{K_complete_model_localization_minimizer}]
Since the $w_{c_n}$'s are minimizers, we have for any $R_j$ that
$${\mathscr E}_{\K_N,c_n}(w_{c_n})\leq {\mathscr E}_{\K_N,c_n}\Big(w_{c_n}\Big(\cdot+\frac{x_n}{c_n}-R_j\Big)\Big)$$
which leads to
\begin{multline*}
-\sum\limits_{i=1}^{N^3}\int_{\K_{N c_n}}{G_{\K_N}\Big(\frac{x}{c_n}+\frac{x_n}{c_n}-R_i\Big) \left|\breve{w}_{c_n}\left(x+x_n\right)\right|^2\dd x}\\
\leq -\sum\limits_{i=1}^{N^3}\int_{\K_{N c_n}}{ G_{\K_N}\Big(\frac{x}{c_n}+R_j-R_i\Big) \left|\breve{w}_{c_n}\left(x+x_n\right)\right|^2\dd x}
\end{multline*}
since the first four terms of ${\mathscr E}_{\K_N,c}$ are invariant under spatial translations. \clm{K_complete_model_convergence_second_G_K_term} below then gives, on one hand, that the right hand side of this inequality is equal to $-c_n\int_{\R^3}{\frac{Q^2(x)}{|x|}\dd x}+o(c_n)$ because $c_n|R_j-R_i|\to\infty$ for $i\neq j$ and, on the other hand, that $|x_n-c_n R_i|$ must be bounded for one $i$, that we denote $i_0$, because otherwise the left hand side would be equal to $o(c_n)$. Therefore, still by \clm{K_complete_model_convergence_second_G_K_term}, the term for $i_0$ in the left hand side is equal to $-c_n\int_{\R^3}{\frac{Q^2(x)}{|x-\eta|}\dd x}+o(c_n)$ for a given $\eta\in\R^3$ (and up to a subsequence) and the other terms of the sum to $o(c_n)$. However,
$$\int_{\R^3}{\frac{Q^2(x)}{|x|}\dd x}>\int_{\R^3}{\frac{Q^2(x)}{|x-\eta|}\dd x}$$
if $\eta\neq0$, implying that the $w_{c_n}$ are not minimizers for $n$ large enough. Hence $\eta=0$, which means by \clm{K_complete_model_convergence_second_G_K_term} that $x_n= {c_n}R_{i_0}+o(1)$ as $n\to\infty$.

The last result of \cpr{K_complete_model_localization_minimizer} is a direct consequence of the convergence of the $L^p(\K_{c_n})$-norms proved in \cpr{Main_term_expansion_of_minimum} and \clm{K_models_Linfinity_convergence_minimizers} together with the fact that $x_n-c_n R_{i_0}=o(1)$.
\begin{lemme}\label{K_complete_model_convergence_second_G_K_term}
Let $\{y_n\}_n\subset \K$, $\{f_c\}_c\subset L^2_{\textrm{per}}(\K_c)$ and $\{g_c\}_c\subset L^2_{\textrm{per}}(\K_c)$ be two sequences such that $\norm{f_c}_{H^1_{\textrm{per}}(\K_c)}+\norm{g_c}_{H^1_{\textrm{per}}(\K_c)}$ is uniformly bounded. We assume that there exist $f$ and $g$ in $H^1(\R^3)$ and a subsequence $c_n$ such that $\norm{f_{c_n}-f}_{L^2(\\K_{c_n})}\underset{n\to\infty}{\to} 0$ and $\mathds{1}_{\K_{c_n}}g_{c_n}\underset{n\to\infty}{\wto} g$ weakly in $L^2(\R^3)$. Then,
\begin{enumerate}[label=\roman*.,leftmargin=1.4em]
	\item\label{K_complete_model_convergence_second_G_K_term_item_infinity} if $c_n|y_n|\to+\infty$, then ${c_n}^{-1}\int_{\K_{c_n}}{G_\K({c_n}^{-1}\cdot-y_n) f_{c_n} g_{c_n}}\underset{n\to\infty}{\longrightarrow}0$,
	\item\label{K_complete_model_convergence_second_G_K_term_item_0} if $c_n|y_n|\to0$, then ${c_n}^{-1}\int_{\K_{c_n}}{G_\K({c_n}^{-1}\cdot-y_n) f_{c_n} g_{c_n}}\underset{n\to\infty}{\longrightarrow}\int_{\R^3}{\frac{f(x)g(x)}{|x|}\dd x}$,
	\item\label{K_complete_model_convergence_second_G_K_term_item_bounded} otherwise, there exist $\eta\in\R^3\setminus\{0\}$ and a subsequence $n_k$ such that
	$${c_{n_k}}^{-1}\int_{\K_{c_{n_k}}}{G_\K({c_{n_k}}^{-1}\cdot-y_{n_k}) f_{c_{n_k}} g_{c_{n_k}}}\underset{k\to\infty}{\longrightarrow}\int_{\R^3}{\frac{f(x)g(x)}{|x-\eta|}\dd x}.$$
\end{enumerate}
Moreover, replacing $\norm{f_{c_n}-f}_{L^2(\\K_{c_n})}\underset{n\to\infty}{\to} 0$ by $\norm{f_{c_n}-f}_{H^1(\\K_{c_n})}\underset{n\to\infty}{\to} 0$, the uniform bound on $\norm{g_c}_{H^1_{\textrm{per}}(\K_c)}$ by an uniform bound on $\norm{g_c}_{L^2_{\textrm{per}}(\K_c)}$ and $g\in H^1(\R^3)$ by $g\in L^2(\R^3)$, then \ref{K_complete_model_convergence_second_G_K_term_item_infinity} still holds true and, in the special case $y_n=0$, \ref{K_complete_model_convergence_second_G_K_term_item_0} too.
\end{lemme}
\begin{rmq*}
We state the lemma in a more general setting than needed for \cpr{K_complete_model_localization_minimizer} in order for it to be also useful for the proof of \clm{K_complete_model_Lplus_c_elliptic_liminf_ineq}.
\end{rmq*}
\begin{proof}[Proof of~\clm{K_complete_model_convergence_second_G_K_term}]
Using the same notation $\K^{\boldsymbol\sigma}$ as in the proof of~\clm{cvgce_EK_JK}, we notice that $\K-\tau:=\{x\in\R^3|x-\tau\in\K\}\subset\K_2=\K\cup\bigcup_{(0,0,0)\neq\boldsymbol\sigma\in\{0;\pm1\}^3}\K^{\boldsymbol\sigma}$, for any $\tau\in\K$. Therefore, by \clm{G_K_in_Lp}, there exists $C>0$ such that for any $\phi_c\in L^2(\K_c)$, $\psi_c\in H^1(\K_c)$, $y\in\K$ and $c>0$,
$$c^{-1}\left|\int_{\K_c}{G_\K(c^{-1}\cdot-y)\phi_c\psi_c}\right|\leq C\sum_{\boldsymbol\sigma\in\{-1;0;+1\}^3}\norm{\frac{\phi_c\psi_c}{|\cdot-c(y+\boldsymbol\sigma)|}}_{L^1(\K_c)}.$$
Then, by the Hardy inequality on $\K_c$, which is uniform on $[c_*,\infty)$ for any $c_*>0$, there exists $C'$ such that for any $y\in\K$ and any $c\geq1$, we obtain
$$c^{-1}\left|\int_{\K_c}{G_\K(c^{-1}\cdot-y)\phi_c\psi_c}\right|\leq 27C'\norm{\phi_c}_{L^2(\K_c)}\norm{\psi_c}_{H^1(\K_c)}.$$
Therefore, the weak convergence of $g_{c_n}$ and the Hardy inequality to $f$ on $\R^3$ give
\begin{multline*}
{c_n}^{-1}\bigg|\int_{\K_{c_n}}{G_\K({c_n}^{-1}\cdot-y_n)(f_{c_n} g_{c_n}-fg)}\bigg|\\
\leq 27\Big(C'\norm{f_{c_n}-f}_{L^2(\K_{c_n})}\norm{g_{c_n}}_{H^1(\K_{c_n})}+2C\norm{\frac{f(g_{c_n}-g)}{|\cdot-c(y+\boldsymbol\sigma)|}}_{L^1(\K_c)}\Big)\underset{n\to\infty}{\to}0.
\end{multline*}
Replacing $\norm{f_{c_n}-f}_{L^2(\K_{c_n})}\norm{g_{c_n}}_{H^1(\K_{c_n})}$ by $\norm{f_{c_n}-f}_{H^1(\K_{c_n})}\norm{g_{c_n}}_{L^2(\K_{c_n})}$ gives this same convergence to $0$ under the second set of conditions.

We are therefore left with the study of ${c_n}^{-1}\int_{\K_{c_n}}{G_\K({c_n}^{-1}\cdot-y_n)fg}$ as $n\to\infty$ and we start with the case $c_n|y_n|\to+\infty$. For $c>0$, $y\in\K$ and $\boldsymbol\sigma\in\{-1;0;+1\}^3$, we have
\begin{multline*}
c^{-1}\int_{\K_c}{\mathds{1}_{\K^{\boldsymbol\sigma}}(c^{-1}\cdot-y)G_\K(c^{-1}\cdot-y)|fg|}\\
	\lesssim \int_{\R^3}{\frac{\mathds{1}_{B(0,\frac{c}2|y+\boldsymbol\sigma|)}}{|\cdot-c(y+\boldsymbol\sigma)|}|fg|}
		    +\int_{\R^3}{\frac{\mathds{1}_{B(c(y+\boldsymbol\sigma),R)}}{|\cdot-c(y+\boldsymbol\sigma)|}|fg|}
		    +\int_{{}^\complement\!B(0,\frac{c}2|y+\boldsymbol\sigma|)}{\frac{\mathds{1}_{{}^\complement\!B(c(y+\boldsymbol\sigma),R)}}{|\cdot-c(y+\boldsymbol\sigma)|}|fg|}\\
	\lesssim \frac2{c|y+\boldsymbol\sigma|}\norm{fg}_{L^1(\R^3)}+\norm{f}_{H^1(\R^3)}\norm{g}_{L^2(B(c(y+\boldsymbol\sigma),R)}+\frac1R\norm{fg}_{L^1({}^\complement\!B(0,\frac{c}2|y+\boldsymbol\sigma|))},
\end{multline*}
for any $R>0$. Since $f$ is in $H^1(\R^3)$ and $g$ at least in $L^2(\R^3)$, the last two terms tends to $0$ and $\norm{fg}_{L^1(\R^3)}$ is bounded hence, on one hand we obtain, for $\boldsymbol\sigma=(0,0,0)$, the convergence to $0$ (for the subsequence $c_n$) from $c_n|y_n|\to+\infty$ and, on the other hand, there exists $R'>0$ such that $|y+\boldsymbol\sigma|>R'$ for any $\{-1;0;+1\}^3\ni\boldsymbol\sigma\neq(0,0,0)$ and any $y\in\K$, ending the proof that the above tends to $0$. We finally obtain that
$$\frac1{c_n}\int_{\K_{c_n}}{G_\K({c_n}^{-1}\cdot-y_n)|fg|}=\sum\limits_{\boldsymbol\sigma\in\{0;\pm1\}^3} \frac1{c_n}\int_{\K_{c_n}}{\left[\mathds{1}_{\K^{\boldsymbol\sigma}}G_\K\right]({c_n}^{-1}\cdot-y_n)|fg|}\underset{n\to\infty}{\longrightarrow}0,$$
concluding the proof of \emph{\ref{K_complete_model_convergence_second_G_K_term_item_infinity}} under the two sets of hypothesis.

We now suppose that $c_n|y_n|$ does not diverge hence it is bounded  up to a subsequence $n_k$ and, consequently, $y_{n_k}\to0$. However, by \clm{G_K_in_Lp}, there exists $M'>0$ such that $\left||\cdot|^{-1}-G_{\K}\right|\leq M'$ on $\K$, thus there exists $M>0$ such that
\begin{align*}
\left|G_\K-\frac1{|\cdot|}\right|\mathds{1}_{\K-\tau}&\leq \bigg(M'\mathds{1}_{\K}+\frac{\mathds{1}_{{}^\complement\K}}{|\cdot|}+C\sum_{(0,0,0)\neq\boldsymbol\sigma\in\{0;\pm1\}^3}\frac{\mathds{1}_{\K^{\boldsymbol\sigma}}}{|\cdot+\tau-\boldsymbol\sigma|-|\tau|}\bigg)\mathds{1}_{\K-\tau}\\
	&\leq \left(M'+R^{-1}+52CR^{-1}\right)\mathds{1}_{\K-\tau}\leq M\mathds{1}_{\K-\tau}.
\end{align*}
for $\tau\in B(0,R/2)$ and where $R:=\min_{x\in\partial\K}|x|>0$ therefore $B(0,R)\subset\K$. Hence
$$\bigg|\int_{\K_{c_{n_k}}}{\left(\frac1{c_{n_k}}G_\K(\frac\cdot{c_{n_k}}-y_{n_k})-|\cdot-c_{n_k}y_{c_{n_k}}|^{-1}\right)fg}\bigg|\leq \frac{M}{c_{n_k}}\norm{fg}_{L^1(\R^3)}=O(\frac1{c_{n_k}}).$$
Moreover,
$$\bigg|\int_{\R^3}{(1-\mathds{1}_{\K_{c_{n_k}}}(x))\frac{f(x)g(x)}{|x-c_{n_k}y_{c_{n_k}}|}\dd x}\bigg|\lesssim \norm{f}_{L^2({}^\complement\K_{c_{n_k}})}\norm{g}_{H^1(\R^3)}\to0$$
and we are left with the study of
$$\bigg|\int_{\R^3}{\frac{f(x)g(x)}{|x-c_{n_k}y_{c_{n_k}}|}-\frac{f(x)g(x)}{|x-\eta|}\dd x}\bigg|\leq4|\eta-c_{n_k}y_{c_{n_k}}|\norm{f}_{H^1(\R^3)}\norm{g}_{H^1(\R^3)}$$
which tends to $0$ if we choose $\eta$ as the limit (up to another subsequence) of the bounded sequence $c_{n_k}y_{n_k}$. Finally, if we have in fact $c_ny_n\to0$ then $\eta=0$, otherwise, we can find a subsequence such that $c_{n_k}y_{n_k}\to\eta\neq0$.

Under the second set of conditions and if $y_n=0$, we have
$$\bigg|\int_{\K_{c_n}}{({c_n}^{-1}G_\K({c_n}^{-1}x)-|x|^{-1})f(x)g(x)\dd x}\bigg|\leq \frac{M'}{c_n}\norm{fg}_{L^1(\R^3)}=O({c_n}^{-1}).$$

This concludes the proof of \clm{K_complete_model_convergence_second_G_K_term}.
\end{proof}

This concludes the proof of~\cpr{K_complete_model_localization_minimizer}.
\end{proof}

We now prove that $E_{\K_N,N^3\lambda}(c)$ admits at least $N^3$ distinct minimizers.
\begin{prop}\label{K_complete_model_number_minimizers_lowerbound}
For $c_n$ large enough, there exist at least $N^3$ nonnegative minimizers to the minimization problem $E_{\K_N,N^3\lambda}(c_n)$ which are translations one of each other by vectors $R_j-R_k$, $1\leq j\neq k\leq N^3$, where $\{R_i\}_{1\leq i\leq N^3}$ are the respective positions of the $N^3$ charges inside $\K_N$.
\end{prop}
\begin{proof}[Proof of~\cpr{K_complete_model_number_minimizers_lowerbound}]
First, in~\cpr{K_complete_model_localization_minimizer}, we have seen that any sequence $\{w_c\}_{c\to+\infty}$ of minimizers of ${E}_{\K_N,N^3\lambda}(c)$ must concentrate, up to a subsequence, at the position of one nucleus of the unit cell, denoted $R_{j_0}$. Then, given that the four first terms of ${\mathscr E}_{\K_N,c}$ are invariant under any translations and $\int G_\K|w_c|^2$ is invariant under $R_j-R_k$ translations, we have that each $w_c(\cdot+R_i-R_{j_0})$, for $1\leq i\leq N^3$, is also a minimizer of ${E}_{\K_N,N^3\lambda}(c)$. Since, the $N^3$ sequences of minimizers $\left\{w_{c_n}(\cdot+R_i-R_{j_0})\right\}_{i}$ have distinct limits as $n\to\infty$, there are at least $N^3$ distinct minimizers for $n$ large enough.
\end{proof}

\subsection{Second order expansion of \texorpdfstring{$E_{\K,\lambda}(c)$}{the minimum E}}\label{section_second_order_expansion_E}
The goal of this subsection is to prove the expansion~\eqref{K_complete_model_expansion_E}. To do so, we improve the convergence rate of the first order expansion of $J_{\K,\lambda}(c)$ proved in \cpr{Main_term_expansion_of_minimum}. Namely, we prove that there exists $\beta>0$ such that
\begin{equation}\label{K_eff_model_expansion_exponential_rate_equation}
J_{\K,\lambda}(c)= c^2 J_{\R^3}(\lambda)+o(e^{-\beta c}).
\end{equation}
We recall that we have proved in~\clm{R3_eff_model_main_term_expansion_upper_bound} that there exists $\beta>0$ such that 
$$J_{\K,\lambda}(c)\leq c^2 J_{\R^3}(\lambda)+o(e^{-\beta c})$$
and we now turn to the proof of the converse inequality.
\begin{lemme}\label{K_eff_model_expansion_exponential_rate}
There exists $\beta>0$ such that
\begin{equation*}\label{K_eff_model_above_expansion_exponential_rate_equation}
J_{\K,\lambda}(c)\geq c^2 J_{\R^3,\lambda}+o(e^{-\beta c}).
\end{equation*}
\end{lemme}
\begin{proof}[Proof of~\clm{K_eff_model_expansion_exponential_rate}]
As the problems ${J}_{\K,\lambda}(c)$ are invariant by spatial translations, we can suppose that $x_n=0$ in the convergences of the subsequence of rescaled functions $\mathds{1}_{\K_{c_n}}\breve{v}_{c_n}(\cdot+x_n)$. Our proof relies on the exponential decay with $c$ of the minimizers to $J_{\K_c,\lambda}(1)$ close to the border of the cube $\K_c$.
\begin{lemme}[Exponential decrease of minimizers to $J_{\K_c,\lambda}(1)$]\label{K_eff_model_minimizer_exponential_decr}
Let $\{v_c\}_{c}$ be a sequence of nonnegative minimizers to ${J}_{\K,\lambda}(c)$ such that a subsequence of rescaled functions $\mathds{1}_{\K_{c_n}}\breve{v}_{c_n}$ converges weakly to a minimizer of $J_{\R^3}(\lambda)$. Then there exist $C,\gamma>0$ such that for $c$ large enough, we have $0\leq\breve{v}_{c_n}(x)<C e^{- \gamma c}$
for $x\in\K_c\setminus\K_{c-1}$.
\end{lemme}
\begin{proof}[Proof of~\clm{K_eff_model_minimizer_exponential_decr}]
We denote by $u$ the minimizer of $J_{\R^3}(\lambda)$ to which $\mathds{1}_{\K_{c_n}}\breve{v}_{c_n}$ converges strongly and by $\mu_{\R^3}$ the Euler--Lagrange parameter~\eqref{R3_eff_model_EulerLagrange} associated with this specific $u$. The Euler--Lagrange equation associated with $J_{\K_{c_n},\lambda}(1)$ --- solved by $\breve{v}_{c_n}$ --- gives
$$\left(-\Delta+\frac{\mu_{\R^3}}4\right)\breve{v}_{c_n} \leq  \left(|\breve{v}_{c_n}|^{\frac23} +\frac{\mu_{\R^3}}4-{c_n}^{-2}\mu_{c_n}\right)\breve{v}_{c_n}.$$
We now define $\Omega_{c_n}=(1+\epsilon)\K_{c_n}\setminus B(0, \alpha)$ where $\alpha$ is such that $|u|^{\frac23}\leq \min\{\frac12, \frac{\mu_{\R^3}}4\}$ on $\R^3\setminus B(0, \alpha)$. Such $\alpha$ exists by the exponential decay of $u$ at infinity. Therefore, by~\clm{K_models_Linfinity_convergence_minimizers}, for any $c_n$ large enough, we have $|\breve{v}_{c_n}|^{2/3}\leq \min\left\{1, \frac{\mu_{\R^3}}2\right\}$ on $\K_{c_n}\setminus  B(0, \alpha)$ but also on $\Omega_{c_n}$ by periodicity of $\breve{v}_{c_n}$ and for any $c_n$ large enough (depending on $\epsilon$) in order to have
$$(1+\epsilon)\K_{c_n}\cap \bigcup\limits_{k\in {\mathscr L}_\K\setminus\{0\}} B({c_n}k,\alpha)=\emptyset.$$
Together with~\ccr{convergence_EulerLagrange_multiplier}, it gives on $\Omega_{c_n}$, for $c_n$ large enough, that
$$\left(-\Delta+\frac{\mu_{\R^3}}4\right)\breve{v}_{c_n}\leq0 \qquad \textrm{ and } \qquad |\breve{v}_{c_n}|\leq 1.$$

We now define on $\R^3\setminus B(0,\nu)$, for any $\nu>0$, the positive function
$$f_\nu(x)=\nu|x|^{-1}e^{\frac{\sqrt{\mu_{\R^3}}}2(\nu-|x|)}$$
which solves
$$-\Delta f_\nu + \frac{\mu_{\R^3}}4 f_\nu=0$$
on $\R^3\setminus B(0,\nu)$ and verifies $f_\nu=1$ on the boundary $\partial B(0,\nu)$. On each $(1+\epsilon)\K_{c_n}$, we define the positive function
$$f_0(x)=\sum\limits_{j=1}^3 \frac{\cosh\left(\frac{\sqrt{\mu_{\R^3}}}2 x_j\right)}{\cosh\left(\frac{\sqrt{\mu_{\R^3}}}4(1+\epsilon){c_n}\right)}$$
which solves
$$-\Delta f_0 + \frac{\mu_{\R^3}}4 f_0=0$$
on $(1+\epsilon)\K_{c_n}$ and verifies $1\leq f_0\leq3$ on the boundary $\partial \left((1+\epsilon)\K_c\right)$. Denoting by $g$ the function $g:=f_0+ f_\alpha$, we have for $c_n$ large enough that
$$\left(-\Delta+\frac{\mu_{\R^3}}4\right)(\breve{v}_{c_n}-g)\leq0, \textrm{ on } \Omega_{c_n} \qquad \textrm{ and } \qquad \breve{v}_{c_n}-g\leq0, \textrm{ on } \partial\Omega_{c_n},$$
hence the maximum principle implies that $\breve{v}_{c_n}\leq g$ on $\Omega_{c_n}$.

On one hand, since the function $f_0$ is even along each spatial direction of the cube and increasing on $[0;(1+\epsilon)\frac{c_n}2)$ in those directions, we have that for any $x\in \K_{c_n}$, so in particular on $\K_{c_n}\setminus\K_{c_n-1}$, that
$$0<f_0(x)\leq f_0\left(\frac{c_n}2(1,1,1)\right)\leq2\sum\limits_{j=1}^3 e^{-\epsilon\frac{\sqrt{\mu_{\R^3}}}4{c_n}}.$$
On the other hand, $|x|\geq (c_n-1) m>0$ for $x\in\K_{c_n}\setminus\K_{c_n-1}$, with $m:=\min\limits_{\partial \K} |x|$, thus
$$0<f_\alpha(x)\leq\alpha e^{\frac{\sqrt{\mu_{\R^3}}}2(\alpha+m)} m^{-1}(c_n-1)^{-1} e^{-\frac{\sqrt{\mu_{\R^3}}}2 m c_n}$$
on $\K_{c_n}\setminus\K_{c_n-1}$. Hence there exist $C>0$ and $\gamma:=\frac{\sqrt{\mu_{\R^3}}}2\min\{\frac{\epsilon}2;m\}>0$ such that for $c_n$ large enough and any $x\in \K_{c_n}\setminus\K_{c_n-1}$, we conclude that
\begin{equation*}
0\leq\breve{v}_{c_n}(x)\leq g(x)<C e^{- \gamma c}.\qedhere
\end{equation*}
\end{proof}
We now conclude the proof of~\clm{K_eff_model_expansion_exponential_rate}. We define $\chi_c\in C^\infty_c(\R^3)$, $0\leq\chi_c\leq1$, $\chi_c\equiv0$ on $\R^3\setminus \K_c$ and $\chi_c\equiv1$ on $\K_{c-1}$. Since $\left|\K_c\setminus\K_{c-1}\right|\leq |\K_c|={c}^3|\K|$ for any $c>1$ and by~\clm{K_eff_model_minimizer_exponential_decr}, we have that there exist $0<\alpha<\gamma$ such that
\begin{align*}
0\leq \norm{\breve{v}_{c_n}}^p_{L^p(\K_{c_n})}-\norm{\chi_{c_n} \breve{v}_{c_n}}^p_{L^p(\R^3)}&=\int_{\K_{c_n}\setminus\K_{c_n-1}} (1-{\chi_{c_n}}^p)|\breve{v}_{c_n}|^p\\
	&\leq C^p e^{- p \gamma c_n}\left|\K_{c_n}\setminus\K_{c_n-1}\right|=o\left(e^{- p \alpha {c_n}}\right),
\end{align*}
for any $p\in[2;6]$. Moreover, for any $c>1$, we have
\begin{align*}
\left|\int_{\R^3}\chi_c\breve{v}_c\nabla\chi_c\cdot\nabla\breve{v}_c\right|&=\frac12\left|\int_{\R^3}|\breve{v}_c|^2\nabla(\chi_c\nabla\chi_c)\right|\leq \frac12\int_{\K_c\setminus\K_{c-1}} |\breve{v}_c|^2 \left(\chi_c|\Delta\chi_c| +|\nabla\chi_c|^2\right)
\end{align*}
hence
$$\norm{\nabla(\chi_{c_n} \breve{v}_{c_n})}^2_{L^2(\R^3)}=\norm{\chi_{c_n} \nabla\breve{v}_{c_n}}^2_{L^2(\K_{c_n})}+o(e^{-2\alpha c_n})\leq\norm{\nabla\breve{v}_{c_n}}^2_{L^2(\K_{c_n})}+o(e^{-2\alpha c_n}).$$
Consequently, there exists $\beta>0$ such that
$$J_{\R^3}(\lambda)\leq {\mathscr J}_{\R^3}\bigg(\frac{\sqrt{\lambda} \chi_{c_n} u}{\norm{\chi_{c_n} u}_{L^2(\R^3)}}\bigg)\leq{\mathscr J}_{\K_{c_n}}(\breve{v}_{c_n})+o(e^{-\beta c_n})=J_{\K_{c_n}}(\lambda)+o(e^{-\beta c_n}).$$
This concludes the proof of~\clm{K_eff_model_expansion_exponential_rate}.
\end{proof}

We can now turn to the proof of the second-order expansion of the energy.
\begin{prop}[Second order expansion of the energy]\label{K_complete_model_2nd_order_expansion_energy}
We have the expansion
\begin{multline}
E_{\K_N,N^3\lambda}(c)={c}^2 J_{\R^3,N^3\lambda}\\
+c \inf_u\left\{\frac12\int_{\R^3}{\int_{\R^3}{\frac{|u(x)|^2|u(y)|^2}{|x-y|}\dd y}\dd x} - \int_{\R^3} \frac{|u(x)|^2}{|x|} \dd x\right\} +o(c),
\end{multline}
where the infimum is taken over all the minimizers of $J_{\R^3,N^3\lambda}$.
\end{prop}
\begin{proof}[Proof of~\cpr{K_complete_model_2nd_order_expansion_energy}]
In order to deal with the term $D_\K$, we first prove a convergence result similar to what we did in~\clm{K_complete_model_convergence_second_G_K_term} for term $\int G|w|^2$.
\begin{lemme}\label{K_complete_model_convergence_D_K}
Let $v_c$ be such that the rescaled function $\breve{v}_c=c^{-3/2}v_c(c^{-1} x)$ verifies
$$\mathds{1}_{\K_c}\breve{v}_c\underset{c\to\infty}{\longrightarrow} v$$
strongly in $L^2(\R^3)\cap L^{\frac{12}5}(\R^3)$, then
\begin{equation*}
{c}^{-1}D_\K({v_c}^2,{v_c}^2)\to \int_{\R^3}{\int_{\R^3}{\frac{v^2(x)v^2(y)}{|x-y|}\dd y}\dd x}=:D_{\R^3}(v^2,v^2).
\end{equation*}
\end{lemme}
\begin{proof}[Proof of~\clm{K_complete_model_convergence_D_K}]We have
\begin{multline*}
D_{\R^3}(v^2,v^2)- {c}^{-1} D_\K(v_c^2, v_c^2)\\
= D_{\R^3}(v^2, v^2-\mathds{1}_{\K_c}{\breve{v}_c}^2) + D_{\R^3}(v^2-\mathds{1}_{\K_c}{\breve{v}_c}^2,\mathds{1}_{\K_c}{\breve{v}_c}^2)\\
+ {c}^{-1} \int_{\K}{\int_{\K}{v_c^2(x)\left(|x-y|^{-1} - G_\K(x-y)\right)v_c^2(y)\dd y}\dd x}.
\end{multline*}
By the Hardy--Littlewood--Sobolev inequality and the strong convergence of $\mathds{1}_{\K_c}\breve{v}_c$ in $L^{12/5}(\R^3)$, the two first terms of the right hand side vanish.

To prove that the last term vanishes too, we split the double integral over $\K\times\K$ into several parts depending on the location of $x-y$.

We start by proving the convergence for $x-y\in \K$. By~\clm{G_K_in_Lp},
\begin{multline*}
{c}^{-1} \iint_{\substack{\K\times \K\\x-y\in \K}}{v_c^2(x)\left||x-y|^{-1} - G_\K(x-y)\right|v_c^2(y)\dd y\dd x}\\
\leq \frac{M}{c}\iint_{\substack{\K\times \K\\x-y\in \K}}v_c^2(x)v_c^2(y)\dd x\dd y\leq \frac{M}{c}\norm{v_c}_{L^2(\K)}^4=\frac{M}{c}\norm{\breve{v}_c}_{L^2(\K_c)}^4\underset{c\to\infty}{\longrightarrow}0.
\end{multline*}

When $x-y\notin \K$, we treat first the term due to $|\cdot|^{-1}$. We have
\begin{equation*}
{c}^{-1}\iint_{\substack{\K\times \K\\x-y\in 2\K\setminus \K}}{\frac{v_c^2(x)v_c^2(y)}{|x-y|}\dd y\dd x}\leq2{c}^{-1}\norm{v_c}_{L^2(\K)}^4\underset{c\to\infty}{\longrightarrow}0.
\end{equation*}

To deal with the remaining terms due to $G_\K$ when $x-y\notin \K$, we will use the same notation $\K^{\boldsymbol\sigma}$ as in the proof of~\clm{cvgce_EK_JK}. By~\eqref{equation_G_K_in_Lp}, we therefore have to prove, for $\boldsymbol\sigma\in\{-1,0,+1\}^3\setminus(0,0,0)$, the vanishing of
$$\Big|{c}^{-1}\iint_{\substack{\K\times \K\\x-y\in \K^{\boldsymbol\sigma}}}{v_c^2(x) G_\K(x-y) v_c^2(y)\dd y\dd x}\Big|\lesssim \iint_{\substack{\K_c\times \K_c\\x-y\in c\cdot\K^{\boldsymbol\sigma}}}{\frac{{\breve{v}_c}^2(x) {\breve{v}_c}^2(y)}{|x- y-c \boldsymbol\sigma|}\dd y\dd x}.$$
Let $0<\nu<\frac14$. Given that $\boldsymbol\sigma\neq(0,0,0)$, we have
$$\left\{(x,y)\in \K_c\times \K_c \left|\; x-y\in c\cdot\K^{\boldsymbol\sigma}\right.\right\} \cap B(0,c\nu)\times B(0,c\nu) =\emptyset.$$
Hence, using the Hardy--Littlewood--Sobolev inequality, we obtain
$$\Big|\frac1{c}\iint_{\substack{\K\times \K\\x-y\in \K^{\boldsymbol\sigma}}}{v_c^2(x) G_\K(x-y) v_c^2(y)\dd y\dd x}\Big|\lesssim2\norm{\breve{v}_c}_{L^{12/5}(\K_c\setminus B(0,c\nu))}^2 \norm{\breve{v}_c}_{L^{12/5}(\K_c)}^2$$
and the right hand side vanishes when $c\to0$ since $\norm{\breve{v}_c}_{L^{12/5}(\K_c\setminus B(0,c\nu))}^2$ vanishes and $\norm{\breve{v}_c}_{L^{12/5}(\K_c)}^2$ is bounded, both by the $L^{12/5}(\R^3)$-convergence of $\mathds{1}_{\K_c}\breve{v}_c$. This concludes the proof of~\clm{K_complete_model_convergence_D_K}.
\end{proof}

Let $w_c$ be a sequence of minimizers to $E_{\K_N,N^3\lambda}(c)$. By Propositions \ref{Main_term_expansion_of_minimum} and \ref{K_complete_model_localization_minimizer}, the convergence rate \eqref{K_eff_model_expansion_exponential_rate_equation}, and Lemmas \ref{K_eff_model_expansion_exponential_rate} and~\ref{K_complete_model_convergence_D_K}, we obtain
$$E_{\K_N,N^3\lambda}(c)={c}^2 J_{\R^3,N^3\lambda}+c\left(\frac12D_{\R^3}(|Q|^2,|Q|^2) - \int_{\R^3} \frac{|Q(x)|^2}{|x|} \dd x\right) +o(c),$$
where $Q$ is the minimizer of $J_{\R^3,N^3\lambda}$ to which $\mathds{1}_{c_n\cdot\K_N}\breve{w}_{c_n}(\cdot+x_n)$ converges strongly.

Let us now prove that $Q$ must also minimize the term of order $c$. We suppose that there exists a minimizer $u$ of $J_{\R^3,N^3\lambda}$ such that ${\mathscr S}(u)<{\mathscr S}(Q)$, where
$${\mathscr S}(f):=\frac12\int_{\R^3}{\int_{\R^3}{\frac{|f(x)|^2|f(y)|^2}{|x-y|}\dd y}\dd x} - \int_{\R^3} \frac{|f(x)|^2}{|x|} \dd x.$$
By arguing as in Propositions \ref{R3_eff_model_main_term_expansion_upper_bound} and \ref{K_eff_model_expansion_exponential_rate}, and defining, for a fixed small $\eta\in(0;1)$, the smooth function $\chi\in C^\infty_0(\K_N)$ verifying $0\leq\chi\leq1$, ${\chi}_{|(1-\eta)\cdot\K_N}\equiv1$, ${\chi}_{|\R^3\setminus \K_N}\equiv0$, we can prove that there exists $\nu>0$ such that
$${\mathscr J}_{\K_N,c}\left(\sqrt{N^3\lambda}\frac{u(c\cdot) \chi}{\norm{u(c\cdot) \chi}_{L^2(\K_N)}}\right)=c^2 J_{\R^3,N^3\lambda}+o(e^{-\nu c})_{c\to\infty}.$$
On the other hand, since $\frac{\sqrt{N^3\lambda}\chi(c^{-1}\cdot)}{\norm{c^{3/2}u(c\cdot) \chi}_{L^2(\K_N)}}\mathds{1}_{c\cdot\K_N}u\to u$ strongly in $L^2(\R^3)\cap L^4(\R^3)$, we apply Lemmas~\ref{K_complete_model_convergence_second_G_K_term} and~\ref{K_complete_model_convergence_D_K} to it and finally obtain
\begin{align*}
{\mathscr E}_{\K_N,c}\left(\sqrt{N^3\lambda}\frac{[u(c\cdot)\chi](\cdot-R_{j_0})}{\norm{u(c\cdot)\chi}_{L^2(\K_N)}}\right)&=c^2 J_{\R^3,N^3\lambda}+c{\mathscr S}(u)+o(c)\\
	&<c^2 J_{\R^3,N^3\lambda}+c{\mathscr S}(Q)+o(c)=E_{\K_N,N^3\lambda}(c),
\end{align*}
leading to a contradiction which finally proves that $Q$ minimizes ${\mathscr S}$ and thus concludes the proof of~\cpr{K_complete_model_2nd_order_expansion_energy}.
\end{proof}

\cth{main_result} is therefore proved combining the results of~\cpr{Main_term_expansion_of_minimum}, \cpr{K_complete_model_localization_minimizer}, \cpr{K_complete_model_number_minimizers_lowerbound} and~\cpr{K_complete_model_2nd_order_expansion_energy}.

\subsection{\texorpdfstring{Proof of~\cth{R3_eff_model_consequence_conjecture_uniqueness_and_monotony_M} on the number of minimizers}{Number of minimizers}}
\label{section_number_minimizers_under_conjecture}
The arguments developed in this section do not rely on what we have done in Section~\ref{section_second_order_expansion_E}.

We can expand the functional ${\mathscr E}_{\K,c}$ around a minimizer $w_c$ as
\begin{multline}\label{K_complete_model_expansion_around_minimizer_of_energy}
{\mathscr E}_{\K,c}(w_c+f)=E_{\K,\lambda}(c) +\pscalSM{\mathring{L}^+_c f_1,f_1}_{L^2(\K)}+\pscalSM{\mathring{L}^-_c f_2,f_2}_{L^2(\K)}-2\mu_c\pscal{w_c,f_1}_{L^2(\K)}\\
	-\mu_c\norm{f}^2_{L^2(\K)}+2D_\K(\Re(w_c\bar{f}), \Re(w_c\bar{f}))+o(\norm{f}^2_{H^1(\K)}),
\end{multline}
for $f\in H^1_{\textrm{per}}(\K,\C)$, with $f_1:=\Re(f)$, $f_2:=\Im(f)$ and where
\begin{equation}\label{K_complete_model_def_Lmoins}
\mathring{L}^-_c:=-\Delta +c_{TF}|w_c|^{\frac43}-c|w_c|^{\frac23}+\mu_c-{\mathscr G} +|w_c|^2\star G_\K
\end{equation}
and
\begin{equation}\label{K_complete_model_def_Lplus}
\mathring{L}^+_c=-\Delta+\frac73 c_{TF} |w_c|^{\frac43}-\frac53 c |w_c|^{\frac23}+\mu_c-{\mathscr G} +|w_c|^2\star G_\K,
\end{equation}
where ${\mathscr G}$ is defined by
\begin{equation*}
{\mathscr G}:=\sum\limits_{i=1}^{N^3} G_{\K_N}(\cdot-R_i).
\end{equation*}
We have used here that
\begin{multline}\label{Taylor_expansion_power_terms}
\int|w+h|^p-\int|w|^p-p\int|w|^{p-2}\Re(w \bar{h})\\
-\frac{p(p-2)}2\int_{w(\cdot)\neq0}|w|^{p-4}|\Re(w \bar{h})|^2-\frac{p}2\int|w|^{p-2}|h|^2=o\left(\norm{h}_{H^1}^2\right).
\end{multline}
for any complex-valued $w,h\in H^1$ and $2\leq p<4$ (see \cite{Ricaud-PhD} for details).

Let us suppose that~\ccjt{R3_eff_model_conjecture_uniqueness_and_monotony_M} holds and that there exist two sequences $w_c$ and $\nu_c$ of nonnegative minimizers to ${E}_{\K_N,N^3\lambda}(c)$ concentrating around the same nucleus at position $R\in\K$. Then, by \cpr{K_complete_model_localization_minimizer}, we have for $2\leq p<+\infty$ that
$$\norm{\breve{w}_{c_n}(\cdot+c_n R)-Q}_{L^p(\K_{c_n})}+\norm{\breve{\nu}_{c_n}(\cdot+c_n R)-Q}_{L^p(\K_{c_n})}\underset{n\to+\infty}{\longrightarrow}0$$
for a subsequence $c_n$. We define the real-valued $f_n:=w_{c_n}-\nu_{c_n}$, which verifies that $\normSM{\breve{f}_n}_{H^2_{\textrm{per}}(\K_{c_n})}$ uniformly bounded and, for $c_n>0$, the orthogonality properties
\begin{equation}\label{Weinstein_ortho_key_property_1}
\pscalSM{w_{c_n}+\nu_{c_n},f_n}_{L^2_{\textrm{per}}(\K)}=\pscalSM{\breve{w}_{c_n}+\breve{\nu}_{c_n},\breve{f}_n}_{L^2_{\textrm{per}}(\K_{c_n})}=0
\end{equation}
and
\begin{equation}\label{Weinstein_ortho_key_property_2}
\pscalSM{{\mathscr G}({c_n}^{-1}{\cdot}),\nabla((\breve{w}_{c_n}+\breve{\nu}_{c_n})\breve{f}_n)}_{L^2_{\textrm{per}}(\K_{c_n})}=0
\end{equation}
Indeed, the fact that $\nu_c$ and $w_c$ are real-valued gives the orthogonality \eqref{Weinstein_ortho_key_property_1}. Moreover, the orthogonality property stated in the following lemma leads to~\eqref{Weinstein_ortho_key_property_2}.
\begin{lemme}\label{K_complete_model_w_ortho_to_GtimesNabla_w}
If $w_c$ is a real-valued minimizer to $E_{\K,\lambda}(c)$, then $w_c$ is orthogonal to ${\mathscr G}\nabla{w_c}$.
\end{lemme}
\begin{proof}[Proof of~\clm{K_complete_model_w_ortho_to_GtimesNabla_w}]
As mentioned in~\cpr{K_complete_model_number_minimizers_lowerbound}, the four first terms of ${\mathscr E}_{\K,c}$ are invariant under any space translations thus we have
$${\mathscr E}_{\K,c}(w_c(\cdot+\mathbf{\tau}))=E_{\K,\lambda}(c) -2\mathbf\tau\boldsymbol{\cdot}\int_\K {\mathscr G} \Re(w_c \nabla \bar{w}_c) + O(|\mathbf\tau|^2).$$
Hence $\pscal{{\mathscr G},\Re\left(w_c\nabla \bar{w}_c\right)}_{L^2(\K)}=0$ for any minimizer $w_c$. Since ${\mathscr G}$ is real-valued, then $\pscal{w_c, {\mathscr G}\nabla w_c}_{L^2(\K)} =0$ if $w_c$ is a real-valued minimizer.
\end{proof}

By property~\eqref{Weinstein_ortho_key_property_2} together with $D_\K(h,h)\geq0$ (\clm{G_K_in_Lp}) and
$$2\pscalSM{\breve{w}_n,\breve{f}_n}_{L^2(\K_{c_n})}+\normSM{\breve{f}_n}^2_{L^2(\K_{c_n})}=\pscalSM{\breve{w}_n+\breve{\nu}_n,\breve{f}_n}_{L^2(\K_{c_n})}=0,$$
we obtain from~\eqref{K_complete_model_expansion_around_minimizer_of_energy} that
$$E_{\K,\lambda}(c_n)={\mathscr E}_{\K,c_n}(\nu_{c_n})\geq E_{\K,\lambda}(c_n) +{c_n}^2\pscalSM{L^+_n \breve{f}_n,\breve{f}_n}_{\K_{c_n}}+o(\norm{f_n}^2_{H^1(\K)})$$
where the operator $L^+_n$ is defined on $L^2(\K_{c_n})$ by
\begin{equation}\label{K_complete_model_def_Lplus_bigbox}
L^+_n=-\Delta+\frac73 c_{TF} |\breve{w}_c|^{\frac43}-\frac53 |\breve{w}_c|^{\frac23}+\frac{\mu_{c_n}}{{c_n}^2}+{c_n}^{-2}[-{\mathscr G}+|w_{c_n}|^2\star G_\K]({c_n}^{-1}\cdot).
\end{equation}
Therefore, by the ellipticity result $\pscalSM{L^+_n \breve{f}_n,\breve{f}_n}_{L^2(\K_{c_n})}\geq C\normSM{\breve{f}_n}_{H^1(\K_{c_n})}^2\geq0$ of the next proposition, which rely on~\ccjt{R3_eff_model_conjecture_uniqueness_and_monotony_M}, we obtain (for $c_n$ large enough) that
\begin{align*}
0&\geq C{c_n}^2\normSM{\breve{f}_n}_{H^1(\K_{c_n})}^2+o(\norm{f_n}^2_{H^1(\K)})=C{c_n}^2\normSM{\breve{f}_n}_{H^1(\K_{c_n})}^2+o({c_n}^2\normSM{\breve{f}_n}_{H^1(\K_{c_n})}^2)
\end{align*}
hence that $f_n\equiv0$ for $c$ large enough, i.e. $w_{c_n}\equiv \nu_{c_n}$. This means that if~\ccjt{R3_eff_model_conjecture_uniqueness_and_monotony_M} holds then there cannot be more than $N^3$ nonnegative minimizers for $c$ large enough and, together with~\cpr{K_complete_model_number_minimizers_lowerbound}, this concludes the proof of~\cth{R3_eff_model_consequence_conjecture_uniqueness_and_monotony_M}. We are thus left with the proof of the following non-degeneracy result.
\begin{prop}\label{K_complete_model_Lplus_c_elliptic}
Let $(w_c)_c$ be a sequence of minimizer to $E_{\K,\lambda}(c)$ and $L^+_n$ the associated operator as in~\eqref{K_complete_model_def_Lplus_bigbox}. Then there exists $C, c_*>0$ such that for any $c>c_*$ and any $f_n\in H^1(\K_c,\C)$ verifying the two orthogonality properties~\eqref{Weinstein_ortho_key_property_1} and~\eqref{Weinstein_ortho_key_property_2}, we have
\begin{equation}
\pscal{L^+_n f_n,f_n}_{L^2(\K_{c_n})}\geq C\norm{f_n}_{H^1(\K_{c_n})}^2.
\end{equation}
\end{prop}
\begin{proof}[Proof of~\cpr{K_complete_model_Lplus_c_elliptic}]
Following ideas in~\cite{Weinstein-85}, we define
$$\alpha_n:=\inf\limits_{\substack{f\in H^1(\K_c)\\ \pscalSM{\breve{w}_n+\breve{\nu}_n,f}_{L^2(\K_{c_n})}=0\\ \pscalSM{{\mathscr G}({c_n}^{-1}{\cdot}),\nabla((\breve{w}_{c_n}+\breve{\nu}_{c_n})f)}_{L^2(\K_{c_n})}=0}}\frac{\pscal{L^+_n f,f}_{L^2(\K_{c_n})}}{\norm{f}_{H^1(\K_{c_n})}^2}$$
and we will show that $\alpha_n>0$ for $c$ large enough.
\begin{lemme}\label{K_complete_model_Lplus_c_elliptic_liminf_ineq}
Let $(w_c)_c$ be a sequence of minimizer to $E_{\K,\lambda}(c)$ and $Q$ the positive minimizer of $J_{\R^3,\lambda}$ associated with the converging subsequence $\mathds{1}_{\K_{c_n}}\breve{w}_{c_n}(\cdot+c_nR)$. Define as in~\eqref{R3_eff_model_def_Lplus} the operator $L^+_\mu$ associated with $Q$ and, as in~\eqref{K_complete_model_def_Lplus_bigbox}, $L^+_n$ associated with $w_{c_n}$. Let $(f_n)_n$ be a uniformly bounded sequence of $H^1_{\textrm{per}}(\K_{c_n})$ then
$$\pscalSM{L^+_\mu f,f}_{L^2(\R^3)}\leq \liminf\limits_{n\to\infty} \pscalSM{L^+_n f_{n},f_{n}}_{L^2(\K_{c_n})},$$
with $f$ such that $\mathds{1}_{\K_{c_n}} f_{n}(\cdot+c_n R)\wto f$ weakly converges in $L^2(\R^3)$.
\end{lemme}
\begin{proof}[Proof of~\clm{K_complete_model_Lplus_c_elliptic_liminf_ineq}]
Up to the extraction of a subsequence (that we will omit in the notation), there exists $f$ such that $\mathds{1}_{\K_{c_n}} f_{n}(\cdot+c_n R)\wto f$ weakly in $L^2(\R^3)$ because $f_n(\cdot+c_n R)$ is uniformly bounded in $H^1(\K_{c_n})$. Thus, by~\clm{K_eff_dilated_strongconvergence_results_of_weakly_convergent_sequence},
$$\liminf\limits_{c\to\infty}\norm{\nabla f_n}_{L^2(\K_{c_n})}=\liminf\limits_{c\to\infty}\norm{\nabla f_n(\cdot+c_n R)}_{L^2(\K_{c_n})}\geq\norm{\nabla f}_{L^2(\R^3)}.$$

Moreover, $\norm{f_n}_{H^1(\K_{c_n})}$ is uniformly bounded by hypothesis thus
$${c_n}^{-2}\pscalSM{{\mathscr G}({c_n}^{-1}\cdot) f_n,f_n}\leq {c_n}^{-\frac12}\norm{{\mathscr G}}_{L^2(\K)}\norm{f_n}_{L^4(\K_{c_n})}^2\underset{c\to+\infty}{\longrightarrow}0$$
and, by the same argument as the one to obtain \eqref{D_K_HLS}, we have
$${c_n}^{-2}\pscalSM{|w_{c_n}|^2\star G_\K({c_n}^{-1}\cdot) f_n,f_n}\lesssim {c_n}^{-1}\norm{\breve{w}_{c_n}}_{L^{\frac{12}5}(\K_{c_n})}^2\norm{f_n}_{L^{\frac{12}5}(\K_{c_n})}^2\underset{c\to+\infty}{\longrightarrow}0.$$
Moreover, by \cpr{Main_term_expansion_of_minimum}, $\mathds{1}_{\K_{c_n}}\breve{w}_n(\cdot+c_n R)$ strongly converges in $L^q(\R^3)$ for $2\leq q<6$ hence for $p=\frac23$ and $p=\frac43$ we have
$$\pscalSM{|\breve{w}_{c_n}|^p ,|f_n|^2}_{L^2(\K_{c_n})}=\pscalSM{|\breve{w}_{c_n}(\cdot+c_n R)|^p ,|f_n(\cdot+c_n R)|^2}_{L^2(\K_{c_n})}\to\pscalSM{|Q|^p ,|f|^2}_{L^2(\R^3)}.$$

Finally, by~\ccr{convergence_EulerLagrange_multiplier} and weak convergence in $L^2(\R^3)$ of $\mathds{1}_{\K_{c_n}}f_n(\cdot+c_n R)$,
$$\liminf\limits_{n\to\infty}\frac{\mu_{c_n}}{{c_n}^2}\norm{f_n}^2_{L^2(\K_{c_n})}=\liminf\limits_{n\to\infty}\frac{\mu_{c_n}}{{c_n}^2}\norm{f_n(\cdot+c_n R)}^2_{L^2(\K_{c_n})}\geq\mu\norm{f}^2_{L^2(\R^3)}.$$
This concludes the proof of~\clm{K_complete_model_Lplus_c_elliptic_liminf_ineq}.
\end{proof}
We now prove that $\alpha_n$ cannot tend to zero. Let suppose it does, then there exists a sequence of $f_n \in H^1(\K_{c_n})$ such that $\norm{f_n}_{H^1(\K_{c_n})}=1$, $\pscalSM{\breve{w}_{c_n}+ \breve{\nu}_{c_n}, f_n}_{L^2_{\textrm{per}}(\K_{c_n})}=0$ and $\pscalSM{{\mathscr G}({c_n}^{-1}{\cdot}),\nabla((\breve{w}_{c_n}+\breve{\nu}_{c_n})\breve{f}_n)}_{L^2_{\textrm{per}}(\K_{c_n})}=0$, with $\pscal{L^+_n f_n,f_n}_{L^2(\K_{c_n})}\to0$.

Thus, by the uniform boundedness of $\norm{f_n}_{H^1(\K_{c_n})}$, $\mathds{1}_{\K_{c_n}}f_n$ converges weakly in $L^2(\R^3)\cap L^6(\R^3)$ to a $f$ which verifies $\pscalSM{L^+_\mu f,f}_{L^2(\R^3)}\leq0$, by~\clm{K_complete_model_Lplus_c_elliptic_liminf_ineq}, and $\norm{f}_{H^1(\K_{c_n})}\leq1$. We claim that $f$ also solves the orthogonality properties
$$\pscalSM{f,Q}_{L^2(\R^3)}=0 \quad \textrm{ and } \quad \pscalSM{f,Q\nabla|\cdot|^{-1}}_{L^2(\R^3)}=0.$$
Indeed, on one hand we deduce from the uniqueness of $Q\geq0$ (given by the conjecture), that $\mathds{1}_{\K_{c_n}}(\breve{w}_{c_n}(\cdot+c_n R)+ \breve{\nu}_{c_n}(\cdot+c_n R))\to2Q$ in $L^2(\R^3)\cap L^{6-}(\R^3)$. This, together with \eqref{Weinstein_ortho_key_property_1} and the weak convergence of the $f_n$'s leads to $\pscalSM{f,Q}_{L^2(\R^3)}=0$. On another hand, the uniqueness of $Q$ gives also the $L^2(\R^3)$ strong convergence
$$\mathds{1}_{\K_{c_n}}\nabla(\breve{w}_{c_n}(\cdot+c_n R)+\breve{\nu}_{c_n}(\cdot+c_n R))\to2\nabla Q\in H^1(\R^3).$$
Thus, applying \clm{K_complete_model_convergence_second_G_K_term} on one hand to it and $\mathds{1}_{\K_{c_n}} f_{n}(\cdot+c_n R)\wto f\in H^1(\R^3)$  with the first set of conditions in \clm{K_complete_model_convergence_second_G_K_term} and, on the other hand, to $\mathds{1}_{\K_{c_n}}(\breve{w}_{c_n}(\cdot+c_n R)+ \breve{\nu}_{c_n}(\cdot+c_n R))\to2Q$ and $\mathds{1}_{\K_{c_n}} \nabla f_{n}(\cdot+c_n R)\wto \nabla f\in L^2(\R^3)$ --- which comes from \clm{K_eff_dilated_strongconvergence_results_of_weakly_convergent_sequence} --- with the second set of conditions, we obtain
$$\pscalSM{{\mathscr G}({c_n}^{-1}{\cdot}+R),\nabla[(\breve{w}_{c_n}(\cdot+c_n R)+\breve{\nu}_{c_n}(\cdot+c_n R))\breve{f}_n(\cdot+c_n R)]}_{L^2_{\textrm{per}}(\K_{c_n})}\to2\int_{\R^3}\frac{\nabla(fQ)}{|\cdot|}.$$
Finally, \eqref{Weinstein_ortho_key_property_2} implies that $\pscalSM{f,Q\nabla|\cdot|^{-1}}_{L^2(\R^3)}=-\pscalSM{\nabla(fQ),|\cdot|^{-1}}_{L^2(\R^3)}=0$ and our claim is proved.

As we will prove in \cpr{Weinstein_type_proposition}, if~\ccjt{R3_eff_model_conjecture_uniqueness_and_monotony_M} holds then these two orthogonality properties imply that there exists $\alpha>0$ such that
$$\pscalSM{L^+_\mu f,f}_{L^2(\R^3)}\geq\alpha\norm{f}_{H^1(\R^3)}^2$$
hence $f\equiv0$ due to $\pscalSM{L^+_\mu f,f}_{L^2(\R^3)}\leq0$ obtained previously. Since the terms involving a power of $|w_{c_n}|$ converge and $f\equiv0$, we have
$$o(1)=\pscal{L^+_n f_n,f_n}_{L^2(\K_{c_n})}=\norm{\nabla f_n}_{L^2(\K_{c_n})}^2+\mu\norm{f_n}_{L^2(\K_{c_n})}^2+o(1)$$
hence both norms vanish, since $\mu>0$, which means that $\norm{f_n}_{H^1(\K_{c_n})}\to0$. This contradicts $\norm{f_n}_{H^1(\K_{c_n})}=1$ and concludes the proof that $\alpha_n$ cannot vanish, hence that of~\cpr{K_complete_model_Lplus_c_elliptic}.
\end{proof}

We are left with the proof of~\cpr{Weinstein_type_proposition}.
\begin{prop}\label{Weinstein_type_proposition}
If~\ccjt{R3_eff_model_conjecture_uniqueness_and_monotony_M} holds then there exists $\alpha>0$ such that
\begin{equation}\label{Weinstein_type_proposition_equation}
\pscalSM{L^+_\mu f,f}_{L^2(\R^3)}\geq\alpha\norm{f}_{H^1(\R^3)}^2,
\end{equation}
for all $f\in H^1(\R^3)$ such that $\pscalSM{f,Q}_{L^2(\R^3)}=0$ and $\pscalSM{f,Q\nabla|\cdot|^{-1}}_{L^2(\R^3)}=0$.
\end{prop}
The proof of this proposition uses the celebrated method of Weinstein~\cite{Weinstein-85} and Grillakis--Shatah--Strauss~\cite{GriShaStra-87}. The idea is the following. Using a Perron-Frobenius argument in each spherical harmonics sector as in~\cite{Weinstein-85, Lenzmann-09,LewRot-15}, one obtains that the linearized operator $L^+_\mu$ has only one negative eigenvalue with (unknown) eigenfunction $\phi_0$ in the sector of angular momentum $\ell=0$, and has $0$ as eigenvalue of multiplicity three with corresponding eigenfunctions $\partial_{x_i} Q$. On the orthogonal of these four functions, $L^+_\mu$ is positive definite. In our setting, we have to study $L^+_\mu$ on the orthogonal of $Q$ and the three functions $x_i|x|^{-3}Q(x)$ which are different from the mentioned eigenfunctions. Arguing as in~\cite{Weinstein-85}, we show below that the restriction of $L^+_\mu$ to the angular momentum sector $\ell=1$ is positive definite on the orthogonal of the functions $x_i|x|^{-3}Q(x)$. The argument is general and actually works for functions of the form $\partial_{x_i}(\eta(|x|))Q(x)=x_i|x|^{-1}\eta'(|x|)Q(x)$ where $\eta$ is any non constant monotonic function on $\R$. On the other hand, the argument is more subtle for $Q$ in the angular momentum sector $\ell=0$ and this is where we need~\ccjt{R3_eff_model_conjecture_uniqueness_and_monotony_M}.
\begin{proof}[Proof of~\cpr{Weinstein_type_proposition}]First we note that it is obviously enough to prove it for $f$ real valued but also that it is enough to prove
\begin{equation}\label{Weinstein_type_proposition_equation_L2}
\pscalSM{L^+_\mu f,f}_{L^2(\R^3)}\geq\alpha\norm{f}_{L^2(\R^3)}^2
\end{equation}
with $\alpha>0$. Indeed, if $f$ verifies~\eqref{Weinstein_type_proposition_equation_L2} then, for any $\epsilon>0$, we have
$$\pscalSM{L^+_\mu f,f}_{L^2}\geq \left((1-\epsilon)\alpha+\epsilon\left(\mu - \frac73c_{TF}\norm{Q}_{L^\infty}^{\frac43}-\frac53\norm{Q}_{L^\infty}^{\frac23}\right)\right)\norm{f}_{L^2}^2+\epsilon\norm{\nabla f}_{L^2}^2,$$
hence $f$ verifies~\eqref{Weinstein_type_proposition_equation} too (for a smaller $\alpha>0$).

Since $Q$ is a radial function, the operator $L^+_\mu$ commutes with rotations in $\R^3$ and we will therefore decompose $L^2(\R^3)$ using spherical harmonics: for any $f\in L^2(\R^3)$, 
$$f(x)=\sum\limits_{\ell=0}^\infty\sum\limits_{m=-\ell}^\ell f_{\ell}^{m}(r)Y_{\ell}^{m}(\Omega),$$
where $x=r\Omega$ with $r=|x|$ and $\Omega\in{\mathbb S}^2$. This yields the direct decomposition
$$L^2(\R^3)=\bigoplus\limits_{\ell=0}^\infty {\mathcal H}_{(\ell)}$$
and $L^+_\mu$ maps into itself each
$${\mathcal H}_{(\ell)}:=L^2(\R_+, r^2\dd r)\otimes\vect\{Y_{\ell}^{m}\}_{m=-\ell}^\ell.$$
Using the well-known expression of $-\Delta$ on ${\mathcal H}_{(\ell)}$, we obtain that
$$L^+_\mu=\bigoplus\limits_{\ell=0}^\infty L^+_{\mu,\ell}$$
where the $L^+_{\mu,\ell}$'s are operators acting on $L^2(\R_+, r^2\dd r)$ given by
$$L^+_{\mu,\ell}=-\frac{\dd^2}{\dd r^2}-\frac2r\frac{\dd}{\dd r} +\frac{\ell(\ell+1)}{r^2}+\frac73 c_{TF} |Q_\mu|^{\frac43}-\frac53 |Q_\mu|^{\frac23}+\mu.$$

We thus prove inequality \eqref{Weinstein_type_proposition_equation_L2} by showing that there exists $\alpha>0$ such that for each $\ell$ the inequality holds for any $f\in {\mathcal H}_{(\ell)}\cap H^1(\R^3)$ verifying $\pscalSM{f,Q}=0$ and $\pscalSM{f,Q\nabla|\cdot|^{-1}}_{L^2(\R^3)}=0$.

Arguing as in~\cite{Lenzmann-09}, we have first the following result.
\begin{lemme}[Perron--Frobenius property of the $L^+_{\mu,\ell}$]\label{Perron_Frobenius_property}
Each $L^+_{\mu,\ell}$ has the Perron--Frobenius property: its lowest eigenvalue $e_{\mu,\ell}$ is simple and the corresponding eigenfunction $\phi_\ell(r)$ is positive.
\end{lemme}

\medskip
\noindent\textbf{Proof for the sector $\ell=1$.} We start with the case $\ell=1$ and prove that
\begin{equation}\label{Weinstein_type_proposition_definition_alpha1}
\alpha_1:=\inf\limits_{\substack{f\in{\mathcal H}_{(1)}\cap H^1(\R^3)\\ \pscalSM{f,Q\nabla|\cdot|^{-1}}_{L^2(\R^3)}=0}}\frac{\pscalSM{L^+_\mu f,f}_{L^2(\R^3)}}{\norm{f}_{L^2(\R^3)}^2}>0.
\end{equation}
Since $Q$ is radial, we have for $i=1,2,3$, that
$$\partial_{x_i}Q(x)=Q'(r)\frac{x_i}r\in {\mathcal H}_{(1)}.$$
Moreover, by the non-degeneracy result of~\cth{R3_eff_model_existence_and_nondeg}, we know that $\partial_{x_i}Q$ is an eigenfunction of $L^+_\mu$ associated with the eigenvalue $0$ hence $Q'(r)$ is an eigenfunction of $L^+_{\mu,1}$ associated with the eigenvalue $e_{\mu,1}=0$. Therefore, the fact that $Q'(r)<0$ (as proved in~\cth{R3_eff_model_existence_thm}) implies, using the Perron-Frobenius property verified by $L^+_{\mu,1}$, that $e_{\mu,1}=0$ is the lowest eigenvalue of $L^+_{\mu,1}$ and is simple with $-Q'>0$ the associated eigenfunction. Consequently, we have for any $f\in{\mathcal H}_{(1)}$ that
$$\pscalSM{L^+_\mu f,f}_{L^2(\R^3)}=\sum\limits_{m=-1}^1 \pscalSM{L^+_{\mu,1} f^{m}(r),f^{m}(r)}_{L^2(\R_+, r^2\dd r)}\geq0$$
and in particular that $\alpha_1\geq0$.

We thus suppose that $\alpha_1=0$ and prove it is impossible. Let $f_n$ be a minimizing sequence to~\eqref{Weinstein_type_proposition_definition_alpha1} with $\norm{f_n}_{L^2(\R^3)}=1$. One has
$$\norm{\nabla f_n}_{L^2(\R^3)}^2\leq\pscalSM{L^+_\mu f_n,f_n}_{L^2(\R^3)}+\frac53\norm{Q}_{L^\infty(\R^3)}^{\frac23}$$
and consequently the sequence $f_n$ is bounded in $H^1(\R^3)$. We denote by $f$ its weak limit in $H^1(\R^3)$, up to a extraction of a subsequence, which is in ${\mathcal H}_{(1)}$. We have
$$0\leq\pscalSM{L^+_\mu f,f}_{L^2(\R^3)}\leq\liminf\pscalSM{L^+_\mu f_n,f_n}_{L^2(\R^3)}=\alpha_1=0,$$
where the second inequality is due to
$$\liminf \norm{\nabla f_n}_{L^2(\R^3)}^2\geq \norm{\nabla f}_{L^2(\R^3)}^2, \qquad \liminf \norm{f_n}_{L^2(\R^3)}^2\geq \norm{f}_{L^2(\R^3)}^2,$$
$\mu>0$ and to $\pscalSM{|Q|^p f_n,f_n}_{L^2(\R^3)}\to\pscalSM{|Q|^p f,f}_{L^2(\R^3)}$, for $p=\frac23$ and $p=\frac43$, obtained by a similar argument to the one in proof of~\clm{K_complete_model_Lplus_c_elliptic_liminf_ineq}. It implies that
$$\pscalSM{L^+_\mu f,f}_{L^2(\R^3)}=0$$
hence, $f=\sum_{i=1}^3 c_i\partial_{x_i}Q$ by the Perron-Frobenius property and since $\{\frac{x_1}r,\frac{x_2}r,\frac{x_3}r\}$ is an orthogonal basis of $\vect\{Y_{1}^{-1},Y_{1}^{0},Y_{1}^{1}\}$. However, since $\pscalSM{f_n,Q\nabla|\cdot|^{-1}}_{L^2(\R^3)}=0$, we have for any $i=1,2,3$ after passing to the weak limit that
$$\int_{\R^3}{\frac{x_i}{|x|^{3}}f(x)Q(x)\dd x}=0.$$
We then remark that, since $Q$ is radial, we have
$$\int_{\R^3}{\frac{x_i}{|x|^{3}}Q(x)\partial_{x_j}Q(x)\dd x}=\int_{\R^3}{\frac{x_jx_i}{|x|^4}Q(x)Q'(x) \dd x}=0, \qquad \forall i\neq j.$$
This gives, for $i=1,2,3$, that
$$0=\int_{\R^3}{\frac{x_i}{|x|^{3}}f(x)Q(x)\dd x}=c_i\int_{\R^3}{\frac{{x_i}^2}{|x|^4}Q(x)Q'(x) \dd x}$$
but $Q>0$ and $Q'<0$, hence $c_i=0$ thus $f\equiv0$. We thus have obtained, if $\alpha_1=0$, that any minimizing sequence $f_n$ to~\eqref{Weinstein_type_proposition_definition_alpha1} converges weakly to $0$ in $H^1(\R^3)$. This gives $\pscalSM{|Q|^p f_n,f_n}_{L^2(\R^3)}\to0$ and
$$\norm{\nabla f_n}_{L^2(\R^3)}^2+\mu\norm{f_n}_{L^2(\R^3)}^2=\pscalSM{L^+_\mu f_n,f_n}_{L^2(\R^3)}+o(1)\to\alpha_1=0$$
therefore $f_n\to0$ strongly in $H^1(\R^3)$, because $\mu>0$, which contradicts the fact that $\norm{f_n}_{L^2(\R^3)}=1$. We have thus proved that $\alpha_1>0$.

\medskip
\noindent\textbf{Proof for the sector $\ell\geq2$.} We now deal with the cases $\ell\geq2$ and prove that there exists $\alpha>0$, independent of $\ell$, such that
\begin{equation}\label{Weinstein_type_proposition_ineq_to_prove_ell_more_2}
\pscalSM{L^+_{\mu,\ell} \phi,\phi}_{L^2(\R_+,r^2\dd r)}\geq\alpha\norm{\phi}_{L^2(\R_+,r^2\dd r)}^2
\end{equation}
for any $\phi\in L^2(\R_+,r^2\dd r)$. Since for such $\phi$ we have
\begin{equation}\label{Comparison_L_plus_ell_s}
\pscalSM{L^+_{\mu,\ell} \phi,\phi}_{L^2(\R_+,r^2\dd r)}=\pscalSM{L^+_{\mu,\ell-1} \phi,\phi}_{L^2(\R_+,r^2\dd r)}+2\ell\norm{\phi/r}_{L^2(\R_+,r^2\dd r)}^2,
\end{equation}
it is then sufficient to prove~\eqref{Weinstein_type_proposition_ineq_to_prove_ell_more_2} in the case $\ell=2$ in order to prove it for all $\ell\geq2$.

For $\ell=2$, we can assume that $\inf\sigma(L^+_{\mu,2})$ is attained because, otherwise,
$$V:=\frac73 c_{TF} |Q_\mu|^{\frac43}-\frac53 |Q_\mu|^{\frac23}$$
being bounded and vanishing as $r\to\infty$, it is well-known that $\sigma(L^+_{\mu,2})=\sigma_{\textrm{ess}}(L^+_{\mu,2})=[\mu;+\infty)$ and~\eqref{Weinstein_type_proposition_ineq_to_prove_ell_more_2} follows. We thus have, by~\eqref{Comparison_L_plus_ell_s} and $L^+_{\mu,1}\geq0$, that the eigenvalue $e_{\mu,2}=\inf\sigma(L^+_{\mu,2})$ and its associated eigenfunction $\phi_2\nequiv0$ verify that
$$e_{\mu,2}=\inf\sigma(L^+_{\mu,2})\geq 2\frac{\norm{\phi_2/r}_{L^2(\R_+,r^2\dd r)}^2}{\norm{\phi_2}_{L^2(\R_+,r^2\dd r)}^2}>0$$
and~\eqref{Weinstein_type_proposition_ineq_to_prove_ell_more_2} is therefore proved. It concludes the case $\ell\geq2$.

\medskip
\noindent\textbf{Proof for the sector $\ell=0$.} We conclude with the case $\ell=0$ and prove that for any $f\in{\mathcal H}_{(0)}$, we have
\begin{equation}\label{Weinstein_type_proposition_definition_alpha0}
\alpha_0:=\inf\limits_{\substack{f\in{\mathcal H}_{(0)}\cap H^1(\R^3)\\ \pscalSM{f,Q}_{L^2(\R^3)}=0}}\frac{\pscalSM{L^+_\mu f,f}_{L^2(\R^3)}}{\norm{f}_{L^2(\R^3)}^2}>0.
\end{equation}
We already know that $\alpha_0\geq0$ because $Q$ is a minimizer. Indeed, for $f\in H^1(\R^3)$ such that $\pscalSM{f,Q}_{L^2(\R^3)}=0$, through a computation similar to~\eqref{K_complete_model_expansion_around_minimizer_of_energy} and using~\eqref{R3_eff_model_EulerLagrange}, \eqref{R3_eff_model_EulerLagrange_mu_formulae}, \eqref{Taylor_expansion_power_terms} and that $Q$ is a minimizer of $J_{\R^3}(\lambda)$, we obtain
\begin{multline*}
{\mathscr J}_{\R^3}(Q)\leq{\mathscr J}_{\R^3}\left(\frac{Q+\epsilon f}{\norm{Q+\epsilon f}_2}\norm{Q}_2\right)\\
	={\mathscr J}_{\R^3}(Q)+\epsilon^2(\pscalSM{L^+_\mu \Re{f},\Re{f}}_{L^2(\R^3)}+\pscalSM{L^-_\mu \Im{f},\Im{f}}_{L^2(\R^3)})+o(\epsilon^2)
\end{multline*}
which implies in particular that $\pscalSM{L^+_\mu f,f}_{L^2(\R^3)}\geq0$ for as soon as $\pscalSM{f,Q}_{L^2(\R^3)}=0$.

We thus suppose $\alpha_0=0$ and prove it is impossible. Let $f_n$ be a minimizing sequence to~\eqref{Weinstein_type_proposition_definition_alpha0} with $\norm{f_n}_{L^2(\R^3)}=1$. As in the proof of case $\ell=1$ above, $f_n$ is in fact bounded in $H^1(\R^3)$ and denoting by $f\in {\mathcal H}_{(0)}$ its weak limit in $H^1(\R^3)$, up to a subsequence, we have $\pscalSM{L^+_\mu f,f}_{L^2(\R^3)}=0$. This leads, to $L^+_\mu f=\beta Q$ thus, using that $L^+_\mu$ is inversible, to $f=\beta(L^+_\mu)^{-1}Q$. Consequently,
$$0=\pscalSM{f,Q}_{L^2(\R^3)}=\beta\pscalSM{Q,(L^+_\mu)^{-1}Q}_{L^2(\R^3)}$$
hence $\beta=0$ since $\pscalSM{Q,(L^+_\mu)^{-1}Q}_{L^2(\R^3)}<0$ by~\ccjt{R3_eff_model_conjecture_uniqueness_and_monotony_M}. We have obtained $f\equiv0$ which is absurd as before.
\end{proof}
This concludes the proof of~\cth{R3_eff_model_consequence_conjecture_uniqueness_and_monotony_M}.\qed


\end{document}